\documentclass{lmcs}

\usepackage{amsmath, amssymb, latexsym}
\usepackage{tikz}
\usetikzlibrary{matrix}
\usetikzlibrary{automata}
\usetikzlibrary{arrows,decorations.markings}
\usetikzlibrary{patterns,snakes}
\usepackage[all]{xypic}

\input{diagxy}

\usepackage{enumerate}
\usepackage{hyperref}

\theoremstyle{plain}\newtheorem{claim}[thm]{Claim}

\def\Bbox{
{\unskip\nobreak\hfil\penalty50
\hskip1em\hbox{}\nobreak\hfil{\lower .5pt \hbox{$\Box$}}
\parfillskip=0pt \finalhyphendemerits=0 \par}
}

\def\eop{
\ifmmode {\hbox{\Bbox}} \else \Bbox \fi
}

\def\bbox{
\ifmmode {\hbox{\bbox}} \else \Bbox \fi
}

\newcommand{\llangle}{{\langle \kern -.2em \langle}}
\newcommand{\rrangle}{{\rangle \kern -.2em \rangle}}

\newcommand{\Sta}{\mathsf{ST}(A)}

\newcommand{\A}{\mathcal {A}}

\newcommand{\N}{\mathbb{N}}
\renewcommand{\L}{\mathcal {L}}

\renewcommand{\to}{\rightarrow}
\renewcommand{\S}{S \langle \kern -.2em \langle Z^* \rangle \kern -.2em \rangle}

\newcommand{\ex}{\mathsf{ex}}
\newcommand{\Path}{\mathsf{Path}}
\newcommand{\Wedge}{\bigwedge}
\newcommand{\Colim}{\mathsf{Colim}}

\newcommand{\Tree}{\mathsf{Tree}}
\newcommand{\Graph}{\mathsf{Graph}}
\newcommand{\tand}{\;\text{and}\;}
\newcommand{\First}{\mathrm{First}}
\newcommand{\era}[1]{\overset{#1}{\longrightarrow}}

\newcommand{\erb}[2]{\overset{#1}{\underset{#2}{\longrightarrow}}}

\newcommand{\Unf}{\mathrm{Unf}}
\newcommand{\trans}{\mathcal{T}}

\renewcommand{\textbf}{\emph}

\newcommand{\Ima}{\mathrm{Im}}
\newcommand{\Dom}{\mathrm{Dom}}

\begin{document}

\title{Algebraic Synchronization Trees and Processes}

\author[L. Aceto]{Luca Aceto}	
\address{ICE-TCS, School of Computer Science,
 Reykjavik University, Iceland}	
\thanks{Luca Aceto and Anna Ing\'olfsd\'ottir have been partially supported by
the project `Meta-theory of Algebraic Process Theories'
(nr.~100014021) of the Icelandic Research Fund. The work on the paper
was partly carried out while Luca Aceto and Anna Ing\'olfsd\'ottir
held Abel Extraordinary Chairs at Universidad Complutense de Madrid,
Spain, supported by the NILS Mobility Project. Arnaud Carayol has been
supported by the project AMIS (ANR 2010 JCJC 0203 01 AMIS). Zolt\'an
\'Esik's work on this paper was partly supported by grant T10003
from Reykjavik University's Development Fund and by the Labex B{\'e}zout part
of the program \emph{Investissements d'Avenir} (ANR-10-LABX-58).}	

\author[A. Carayol]{Arnaud Carayol}	
\address{Universit\'e Paris-Est, LIGM (UMR 8049), CNRS, ENPC, ESIEE, UPEM, France}	

\author[Z. \'Esik]{Zolt\'an \'Esik}	
\address{Institute of Informatics, University of Szeged,  Hungary}	

\author[A. Ing\'olfsd\'ottir]{Anna Ing\'olfsd\'ottir}	
\address{ICE-TCS, School of Computer Science,
 Reykjavik University, Iceland}	
\thanks{}	



\titlecomment{An extended abstract of this article was published in the proceedings of ICALP 2012.}

%
%

\begin{abstract}
We study algebraic synchronization trees, i.e., initial solutions of
algebraic recursion schemes over the continuous categorical algebra of
synchronization trees. In particular, we investigate the relative
expressive power of algebraic recursion schemes over two signatures,
which are based on those for Basic CCS and Basic Process Algebra, as a
means for defining synchronization trees up to isomorphism as well as
modulo bisimilarity and language equivalence. The expressiveness of
algebraic recursion schemes is also compared to that of the low
levels in Caucal's pushdown hierarchy.
\end{abstract}
\maketitle

\section{Introduction}\label{Sect:intro}

The study of recursive program schemes is one of the classic topics in
programming language semantics. (See,
e.g.,~\cite{deBakker,Courcelle83,Gue81,Nivat75,Scott71} for some of the early
references.) One of the main goals of this line of research is to
define the semantics of systems of recursive equations such as 
\begin{equation}\label{algschema}
F(n) = \text{ifzero}(n,1,\text{mult}(2,F(\text{pred}(n)))) . 
\end{equation}
In the above recursion scheme, the symbols ifzero, add, pred, $1$ and
$2$ denote given function symbols; these are used to define the
derived unary function $F(n)$, which we will refer to as a functor
variable. Interpreting ifzero as the function over the set of triples of $\mathbb{N}^3$
that returns its second argument when the first is zero and the third
otherwise, mult as multiplication and pred as the predecessor
function, intuitively one would expect the above recursion scheme to
describe the function over the natural numbers that, given an input
$n$, returns $2^n$. This expectation has been formalized in several
ways in the literature on recursive program schemes. A classic answer
can be summarized by the `motto' of the initial-algebra-semantics
approach: `The semantics of a recursive program scheme is the infinite term 
tree (or ranked tree) that is the least fixed point of the system of equations
associated with the program scheme.'

In the light of the role that infinite term trees play in defining the
semantics of recursive program schemes, it is not surprising that the
study of infinite term trees has received a lot of attention in the
research literature. Here we limit ourselves to mentioning Courcelle's
classic survey paper~\cite{Courcelle83}, which presents results on
topological and order-theoretic properties of infinite trees, notions
of substitutions for trees as well as regular and algebraic term 
trees. {\em Algebraic term trees} are those that arise as solutions of
recursive program schemes that, like (\ref{algschema}), are `first
order'. On the other hand, {\em regular term trees} are the solutions of
systems of equations like
\begin{eqnarray*} 
X & = & f(X,Y) \\
Y & = & a , 
\end{eqnarray*}
which define parameterless functions $X$ and $Y$. Regular term trees
arise naturally as the unfoldings of flowcharts, whereas algebraic
term trees stem from the unfoldings of recursion schemes that
correspond to functional programs~\cite{Courcelle83}.

In this paper, we study recursion schemes from a process-algebraic
perspective and investigate the expressive power of algebraic
recursion schemes over the signatures of Basic CCS~\cite{Mil89CC} and
of Basic Process Algebra (BPA)~\cite{BaetenBR2009} as a way of
defining possibly infinite synchronization trees~\cite{Milner}. Both
these signatures allow one to describe every finite synchronization
tree and include a binary choice operator +. The difference between
them is that the signature for Basic CCS, which is denoted by $\Gamma$
in this paper, contains a unary action prefixing operation $a.\_$ for
each action $a$, whereas the signature for BPA, which we denote by
$\Delta$, has one constant $a$ for each action that may label the edge
of a synchronization tree and offers a full-blown sequential
composition, or sequential product, operator. Intuitively, the
sequential product $t \cdot t'$ of two synchronization trees is
obtained by appending a copy of $t'$ to the leaves of $t$ that
describe successful termination of a computation. In order to
distinguish successful and unsuccessful termination, both the
signatures $\Gamma$ and $\Delta$ contain constants $0$ and $1$, which
denote unsuccessful and successful termination, respectively.

As an example of a regular recursion scheme over the signature
$\Delta$, consider
\begin{equation*}
X = (X \cdot a ) + a . 
\end{equation*}
On the other hand, the following recursion scheme is
$\Gamma$-algebraic, but not $\Gamma$-regular:
\begin{eqnarray*}
F_1 & = & F_2(a.1) \\
F_2(v)  & = & v + F_2(a.v) . 
\end{eqnarray*}
It turns out that both these recursion schemes define the infinitely
branching synchronization tree depicted on Figure~\ref{Fig:anyfinitedepth}.

In the setting of process algebras such as CCS~\cite{Mil89CC} and
ACP~\cite{BaetenBR2009}, synchronization trees are a classic model of
process behaviour. They arise as unfoldings of labelled transition
systems that describe the operational semantics of process terms and
have been used to give denotational semantics to process description
languages---see, for instance,~\cite{Abramsky91}. Regular
synchronization trees over the signature $\Gamma$ are unfoldings of
processes that can be described in the regular fragment of CCS, which
is obtained by adding to the signature $\Gamma$ a facility for the
recursive definition of processes. On the other hand, regular
synchronization trees over the signature $\Delta$ are unfoldings of
processes that can be described in Basic Process Algebra
(BPA)~\cite{BaetenBR2009} augmented with constants for the deadlocked
and the empty process as well as recursive definitions.

As is well known, the collection of regular synchronization trees over
the signature $\Delta$ strictly includes that of regular
synchronization trees over the signature $\Gamma$ even up to language
equivalence. Therefore, the notion of regularity depends on the
signature.  But what is the expressive power of algebraic recursion
schemes over the signatures $\Gamma$ and $\Delta$? The aim of this
paper is to begin the analysis of the expressive power of those
recursion schemes as a means for defining synchronization trees, and
their bisimulation or language equivalence classes.

In order to characterize the expressive power of algebraic recursion
schemes defining synchronization trees, we interpret such schemes in
continuous categorical $\Gamma$- and $\Delta$-algebras of
synchronization trees. Continuous categorical $\Sigma$-algebras are a
categorical generalization of the classic notion of continuous
$\Sigma$-algebra that underlies the work on algebraic
semantics~\cite{BloomTWW83,CourcelleN76,GoguenTWW77,Gue81}, and have
been used in \cite{BEreg,BEmezei,Esiksynch} to give semantics to
recursion schemes over synchronization trees and words. (We refer the
interested reader to~\cite{MiliusM2006} for a recent discussion of
category-theoretic approaches to the solution of recursion schemes.)
In this setting, the $\Gamma$-regular (respectively,
$\Gamma$-algebraic) synchronization trees are those that are initial
solutions of regular (respectively, algebraic) recursion schemes over
the signature $\Gamma$. $\Delta$-regular and $\Delta$-algebraic
synchronization trees are defined in similar fashion.
 
 Our first
contribution in the paper is therefore to provide a
 categorical
semantics for first-order recursion schemes that define
 processes,
whose behaviour is represented by synchronization
 trees. The use of
continuous categorical $\Sigma$-algebras allows us
 to deal with
arbitrary first-order recursion schemes; there is no need
 to
restrict oneself to, say, `guarded' recursion schemes, as one is
forced to do when using a metric semantics (see, for
instance,~\cite{vBreugel2001} for a tutorial introduction to metric
semantics), and this categorical approach to giving semantics to
first-order recursion schemes can be applied even when the
order-theoretic framework either fails because of the lack of a
`natural' order or leads to undesirable identities.
 
 As a second
contribution, we provide a comparison of the expressive
 power of
regular and algebraic recursion schemes over the signatures
 $\Gamma$
and $\Delta$, as a formalism for defining processes described
 by
their associated synchronization trees. We show that each
$\Delta$-regular tree is $\Gamma$-algebraic
(Theorem~\ref{Thm:Delta2Gamma}) by providing an algorithm for
transforming a $\Delta$-regular recursion scheme into an equivalent
$\Gamma$-algebraic one that involves only unary functor variables. In
addition, we prove that every synchronization tree that is defined by
a $\Gamma$-algebraic recursion scheme of a certain form that involves
only unary functor variables can be transformed into an equivalent
$\Delta$-regular recursion scheme.
 
 We provide examples of
$\Gamma$-algebraic synchronization trees that
 are not
$\Delta$-regular and of $\Delta$-algebraic trees that are not
$\Gamma$-algebraic, not even up to bisimulation equivalence. In
particular, in Proposition~\ref{Prop:bagnotalg} we prove that the
synchronization tree associated with the bag over a binary alphabet
(which is depicted on Figure~\ref{fig:bag} on page~\pageref{fig:bag})
is not $\Gamma$-algebraic,
 even up to language equivalence, and that
it
 is not $\Delta$-algebraic up to bisimilarity. These results are
a
 strengthening of a classic theorem from the literature on process
algebra proved by Bergstra and Klop in~\cite{BergstraK84}.

Since each $\Gamma$-algebraic synchronization tree is also
$\Delta$-algebraic, we obtain the following strict expressiveness
hierarchy, which holds up bisimilarity~\cite{Mil89CC,Par81}:  
\begin{center}
$\Gamma$-regular $\subset$ $\Delta$-regular $\subset$ $\Gamma$-algebraic 
$\subset$ $\Delta$-algebraic. 
\end{center}
In the setting of language equivalence, the notion of $\Gamma$-regularity
corresponds to the regular languages, the one of $\Delta$-regularity
or $\Gamma$-algebraicity 
corresponds to the context-free languages, and $\Delta$-algebraicity
corresponds to the macro languages~\cite{Fischer}, which coincide with 
the languages generated by Aho's indexed grammars~\cite{Aho68}.

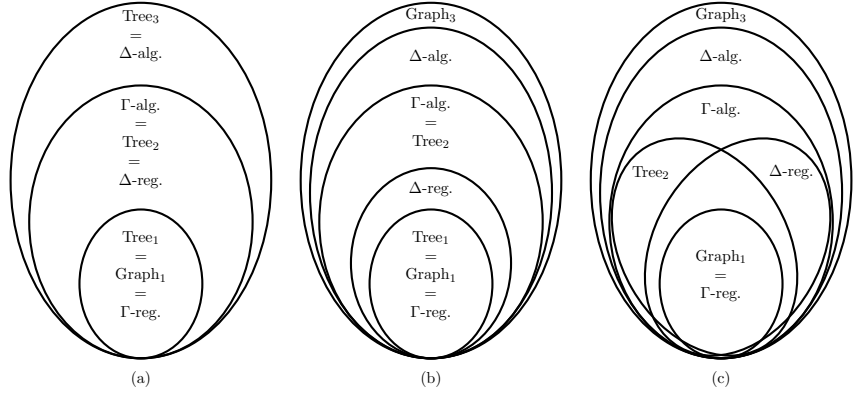
\begin{figure}
\begin{center}
\begin{tikzpicture}[thick,scale=0.55,transform shape]
\draw (0,-0.2) ellipse (1.65*0.90 and 1.8);
\draw (0,1.3) ellipse (3*0.90 and 3.3); 
\draw (0,2.3) ellipse (3.5*0.90 and 4.3); 

\node[text width=1.4cm, text centered] at (0,3.2) {$\Gamma$-alg. $\mbox{}\;=\;\;$ $\mathrm{Tree}_2$ $\;=\;\;$ $\Delta$-reg.};

\node at (0,5.8)[text width=1.4cm, text centered] {$\mathrm{Tree}_{3}$
$\mbox{}=\quad$ $\Delta$-alg.};
\node[text width=1.4cm, text centered] at (0,0) {$\mathrm{Tree}_{1}$ $\quad$ $=$
 $\mathrm{Graph_{1}}$ $\quad$ $\mbox{}\quad=$ $\quad$ $\Gamma$-reg.};
\node at (0,-2.5) {(a)};
\end{tikzpicture}
$\;$
\begin{tikzpicture}[thick,scale=0.55,transform shape]
\draw (0,-0.2) ellipse (1.65*0.90 and 1.8);
\draw (0,2) ellipse (3.25*0.90 and 4); 
\draw (0,1.3) ellipse (3*0.90 and 3.3); 
\draw (0,2.3) ellipse (3.5*0.90 and 4.3); 
\draw (0,0.3) ellipse (2.15*0.90 and 2.3); 

\node at (0,2.1) {$\Delta$-reg.};
\node[text width=1.4cm, text centered] at (0,3.7) {$\Gamma$-alg. $\mbox{}\;=\;\;$ $\mathrm{Tree}_2$};
\node at (0,5.3) {$\Delta$-alg.};
\node at (0,6.3) {$\mathrm{Graph}_{3}$};
\node[text width=1.4cm, text centered] at (0,0) {$\mathrm{Tree}_{1}$ $\quad$ $=$
 $\mathrm{Graph_{1}}$ $\quad$ $\mbox{}\quad=$ $\quad$ $\Gamma$-reg.};
\node at (0,-2.5) {(b)};
\end{tikzpicture}
$\;$
\begin{tikzpicture}[thick,scale=0.55,transform shape]
\draw (0,-0.2) ellipse (1.65*0.90 and 1.8);
\draw (0,2) ellipse (3.25*0.90 and 4); 
\draw (0,1.3) ellipse (3*0.90 and 3.3); 
\draw (0,2.3) ellipse (3.5*0.90 and 4.3); 
\draw (0,0.8)[rotate=30] ellipse (2.15*0.94 and 2.8); 
\draw (0,0.8)[rotate=-30] ellipse (2.15*0.94 and 2.8); 

\node at (-1.7,2.5) {$\mathrm{Tree}_{2}$};
\node at (1.7,2.5) {$\Delta$-reg.};
\node at (0,4) {$\Gamma$-alg.};
\node at (0,5.3) {$\Delta$-alg.};
\node at (0,6.3) {$\mathrm{Graph}_{3}$};
\node[text width=1.4cm, text centered] at (0,0) {$\mathrm{Graph_{1}}$ $\quad$ $\mbox{}\quad=$ $\quad$ $\Gamma$-reg.};
\node at (0,-2.5) {(c)};
\end{tikzpicture}
\end{center}
\caption{\label{Fig:hierarchy}The expressiveness hierarchies up to language equivalence (a), up to bisimilarity (b) and up to isomorphism (c)}
\end{figure}

In order to obtain a deeper understanding of $\Gamma$-algebraic
recursion schemes, we
characterize their expressive power by following the lead of
Courcelle~\cite{Courcelle78a,Courcelle78b,Courcelle83}. In those
references, Courcelle proved that a term  tree
is algebraic if, and only if, its branch language is a deterministic
context-free language. In our setting, we associate with each
synchronization tree with bounded branching a family of branch
languages and we show that a synchronization tree with bounded
branching is $\Gamma$-algebraic if, and only if, the family of branch
languages associated with it contains a deterministic context-free
language (Theorem~\ref{thm-DCFL}). In conjunction with standard tools
from formal language theory, this result can be used to show that
certain synchronization trees are not $\Gamma$-algebraic.

  As a final main contribution of the paper, we compare the expressiveness of those recursion schemes to that of the low levels in Caucal's hierarchy. This hierarchy is already known to include the term trees defined by safe\footnote{Safety is a syntactic restriction which is trivial for first order schemes. In particular, it does not play any role in our setting.}
 higher-order recursion schemes when interpreted over the free continuous algebra \cite{Caucal}. Unsurprisingly, we show that the classes of $\Gamma$-algebraic and $\Delta$-algebraic synchronization trees belong to the 
  third level of the Caucal hierarchy. We provide a more detailed comparison, which is summarized in Figure~\ref{Fig:hierarchy}. 

As a benefit of the comparison with the Caucal hierarchy, we obtain structural properties and decidability of the monadic second-order theories of $\Delta$-algebraic synchronization trees~\cite{Thomas03}. By contraposition, this implies that a synchronization tree with an undecidable monadic second-order theory cannot be $\Delta$-algebraic. This allows us to show that $\Gamma$-algebraic (and hence $\Delta$-algebraic trees) are not closed under minimization with respect
to bisimulation equivalence (Proposition~\ref{Prop:QuotientnotCaucal}).

The technical developments in this paper make use of techniques and
results from a variety of areas of theoretical computer science. We
employ elementary tools from category theory and initial-algebra
semantics to define the meaning of recursion schemes over algebras of
synchronization trees. Tools from concurrency and formal-language
theory are used to obtain separation results between the different
classes of synchronization trees we consider in this paper. The proof
of Theorem~\ref{thm-DCFL}, characterizing the expressive power of
$\Gamma$-algebraic recursion schemes in the style of Courcelle, uses a
Mezei-Wright theorem for categorical algebras~\cite{BEmezei} as well
as tools from monadic second-order
logic~\cite{Braudthesis,CarayolWohrle}. Overall, we find it pleasing
that methods and results developed by different communities within theoretical
computer science play a role in the study of natural structures like
the synchronization trees that arise from the solution of algebraic
recursion schemes.

The paper is organized as follows. In Section~\ref{Sect:definitions}, we
introduce (first-order) recursion schemes and the synchronization trees they define using continuous categorical algebras. In Section~\ref{Sect:comparison}, we compare  the classes of $\Gamma$-regular, $\Delta$-regular, $\Gamma$-algebraic and $\Delta$-algebraic synchronization trees up to isomorphism, bisimulation and language equivalence.  In Section~\ref{Sect:caucal}, we compare
our classes of synchronization trees to the first levels of the Caucal hierarchy. In Section~\ref{Sect:branchlang}, following Courcelle, we characterize the expressive power of $\Gamma$-algebraic recursion
schemes by studying the branch languages of synchronization trees
whose vertices have bounded outdegree. In Section~\ref{Sect:treesandlogic}, we show how the previous results can be used to prove that 
some synchronization trees cannot be defined using algebraic recursion schemes. Section~\ref{Sect:concl} concludes the paper and lists topics for future research.

\section{Categorical semantics of first-order recursion schemes}
\label{Sect:definitions}

In this section, we provide a categorical semantics for first-order recursion schemes using continuous categorical algebras, and introduce several notions and results that will be used throughout the paper. Section~\ref{Sect:contalg} introduces the necessary categorical background. The continuous categorical 
algebras associated with synchronization trees are introduced in Section~\ref{Sect:syntrees}.  Sections~\ref{Sect:edge-labelled-graphs}--\ref{Sect:graph-transformations} present edge-labelled graphs, monadic second order logic over such structures and the kinds of graph transformations used in the paper. 
The definition of the synchronization tree defined by a recursion scheme is given in Section~\ref{sec-algebraic objects}. 
Sections~\ref{Sect:coa} and \ref{sec:morphism} describe the essentially unique morphisms from the initial continuous ordered  $\Gamma$-algebra (resp. $\Delta$-algebra) of $\Gamma$-term trees (resp. $\Delta$-term trees) to the continuous $\Gamma$-algebra (resp. $\Delta$-algebra) of synchronization trees. These morphisms play an important
role when comparing the synchronization trees defined by recursion schemes to the low levels of the Caucal hierarchy in Section~\ref{Sect:caucal}. We conclude by giving some basic properties of synchronization trees defined by first-order recursion schemes in Section~\ref{Sect:closure}.

Throughout the paper, when $n$ is a non-negative integer, we denote the set 
$\{1,\ldots,n\}$ by $[n]$.
\newcommand{\cca}{C}
\newcommand{\ccb}{D}
\subsection{Continuous categorical algebras}\label{Sect:contalg}

In this section, we recall the notion of continuous categorical 
$\Sigma$-algebra. These structures were used in \cite{BEreg,BEmezei,Esiksynch} 
to give semantics to recursion schemes over synchronization trees and words.  

Let  $\Sigma=\bigcup_{n \geq 0} \Sigma_n$ 
be a ranked set (or `signature'). A \textbf{categorical $\Sigma$-algebra}
is a small category $\cca$ equipped with a functor
$\sigma^{\cca} : \cca^n \to \cca$ for each $\sigma \in \Sigma_n$, $n \geq 0$. 
A \textbf{morphism} between categorical $\Sigma$-algebras $\cca$ and $\ccb$ 
is a functor $h: \cca \to \ccb$
such that for each $\sigma \in \Sigma_n$, the diagram
\begin{eqnarray*}
\bfig
\square<1000,600>[{\cca}^n`\cca`{\ccb}^n`\ccb;\sigma^{\cca}`h^n`h`\sigma^{\ccb}]
\efig
\end{eqnarray*}
commutes \textit{up to a natural isomorphism}.
Here, the functor $h^n:{\cca}^n \rightarrow  {\ccb}^n$ maps each 
object and morphism $(x_1,\ldots,x_n)$ in $\cca^n$ to $(h(x_1),\ldots,h(x_n))$
in $\ccb^n$.

Suppose that $\cca$ is a categorical $\Sigma$-algebra. 
We call $\cca$ \textbf{continuous}  if $\cca$ has a distinguished 
initial object (denoted $\bot^{\cca}$ or $0^\cca$) and colimits of all 
$\omega$-diagrams 
$(f_k: a_k \rightarrow  a_{k+1})_{k \geq 0}$. Moreover, each
functor $\sigma^{\cca}$ is continuous, i.e., 
preserves colimits of $\omega$-diagrams. 
Thus, if $\sigma  \in \Sigma_2$,
say, and if
\begin{eqnarray*}
  x_0 \stackrel{f_0}{\to} x_1 \stackrel{f_1}{\to} x_2 \stackrel{f_2}{\to} \ldots \\
  y_0 \stackrel{g_0}{\to} y_1 \stackrel{g_1}{\to} y_2 \stackrel{g_2}{\to} \ldots 
\end{eqnarray*}
are $\omega $-diagrams in $\cca$ with colimits $(x_k \stackrel{\phi_k}{\to} x)_k$
and $(y_k \stackrel{\psi_k}{\to} y)_k$,  respectively, then
\begin{eqnarray*}
  \sigma^{\cca}(x_0,y_0) \stackrel{\sigma^{\cca}(f_0,g_0)}{\to} \sigma^{\cca}(x_1,y_1) \stackrel{\sigma^{\cca}(f_1,g_1)}{\to} 
  \sigma^{\cca}(x_2,y_2) \stackrel{\sigma^{\cca}(f_2,g_2)}{\to} \ldots 
\end{eqnarray*}
has colimit
\begin{eqnarray*}
  (\sigma^{\cca}(x_k, y_k) &\stackrel{\sigma^{\cca}(\phi_k,\psi_k)}{\to} & \sigma^{\cca}(x,y))_k.
\end{eqnarray*}
A morphism of continuous categorical $\Sigma$-algebras is
a categorical $\Sigma$-algebra morphism which preserves the 
distinguished initial object and colimits of all $\omega$-diagrams. 
In what follows, we will often write just $\sigma$ for $\sigma^\cca$,
in particular when $\cca$ is understood.

For later use, we note that if $\cca$ and $\ccb$ are continuous categorical 
$\Sigma$-algebras then so is $\cca \times \ccb$. Moreover, for each $k \geq 0$, 
the category 
$[\cca^k \to \cca]$ of all continuous functors $\cca^k \to \cca$ is also 
a continuous categorical $\Sigma$-algebra, where 
for each $\sigma \in \Sigma_n$, 
\[
\sigma^{[\cca^k\to \cca]}(f_1,\ldots,f_n) = \sigma^\cca \circ \langle f_1,\ldots,f_n \rangle
\]
with  $\langle f_1,\ldots,f_n \rangle$ standing for the target tupling of the 
continuous functors $f_1,\ldots,f_n: \cca^k \to \cca$. On natural transformations,
$\sigma^{[\cca^k \to \cca]}$ is defined in a similar fashion. In $[\cca^k \to \cca]$,
colimits of $\omega$-diagrams are formed pointwise.

In the rest of the paper, we will assume, without loss of generality,
 that the signature $\Sigma$ contains a special symbol denoted $\bot$
 or $0$ of rank $0$, interpreted in a continuous categorical
 $\Sigma$-algebra as its initial object.

\subsection{Synchronization trees}\label{Sect:syntrees}

A \textbf{synchronization tree} $t=(V,v_0,E,l)$ 
over an alphabet $A$ of `action symbols'
consists of 
\begin{itemize}
\item a finite or countably infinite
set $V$ of `vertices' and an element $v_0 \in V$, the `root';
\item a set $E \subseteq V \times V$ of `edges';
\item a `labelling function' $l: E \to A \cup \{\ex\}$ where $\ex \not\in A$ is a label used to denote successful termination.
\end{itemize}
These data obey the following restrictions.
\begin{itemize}
\item $(V,v_0,E)$ is a rooted tree: for each
$u \in V$, there is a unique path $v_0 \leadsto u$.
\item If $e=(u,v) \in E$ and $l(e)=\ex$, then $v$ is a leaf,
and $u$ is called an \textbf{exit vertex}.
\end{itemize}
A synchronization tree is {\em deterministic} if no node in the tree
has two equally-labelled outgoing edges. 

A \textbf{morphism} $\phi:t \to t'$ of synchronization trees is a
function $V \to V'$ which preserves the root, the edges and the
labels, so that if $(u,v)$ is an edge of $t$, then
$(\phi(u),\phi(v))$ is an edge of $t'$, and $l'(\phi(u),
\phi(v))=l(u,v)$. {In particular such a morphism maps the root of $t$ to the root of $t'$.} Morphisms are therefore functional
\textbf{simulations}~\cite{Mi71,Par81}.  It is clear that the trees over $A$ and tree morphisms form a category, which we denote by $\Sta$. The tree that has a single vertex
and no edges is initial. It is known that the category of trees has
colimits of all $\omega$-diagrams, see~\cite{BEbook}. (It also has
binary coproducts.) In order to make the category of trees small, we
may require that the vertices of a tree form a subset of some fixed
infinite set.

\begin{rem} 
Suppose that $(\phi_n: t_n \to t_{n+1})_{n \geq 0}$ is an $\omega$-diagram
in $\Sta$, where $t_n = (V_n,v_n,E_n,l_n)$ for each $n\geq 0$. 
Then the colimit $\Colim(\phi_n: t_n \to t_{n+1})_{n \geq 0}$
can be constructed in the expected way. First,
we define for each $n \leq m$ the tree morphism $\phi_{n,m}$ 
as the composition of the morphisms $\phi_n,\ldots,\phi_{m -1}$ {or the identify morphism if $n=m$.}
{Then we take the disjoint union $\{ (v,n) \mid v \in V_n \;\textrm{and } n \geq 0\}$ of the sets $V_n$, $n \geq 0$, and 
define a vertex $(v,n)$  equivalent to a vertex $(v',m)$ if and only if
there is some $k \geq \max\{n,m\}$ with $\phi_{n,k}(v) = \phi_{m,k}(v')$.
The set of vertices of the colimit tree will be the equivalence classes 
of $\{ (v,n) \mid v \in V_n \;\textrm{and } n \geq 0\}$. The root will be  the equivalence class containing the roots $(v_n,n)$ 
of the trees $t_n$, $n \geq 0$. When $C$ and $C'$ are 
equivalence classes, the pair $(C,C')$ is an edge in the 
colimit tree exactly when there exist
$(v,n) \in C$ and $(v',n) \in C'$ for some $n$
such that $E_n$ contains the edge $(v,v')$. The label 
of this edge is $l((C,C')) = l_n((v,v'))$.} Note that, in this case, for each
$k \geq n$ we have that $(\phi_{n,k}(v), \phi_{n,k}(v'))\in E_k$,
and that the label of this edge in $t_k$ is equal to $l_n(v,v')$.
This defines the object part $t$ of the colimit. The canonical morphism
$t_n \to t$ maps a vertex $v\in V_n$ to its equivalence class.

In the particular case when the morphisms $\phi_n$ are injective, we 
may usually assume that  
$V_n \subseteq V_{n + 1}$ for each $n$ and that the morphism  
$\phi_n$ is the inclusion of $V_n$ into $V_{n+1}$.
In this case the colimit is simply the `union' of the trees $t_n$.
\end{rem}

Consider the sequence $t_n=(\{0,\ldots,n\},0,E_n,l_n)$, $n \geq 0$ of synchronization trees, where $E_n = \{ (i,i+1) \mid 0\leq i<n \}$ and $l_n$ maps each edge in $E_n$ to the action symbol $a$. Let $\psi_n: t_n \to t_{n+1}$ be the inclusion map from $\{0,\ldots,n\}$ into $\{0,\ldots,n+1\}$. Then $(\phi_n = t_n \to t_{n+1})_{n\geq 0}$ is an $\omega$-diagram and its colimit is $(f_n : t_n \to t_\omega)_n$, where
\[
t_\omega = (\{i \mid i \geq 0\}, 0, \{(i,i+1) \mid i \geq 0\},l)
\]
and $l$ maps each edge to $a$. Pictorially, we have:

\medskip

	
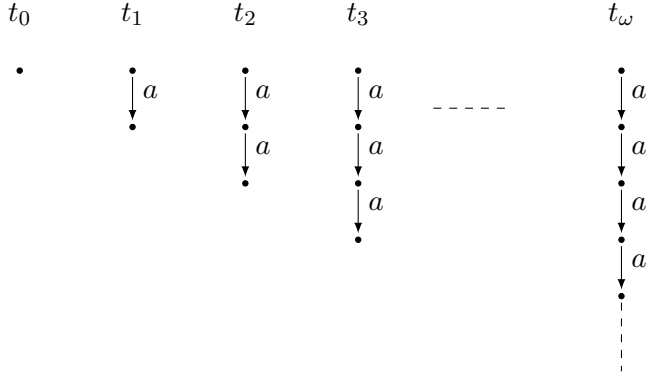
\begin{tikzpicture}[transform shape,scale=1.0]
 \tikzstyle{n}=[circle, draw=black, fill=black, inner sep = 0.25mm, outer sep= 0.5mm]
 \tikzstyle{e}=[ inner sep = 0mm, outer sep= 0mm]
      \begin{scope}
      	\node at (0,0.75) {$t_0$}; 
      	\node[n] (n0) at (0,0) {};  	
      \end{scope}

	\begin{scope}[xshift=1.5cm]
      	\node at (0,0.75) {$t_1$}; 
      	\node[n] (n0) at (0,0) {}; 
      	\node[n] (n1) at (0,-0.75) {};
      	\draw[-latex] (n0) to[right] node[yshift=1mm] {$a$} (n1);  	
      \end{scope}
      
      \begin{scope}[xshift=3cm]
      	\node at (0,0.75) {$t_2$}; 
      	\node[n] (n0) at (0,0) {}; 
      	\node[n] (n1) at (0,-0.75) {};
      	\node[n] (n2) at (0,-1.5) {};
      	\draw[-latex] (n0) to[right] node[yshift=1mm] {$a$} (n1);
      	\draw[-latex] (n1) to[right] node[yshift=1mm] {$a$} (n2);
      \end{scope}
      
       \begin{scope}[xshift=4.5cm]
      	\node at (0,0.75) {$t_3$}; 
      	\node[n] (n0) at (0,0) {}; 
      	\node[n] (n1) at (0,-0.75) {};
      	\node[n] (n2) at (0,-1.5) {};
      	\node[n] (n3) at (0,-2.25) {};
      	\draw[-latex] (n0) to[right] node[yshift=1mm] {$a$} (n1);
      	\draw[-latex] (n1) to[right] node[yshift=1mm] {$a$} (n2);
      	\draw[-latex] (n2) to[right] node[yshift=1mm] {$a$} (n3);
      \end{scope}
      
      \begin{scope}[xshift=5.5cm]
      	\draw[dashed] (0,-0.5) -- (1,-0.5);	
      \end{scope}

	  \begin{scope}[xshift=8cm]
      	\node at (0,0.75) {$t_\omega$}; 
      	\node[n] (n0) at (0,0) {}; 
      	\node[n] (n1) at (0,-0.75) {};
      	\node[n] (n2) at (0,-1.5) {};
      	\node[n] (n3) at (0,-2.25) {};
      	\node[n] (n4) at (0,-3) {};
      	\draw[-latex] (n0) to[right] node[yshift=1mm] {$a$} (n1);
      	\draw[-latex] (n1) to[right] node[yshift=1mm] {$a$} (n2);
      	\draw[-latex] (n2) to[right] node[yshift=1mm] {$a$} (n3);
      	\draw[-latex] (n3) to[right] node[yshift=1mm] {$a$} (n4);
      	\draw[dashed] (n4) -- (0,-4);
      \end{scope}

 \end{tikzpicture}

\medskip

The category $\Sta$ of synchronization trees over $A$ is equipped with
two binary operations: $+$ (sum) and  $\cdot$  
(sequential product  or sequential composition),
and either with a unary operation or a constant associated with 
each letter $a \in A$.  

The \textbf{sum} $t+t'$ of two trees 
is obtained by taking the disjoint  
union of the vertices of $t$ and $t'$ and identifying the roots.
The edges and labelling are inherited.  The 
\textbf{sequential product} $t \cdot t'$ of
two trees is obtained by replacing each edge of $t$ labelled $\ex$ 
by a copy of $t'$. 
With each letter $a \in A$, we can either associate a constant,
or a unary \textbf{prefixing operation}. 
As a constant,  $a$ denotes the tree with
vertices $v_0,v_1,v_2$ and two edges: the edge $(v_0,v_1)$, labelled
$a$, and the edge $(v_1,v_2)$, labelled $\ex$.  
As an operation $a(t)$ is the tree $a \cdot t$, for each tree $t$.
Let $0$ denote the tree with no edges and $1$ the tree 
with a single edge labelled $\ex$. 
On morphisms, all operations are defined in the expected way.
For example, if $h: t \to t'$ and $h': s \to s'$, then 
$h + h'$ is the morphism that agrees with $h$ on the 
non-root vertices of $t$ and that agrees with $h'$ on the 
non-root vertices of $s$. The root of $t+s$ is mapped 
to the root of $t' + s'$.

In the sequel we will consider two signatures for synchronization
trees, $\Gamma$ and $\Delta$. The signature $\Gamma$ contains $+$, $0,1$
and each letter $a \in A$ as a \emph{unary} symbol. In contrast,
$\Delta$ contains $+,\cdot$, $0,1$ and each letter $a\in A$ as a 
\emph{nullary} symbol. It is known that, for both signatures, 
$\Sta$ is a continuous categorical algebra.
See \cite{BEbook} for details. 

In what follows, for each $a \in A$ and term/tree $t$, we write $a.t$
for $a(t)$. We shall also abbreviate $a.1$ to $a$, and write $a^n.t$ for 
\[
\underbrace{a.\ldots.a.}_{\text{$n$ times}}t . 
\]

\begin{rem}
  Note that $\cdot$ is associative and has $1$ as
  unit, at least up to isomorphism, and that for all trees $t_1,t_2$
  and $s$, $(t_1 + t_2)\cdot s = (t_1 \cdot s) + (t_2 \cdot s)$ and $0
  \cdot s = 0$ up to isomorphism.
\end{rem}

\begin{rem}
\label{rem:delta-implies-gamma}
A continuous categorical $\Delta$-algebra $C$ naturally induces an associated $\Gamma$-algebra $D$. For each letter $a \in A$, the functor 
$a^{D}$ associated to unary symbol $a$ is defined for each object and morphism $x$ by $a^{D}(x) = a^{C} \cdot^{C} x$. The functors for the other symbols are inherited (i.e., for all $f \in \Sigma \setminus A$, $f^{D}=f^{C}$). {It is easy to see that $D$ is indeed a continuous categorical $\Gamma$-algebra. In particular the $\Delta$-algebra of synchronization trees induces in this sense the $\Gamma$-algebra of synchronization trees.}
\end{rem}

\begin{defi}\label{def:bisimilarity-languageeq}
Two synchronization trees $t=(V,v_0,E,l)$ and $t'=(V',v'_0,E',l')$ are
{\em bisimilar}~\cite{Mil89CC,Par81} if there is some symmetric
relation $R\subseteq (V \times V')\cup (V' \times V)$,
called a {\em bisimulation}, 
 that relates
their roots, and such that if $(v_1,v_2)\in R$ and there is some edge
$(v_1,v_1')$, then there is an equally-labelled edge $(v_2,v_2')$ with
$(v_1',v_2')\in R$.

 The {\em path language} of a synchronization tree
is composed of the words in $A^*$ that label a path from the root to
the source of an exit edge. Two trees are {\em language equivalent} if
they have the same path language.
\end{defi}

\begin{lem}
\label{lem:bisim-det-iso}
 If $t$ and $t'$ are deterministic, bisimilar synchronization trees, then they are isomorphic. 	
\end{lem}

\begin{rem}
The above lemma fails up to language equivalence. For example, the trees \tikz{\tikzstyle{n}=[circle, draw=black, fill=black, inner sep = 0.25mm, outer sep= 0.5mm]; \node[n] at (0,0) {};} and \tikz[baseline=-0.6em]{\tikzstyle{n}=[circle, draw=black, fill=black, inner sep = 0.25mm, outer sep= 0.5mm]; \node[n] (r) at (0,0) {}; \node[n] (s) at (0,-0.5) {}; \draw[-latex] (r) to[right] node[yshift=1mm] {$a$} (s);} are language equivalent but not isomorphic.
\end{rem}

We conclude this section by defining an operation on synchronization
trees over $A \uplus B$ that contracts all edges labelled by some symbol in $B$. The definition of this operation is illustrated in Figure~\ref{fig:contraction-synchronization-trees}. {This notion will, in particular, be used to characterize $\Gamma$-algebraic trees in the Caucal hierarchy.}

\begin{defi}
Let $t=(V,v_0,E,l)$ be synchronization tree over an alphabet $A \uplus B$.
The $B$-contraction of $t$ is the synchronization tree $t'=(V',v_{0},E',l')$
defined by:
\begin{itemize}
\item $V'= \{ v' \in V \mid \exists v \in V\  (v,v') \in E \;\textrm{with}\; l(v,v') \in A \cup \{\ex\}  \} \cup \{v_{0} \}$,
\item a pair $(v,v') \in V' \times V'$ is an edge in $E'$ if there exists a path in $t$ from $v$ to $v'$  labelled by a word in $B^{*}(A\cup \{\ex\})$,
\item for each  $(v,v') \in E'$, $l(v,v')$ is the unique letter $a \in A \cup \{\ex\}$ such that 
the path from $v$ to $v'$ in $t$ is labelled in $B^{*}a$.
\end{itemize}
\end{defi}

\begin{figure}
\begin{center}
\begin{tikzpicture}[edge from parent/.style={draw,-latex}]
 \tikzstyle{n}=[circle, draw=black, fill=black, inner sep = 0.25mm, outer sep= 0.5mm]
 \tikzstyle{e}=[ inner sep = 0mm, outer sep= 0mm]
 \tikzstyle{l}=[font=\tiny,label distance=-1mm]
 \tikzstyle{level}=[level distance=0.85cm,sibling distance=1cm]
 \tikzstyle{level 1}=[level distance=0.85cm,sibling distance=1.25cm]

\node[n,label={[l]above:{$(v0)$}}] {}
  child { node[n,label={[l]below:{$(v_1)$}}] {}
  edge from parent node[draw=none,above left] {$a$}}
  child { node[n,label={[l]right:{$(v_2)$}}] {}
    child { node[n,label={[l]left:{$(v_4)$}}] {}
      child { node[n,label={[l]below:{$(v_7)$}}] {}
      edge from parent node[draw=none,left] {$b$}}
    edge from parent node[draw=none,left] {$e$}}
    child { node[n,label={[l]right:{$(v_5)$}}] {}
      child { node[n,label={[l]below:{$(v_8)$}}] {}
      edge from parent node[draw=none,right] {$\ex$}}
    edge from parent node[draw=none,right] {$c$}}
  edge from parent node[draw=none,right] {$e$}}
  child { node[n,label={[l]right:{$(v_3)$}}] {}
    child { node[n,label={[l]right:{$(v_6)$}}] {}
      child { node[n,label={[l]below:{$(v_9)$}}] {}
      edge from parent node[draw=none,right] {$b$}}
    edge from parent node[draw=none,right] {$e$}}
  edge from parent node[draw=none,above right] {$e$}}
;

 \tikzstyle{level 1}=[level distance=0.85cm,sibling distance=0.85cm]

\node[n,label={[l]above:{$(v0)$}},xshift=5cm]  {} 
  child { node[n,label={[l]below:{$(v_1)$}}] {}
  edge from parent node[draw=none,above left] {$a$}}
  child { node[n,label={[l]below:{$(v_7)$}}] {}
  edge from parent node[draw=none,left] {$b$}}
  child { node[n,label={[l]right:{$(v_5)$}}] {}
    child { node[n,label={[l]below:{$(v_8)$}}] {}
    edge from parent node[draw=none,right] {$\ex$}}
  edge from parent node[draw=none,right] {$c$}}
  child { node[n,label={[l]below:{$(v_9)$}}] {}
  edge from parent node[draw=none,above right] {$b$}}
;
\end{tikzpicture}
\end{center}

  \caption{A synchronization tree (on the left) over $\{a,b,c,e\}$ and its $\{e\}$-contraction (on the right)}
  \label{fig:contraction-synchronization-trees}
\end{figure}
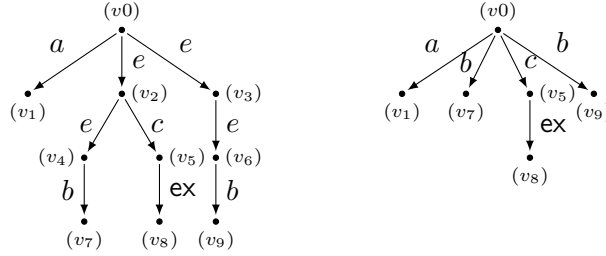

\subsection{Edge-labelled graphs}\label{Sect:edge-labelled-graphs}

An \emph{edge-labelled graph} (or simply graph), whose edges are labelled by letters in a finite alphabet $A$, is a pair $(V,E)$ where $V$ is a finite or countable set of vertices and $E \subseteq V \times A \times V$ is a set of labelled edges. The edge $(u,a,v)$ has source $u$, target $v$ and label $a$. The existence of an edge $(u,a,v)$ in $E$ is denoted by $u\erb{a}{G}v$ or simply $u \era{a} v$ when $G$ is clear from the context. This notation is extended to words in $A^{*}$ in the usual way. For a language $L$ over $A$, we write $u \era{L} v$ if there is some word $w\in L$ such that $u \era{w} v$.

A graph $(V,E)$ is {\em deterministic} if for all $u, v_1, v_2 \in V$ and $a\in A$, if $u \era{a} v_{1}$ and $u \era{a} v_{2}$ then $v_{1}=v_{2}$.

A \emph{root} of a graph $(V,E)$ labelled by $A$ is a vertex $v_{0}$ such that for all $v \in V$, there exists a path from $v_{0}$ to $v$.
If this path is unique for each vertex $v$, the graph $(V,E)$ is a \emph{tree} with $v_{0}$ as its root.

A graph $(V,E)$ labelled in $A \cup \{\ex\}$ represents a synchronization tree if $(V,E)$ is a tree with root $v_{0}$ and if the target of an $\ex$-labelled edge is always a leaf of the tree. We will identify such a graph
with the synchronization tree $(V,v_{0},E',l)$ where $E'=\{ (u,v) \mid \exists a \in A \cup \{\ex\}. u \era{a} v\}$ and for all $(u,v) \in E'$, $l((u,v))$ 
is the unique symbol in $A \cup \{\ex\}$ such that $(u,l((u,v)),v)$ belongs to $E$.

The definition of $B$-contraction for labelled graphs generalizes the
case of synchronization trees introduced in
Section~\ref{Sect:syntrees}. The $B$-contraction of a graph $G$
labelled in $A \uplus B$ 
from one of its root vertices $r$ is the graph
labelled by $A$ whose vertices are the targets (in $G$) of edges labelled
in $A$ together with the vertex $r$. There is an edge from a vertex
$u$ to a vertex $v$ if there is a path from $u$ to $v$ in $G$ labelled
by a word in $B^{*}A$.

\subsection{Monadic second-order logic on edge labelled graphs.}
\emph{Mo\-na\-dic second-order  logic (MSO)} on graphs is the extension of first-order logic with the ability to quantify over sets of vertices. We use monadic-second order logic over edge-labelled  graphs with the standard syntax and semantics (see
e.g. \cite{Ebbinghaus95} for a detailed presentation).

\emph{Monadic second-order formulas}  are built using \emph{first-order variables}, which are interpreted as vertices of the graph and are denoted by lower case letters $x,y\ldots$ and \emph{monadic second-order variables}, which are interpreted as sets of vertices of the graph and are denoted by upper case letters $X,Y\ldots$ Atomic formulas are of the form $x \in X$, $x=y$ and $x \era{a} y$ where $x$ and $y$ are first-order variables, $X$ is a second-order variable and $a$ is an edge label. Complex formulas are built as usual from atomic ones using Boolean connectives and quantifiers. Free and bound occurrences of variables in a formula are defined in the standard fashion. We write $\varphi(X_1, \ldots, X_n, y_1, \ldots, y_m)$ to denote that the formula $\varphi$ has free variables in $\{X_1,\ldots,X_n,y_1,\ldots,y_m\}$.
A \emph{closed formula} has no free variables. 

For a formula $\varphi(X_1, \ldots, X_n, y_1, \ldots, y_m)$, a graph $G$, vertices $u_1,\ldots,u_m$ of $G$ and sets of vertices $U_1,\ldots,U_n$ of $G$, we write $G \models 
\varphi[U_1,\ldots,U_n,u_1,\ldots,u_m]$ to denote that $G$ 
\emph{satisfies} $\varphi$ when the free variable $X_i$, $i \in [n]$, is interpreted as the set of vertices $U_i$ and the free variable $x_j$, $j \in [m]$, is interpreted as the vertex $u_j$. For a closed formula $\varphi$, we simply write $G \models \varphi$. The \emph{MSO-theory} of $G$ is the set of closed formulas satisfied by $G$. 

For example, the MSO formula 
\[
 \varphi(x,y) = \forall X. \left[\left( \forall x \forall y ( x \in X
\wedge x \era{a} y) \Rightarrow y \in X \right) \wedge x \in X \right] \Rightarrow y \in X
\]
expresses that there exists a path from $x$ to $y$ whose edges are labelled $a$. 

\subsection{Graph transformations}\label{Sect:graph-transformations}

In the remainder of the paper, we will use several transformations of edge-labelled graphs, namely the unfolding of a graph, MSO-interpretations and MSO-transductions.

The \emph{unfolding}  $\Unf(G,r)$ of a graph $G$ from one of its vertices $r$ is the tree whose vertices are the paths in $G$ starting from $r$ and with an edge labelled $a$ between 
two such paths $\pi$ and $\pi'$ if $\pi'$ extends $\pi$ by exactly one edge labelled $a$.

An \emph{MSO-interpretation} is a graph transformation specified using MSO-formu\-las. An MSO-interpretation 
 is given by a tuple of MSO-formulas of the form $(\delta(x),(\varphi_b(x,y))_{b \in B})$, where $B$ is a set of labels. This interpretation when applied to a graph $G$ will produce a graph, denoted $\mathcal{I}(G)$, whose edges have labels in $B$. The set of vertices of $\mathcal{I}(G)$ is  the subset of the vertices of $G$ satisfying the formula $\delta(x)$. For each edge label $b \in B$, 
 there is an edge from $u$ to $v$ labelled by $b$ in $\mathcal{I}(G)$ if and only if $G$ satisfies $\varphi_b[u,v]$. More formally, the graph $(V',E')$ obtained by applying $\mathcal{I}$ to a graph $(V,E)$ is such that:
 \[
 \begin{array}{lcl}
 	V' & = & \{ u \in V \mid G \models \delta[u] \} \\
 	E' & = & \{ (u,b,v) \in V' \times B \times V' \mid G \models \varphi_b[u,v]\}. \\ 
 \end{array}
 \]
 
 MSO-interpretations cannot increase the number of vertices  of a graph. To overcome
this weakness the notion of a transduction was introduced by Courcelle in \cite{Courcelle94}. The idea is to perform an operation that will increase the number of vertices before applying the MSO-interpretation.
Let $K=\{k_1,\ldots,k_m\}$ be a finite set of edge labels.  A \emph{$K$-copying operation} applied to a graph $G$
adds, to every vertex of $G$, $m$ outgoing arcs respectively labelled by
$k_1,\ldots k_{m-1}$ and $k_m$ all going to fresh vertices.  An
\emph{MSO-transduction} is a $K$-copying operation
followed by an MSO-interpretation.

\subsection{Algebraic objects and functors}
\label{sec-algebraic objects}

We start by introducing the types of recursion schemes studied in this paper. 
\begin{defi}
\label{def: recursion scheme}
Let $\Sigma$ be a signature. A \textbf{$\Sigma$-recursion scheme}, or
\textbf{recursion scheme over $\Sigma$}, is a sequence $E$ of equations
\begin{eqnarray}
\label{eq:scheme}
\nonumber
F_1(v_1,\ldots,v_{k_1}) &=&  t_1 \\
 & \vdots & \\
\nonumber
 F_n(v_1,\ldots,v_{k_n}) &=& t_n
\end{eqnarray}
where each $t_i$ is a term over the signature $\Sigma_\Phi = \Sigma \cup \Phi$
in the variables $v_1,\ldots,v_{k_i}$, and $\Phi$ contains the symbols
$F_i$ (sometimes called `functor variables' or `function variables") 
of rank $k_i$, $i \in[n]$. 
A $\Sigma$-recursion scheme is \textbf{regular} if $k_i=0$, for 
{each}
$i \in [n]$.
\end{defi}

Suppose that $C$ is a continuous categorical $\Sigma$-algebra.
 Define
\begin{eqnarray*}
C^{r(\Phi)} &=&   [C^{k_1} \to C]\times \cdots \times [C^{k_n} \to C].
\end{eqnarray*}
Then $C^{r(\Phi)}$ is a continuous categorical $\Sigma$-algebra, as
noted in Section~\ref{Sect:contalg}.

When each $F_i$, $i \in [n]$, is interpreted as a continuous functor
$f_i: C^{k_i} \to C$, each term over the extended signature
$\Sigma_\Phi =\Sigma \cup \Phi $ in the variables $v_1,\ldots,v_m$
induces a continuous functor $C^m \to C$ that we denote by
$t^C(f_1,\ldots,f_n)$.  In fact, $t^C$ is a continuous functor
\begin{eqnarray*}
t^C &:& C^{r(\Phi)} \to [C^m \to C].
\end{eqnarray*}
More precisely, we define $t^C$ as follows. Let $f_i,g_i$
denote continuous functors $C^{k_i}\to C$, $i\in [n]$,
and let $\alpha_i$ be a natural transformation $f_i \to g_i$
for each $i\in [n]$. 
When $t$ is the variable $v_i$, say, then $t^C(f_1,\ldots,f_n)$ is the $i$th 
projection functor $C^m \to C$, and $t^C(\alpha_1,\ldots,\alpha_n)$ 
is the identity natural transformation corresponding to this projection 
functor. Suppose now that $t$ is of the form $\sigma(t_1,\ldots,t_k)$, 
where $\sigma \in \Sigma_k$ and $t_1,\ldots,t_k$ are terms. Then 
$t^C(f_1,\ldots,f_n) = \sigma^C \circ \langle h_1,\ldots,h_k \rangle$
and $t^C(\alpha_1,\ldots,\alpha_n) = \sigma^C \circ \langle \beta_1,\ldots,\beta_k\rangle $,
where $h_j = t_j^C(f_1,\ldots,f_n)$ and $\beta_j =
t_j^C(\alpha_1,\ldots,\alpha_n)$ for all $j\in [k]$. (Here, we use the
same notation for a functor and the corresponding identity natural
transformation.)   Finally, when $t$ is of the form
$F_i(t_1,\ldots,t_{k_i})$, then $t^C = f_i \circ \langle
h_1,\ldots,h_{k_i} \rangle$, and the corresponding natural
transformation is $\alpha_i \circ \langle \beta_1,\ldots,\beta_{k_i}
\rangle$, where the $h_j$ and $\beta_j$, $j \in [k_i]$, are defined
similarly as above.

Note that if each $\alpha_i: f_i \to f_i$ is an identity 
natural transformation (so that $f_i = g_i$, for all $i \in [n]$), 
then $t^C(\alpha_1,\ldots,\alpha_n)$ is the identity 
natural transformation 
$t^C(f_1,\ldots,f_n) \to t^C(f_1,\ldots,f_n)$.

In any continuous categorical $\Sigma$-algebra $C$, 
by target-tupling the functors $t_i^C$ associated with the right-hand sides of the equations in a recursion scheme $E$ as in (\ref{eq:scheme}), we obtain a continuous functor
\begin{eqnarray*}
  E^{C}: C^{r(\Phi)} & \to & C^{r(\Phi)}.
\end{eqnarray*}
Indeed,  $t_i^{C}:C^{r(\Phi)} \to [C^{k_i} \to C]$, for $i \in [n]$, so that
\begin{eqnarray*}
E^{C}=\langle t_1^{C},\ldots,t_n^{C} \rangle : C^{r(\Phi)} & \to & C^{r(\Phi)}.
\end{eqnarray*}
Thus, $E^{C}$ has an initial fixed point in $C^{r(\Phi)}$, unique up to 
natural isomorphism, that we denote by
\begin{eqnarray*}
  |E^C| &=&  (|E|_1^C,\ldots,|E|_n^C), 
\end{eqnarray*}
so that, in particular,
\begin{eqnarray*}
  |E|_i^C &=& t_i^C(|E|_1^C,\ldots,|E|_n^C),
\end{eqnarray*}
at least up to isomorphism, for each $i \in [n]$.
 
 It is well
known that the initial fixed point $|E|^C$ can be `computed' in the
following way. 
 Let $g_0 = ( g_{0,1},\ldots,g_{0,n} ) $, where, for
each $j \in [n]$, 
 $g_{0,j} = 0^{[C^{k_j} \to C]}$ is the constant
functor $C^{k_j} \to C$ 
 determined by the object $0^C$. Then, for
each $i \geq 0$, define 
 $g_{i+1} = (g_{i+1,1},\ldots,g_{i+1,n})$,
where $g_{i+1,j} 
 = t_j^C(g_i)$ for all $j\in [n]$. 
 Next, let
$\phi_0 = (\phi_{0,1},\ldots,\phi_{0,n})$, 
 where $\phi_{i,j}$,
$j\in [n]$, is the unique 
 natural transformation $0^{[C^{k_j} \to
    C]} \to g_{1,j}$.
 Moreover, for each $i \geq 0$, define 
$\phi_{i+1} = (\phi_{i+1,1},\ldots,\phi_{i+1,n})$
 as the natural
transformation $\phi_{i+1,j} = t_j^C(\phi_i)$. 
 Then $|E|^C$ is the
colimit of the $\omega$-diagram $(\phi_{i}: g_{i} \to g_{i+1})_{i \geq
  0}$
 in the continuous functor category $C^{r(\Phi)}$.

\begin{defi}
Suppose that $C$ is a continuous categorical $\Sigma$-algebra.
We call a functor $f: C^m \to C$
\textbf{$\Sigma$-algebraic}, if there is a recursion scheme $E$
such that $f$  is isomorphic to $|E|_1^C$, the first component 
of the above-mentioned initial solution of $E$. 
When $m = 0$, we identify a $\Sigma$-algebraic 
functor with a \textbf{$\Sigma$-algebraic object}. 
Last, a \textbf{$\Sigma$-regular object} is an object
isomorphic to the first component of the initial 
solution of a $\Sigma$-regular recursion scheme. 
\end{defi}

In particular, we get the notions of $\Gamma$-algebraic and $\Gamma$-regular 
trees, and $\Delta$-algebraic and $\Delta$-regular trees.

\begin{rem} 
Suppose that $C$ is the continuous categorical $\Sigma$-algebra $\Sta$, where 
$\Sigma$ is either $\Gamma$ or $\Delta$, and consider a recursion
scheme $E$ of the form (\ref{eq:scheme}). If the continuous functors $g_1: \Sta^{k_1} \to \Sta$, $\ldots$, 
$g_n: \Sta^{k_n} \to \Sta$ preserve injective synchronization tree morphisms
(in each variable),  then so does each $t_j(g): \Sta^{r(\Phi)} \to \Sta^{[k_j \to k]}$,
where $g = (g_1,\ldots,g_n)$. It follows by induction that the functors 
$g_{i,j}$ preserve injective synchronization tree morphisms. moreover, 
the components of the natural transformations $\phi_{i,j}$ are injective 
synchronization tree morphisms. Thus, when $k_1 = 0$, so that the recursion
scheme defines a tree, each $\phi_{i,1}$ is an embedding of the tree $g_{i,1}$ 
into $g_{i+1,1}$. If we represent the trees $g_{i,1}$ so that these 
embeddings are inclusions, then the colimit $|E|_1^{\Sta}$ of the $\omega$-diagram 
$(\phi_{i,1} : g_{i,1} \to g_{i+1,1})_{i \geq 0}$ becomes the union 
of the trees $g_{i,1}$, $i \geq 0$.
\end{rem}

\begin{rem}
\label{rem:arbitrary-rank-sum}
  It is sometimes convenient to add to the signature $\Gamma$ sums of arbitrary
  nonzero rank.  For this, we consider the signature\footnote{Although
    the signature is infinite, we will always only use a finite subset
    of it.}  $\tilde{\Gamma}= \{ +^{n} \mid n \geq 1 \} \cup A \cup
  \{0,1\} $ where the rank of each symbol $+^n$ is $n$.  The
  synchronization tree associated with a $\tilde{\Gamma}$-term tree is
  defined similarly as the synchronization tree associated to a
  $\Gamma$-term tree. The $\tilde{\Gamma}$-algebraic synchronization trees 
  are $\Gamma$-algebraic.

  Indeed we can transform an algebraic $\tilde{\Gamma}$-scheme into a
  $\Gamma$-scheme defining the same synchronization tree by replacing
  every occurrence of a subterm of the form $+^{n}(t_{1},\ldots,t_{n})$ on the right-hand side of an equation
   by $((t_{1} + t_{2}) +
  t_{3}) \ldots + t_{n}$ for $n \geq 3$, every subterm of the form
  $t_{1} +^{2} t_{2}$ by $t_{1} + t_{2}$, and finally every subterm of
  the form $+^{1}(t_{1})$ by $t_{1} + 0$.
\end{rem}

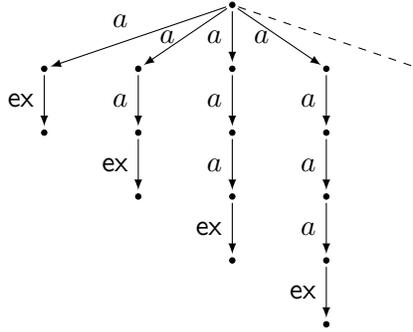
\begin{figure}
  \begin{center}
    \begin{tikzpicture}[transform shape,scale=1.0,edge from parent/.style={draw,-latex}]
 \tikzstyle{n}=[circle, draw=black, fill=black, inner sep = 0.25mm, outer sep= 0.5mm]
 \tikzstyle{e}=[ inner sep = 0mm, outer sep= 0mm]
 \tikzstyle{l}=[font=\tiny,label distance=-1mm]
 \tikzstyle{level}=[level distance=0.85cm,sibling distance=1cm]
 \tikzstyle{level 1}=[level distance=0.85cm,sibling distance=1.25cm]

\node[n] {}
  child { node[n] {}
    child { node[n] {}
    edge from parent node[draw=none,left] {$\ex$}}
  edge from parent node[draw=none,above left] {$a$}}
  child { node[n] {}
    child { node[n] {}
      child { node[n] {}
      edge from parent node[draw=none,left] {$\ex$}}
    edge from parent node[draw=none,left] {$a$}}
  edge from parent node[draw=none,left] {$a$}}
  child { node[n] {}
    child { node[n] {}
      child { node[n] {}
        child { node[n] {}
        edge from parent node[draw=none,left] {$\ex$}}
      edge from parent node[draw=none,left] {$a$}}
    edge from parent node[draw=none,left] {$a$}}
  edge from parent node[draw=none,left] {$a$}}
  child { node[n] {}
    child { node[n] {}
      child { node[n] {}
        child { node[n] {}
          child { node[n] {}
          edge from parent node[draw=none,left] {$\ex$}}
        edge from parent node[draw=none,left] {$a$}}
      edge from parent node[draw=none,left] {$a$}}
    edge from parent node[draw=none,left] {$a$}}
  edge from parent node[draw=none,left] {$a$}}
  child { node[e] {}
  edge from parent[dashed,-]}
;

\end{tikzpicture}
    
  \end{center}

\caption{An infinitely branching synchronization tree}
\label{Fig:anyfinitedepth}
\end{figure}

\begin{exa}\label{Ex:anyfinitedepth}
The following $\Delta$-regular recursion scheme 
\begin{equation}
\label{Eqn:Deltascheme}
X = (X \cdot a ) + a 
\end{equation}
has the infinitely branching tree $\sum_{i\geq 1} a^i$ depicted on
Figure~\ref{Fig:anyfinitedepth} as its initial solution. That tree is
therefore $\Delta$-regular. Note that the tree $\sum_{i\geq 1} a^i$ is
also $\Gamma$-algebraic because it is the initial solution of the
following $\Gamma$-algebraic recursion scheme
\begin{eqnarray*}
F_1 & = & F_2(a) \\
F_2(v)  & = & v + F_2(a.v) . 
\end{eqnarray*}
(Recall that we use $a$ as an abbreviation of $a.1$.) The path language associated with the tree $\sum_{i\geq 1} a^i$ is $\{a^i \mid i \geq 1\}$. Note that the subtrees of that tree whose roots are children of the root of $\sum_{i\geq 1} a^i$ are pairwise inequivalent modulo language equivalence. So
$\Gamma$-algebraic recursion schemes can be used to define infinitely branching
trees that have an infinite number of subtrees, even up to
bisimilarity~\cite{Mil89CC,Par81} or language equivalence. 
\end{exa}

\paragraph{Further examples of algebraic synchronization trees} 

Consider the labelled transition system (LTS)
on Figure~\ref{fig:Dyck-like}. The synchronization tree
associated with that LTS is $\Gamma$-algebraic because it is 
defined by the recursion scheme below.
\begin{eqnarray*}
F_1 & = & F_2(1) \\
F_2(v) & = & v + a. F_2(b.F_2(v))  
\end{eqnarray*}
The idea underlying the above definition is as follows. At any given
vertex in the LTS on Figure~\ref{fig:Dyck-like}, the argument $v$ of
$F_2$ denotes the subtree obtained by unfolding the LTS from that vertex
\emph{and not} taking the $a$-labelled edge as a first step. $F_2(v)$
denotes the tree obtained by unfolding the LTS at that vertex.

This algebraic recursion scheme over $\Gamma$ corresponds to the 
regular recursion scheme over $\Delta$, 
\begin{eqnarray*}
G &=& 1 + a\cdot G \cdot b \cdot G.
\end{eqnarray*}

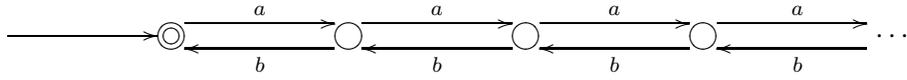
\begin{figure} 
\newcommand{\state}{*+[o][F-]{\makebox(4,4){}}}
\newcommand{\istate}{*+[o][F-]{\makebox(4,4){}}}
\newcommand{\fstate}{*+[o][F=]{\makebox(4,4){}}}
\[\xymatrix@C=12ex@R=12ex{
\ar[r] & \fstate{} \ar@<1ex>[r]^a & \state{}  \ar@<1ex>[r]^a \ar@<1ex>[l]^b  & \state{}  \ar@<1ex>[r]^a \ar@<1ex>[l]^b & \state{}  \ar@<1ex>[r]^a \ar@<1ex>[l]^b & \cdots  \ar@<1ex>[l]^b
}
\]
\caption{An LTS accepting the Dyck language}
\label{fig:Dyck-like}
\end{figure}

As another example, consider the LTS on Figure~\ref{fig:non-BPA2}.  This
LTS is not expressible in BPA modulo modulo bisimilarity---see~\cite[page~206,
  Example (c)]{Moller96}. On the other hand, the synchronization tree
associated with that LTS is $\Gamma$-algebraic because it is the
unique solution of the $\Gamma$-algebraic recursion scheme below.
\begin{eqnarray*}
G & = & F(1,1)  \\
F(v_1,v_2) & = & v_1 + c.v_2 + a.F(b.F(v_1,{v_2}),b.v_2)  
\end{eqnarray*}
The idea underlying the above definition is as follows. At any given
vertex in the LTS on Figure~\ref{fig:non-BPA2}, the argument $v_1$ of
$F$ denotes the subtree obtained by unfolding the LTS from the vertex
\emph{and not} taking the $a$-labelled or $c$-labelled edges as a
first step. The argument $v_2$ of $F$ instead encodes the number of
$b$-labelled edges that one must perform in a sequence in order to
terminate. $F(v_1,v_2)$ denotes the tree obtained by unfolding the LTS
from that vertex.

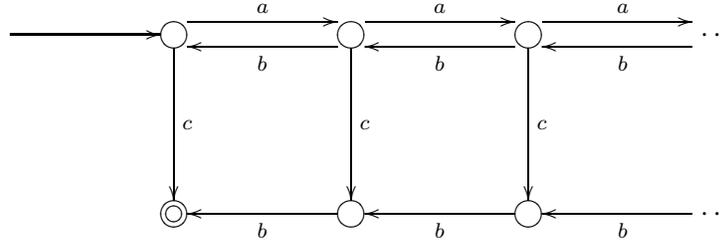
\begin{figure} 
\newcommand{\state}{*+[o][F-]{\makebox(4,4){}}}
\newcommand{\istate}{*+[o][F-]{\makebox(4,4){}}}
\newcommand{\fstate}{*+[o][F=]{\makebox(4,4){}}}
\[\xymatrix@C=12ex@R=12ex{
\ar[r] & \istate{} \ar@<1ex>[r]^a \ar[d]^c & \state{}\ar@<1ex>[l]^b \ar@<1ex>[r]^a \ar[d]^c &  \state{} \ar@<1ex>[r]^a \ar@<1ex>[l]^b \ar[d]^c & \cdots \ar@<1ex>[l]^b \\
& \fstate{} & \state{} \ar[l]^b  & \state{} \ar[l]^b & \ar[l]^b \cdots  
}
\]
\caption{An LTS that cannot be expressed in BPA, but whose
  unfolding is an algebraic synchronization tree}
\label{fig:non-BPA2}
\end{figure}

In our arguments, we will often make use of the following Mezei-Wright theorem
\cite{BEmezei}:

\begin{thm}
\label{thm-mezei}
Suppose that $h: C \to D$ is a morphism of continuous categorical $\Sigma$-algebras.
If $c$ is an algebraic (or regular) object of $C$, then $h(c)$ is an 
algebraic (or regular) object of $D$. Moreover, for every algebraic object $d$ 
of $D$ there is an algebraic object $c$ of $C$ such that $d$ is isomorphic to $h(c)$,
and similarly for regular objects. 
\end{thm}

\subsection{Continuous ordered algebras}\label{Sect:coa}

An important subclass of continuous categorical algebras is formed
by the continuous \emph{ordered} algebras, which constitute the classical 
framework for algebraic semantics \cite{Courcelle83,CourcelleN76,GoguenTWW77,Gue81,Nivat75,Scott71}. 
We say that a continuous categorical $\Sigma$-algebra $\cca$ is ordered if 
its underlying category is a poset 
category, so that there is at most one morphism between any two objects.
The objects of a continuous ordered algebra are usually called its elements
and, as usual, the categorical structure is replaced by a partial order $\leq$. 
The assumptions that $\cca$ has initial object
and colimits of $\omega$-diagrams correspond to the requirements that there is 
a least element $\bot$ and each $\omega$-chain $(a_n)_n$ has a supremum,
denoted $\sup_n a_n$. Thus, $\cca$ is an \emph{$\omega$-complete poset} as is 
each finite power $\cca^n$ of $\cca$, equipped with the pointwise order. 
A continuous functor $\cca^n \to \cca$ is simply a function that preserves the suprema
of $\omega$-chains. Every such continuous function preserves the 
partial order. Continuous functions are ordered pointwise. As is well known, if $\cca$ and $\ccb$ 
are $\omega$-complete posets, then so is the poset $[\cca \to \ccb]$ of all continuous 
functions $\cca \to \ccb$.

Suppose that $A$ is an alphabet and recall the definition of the signatures $\Gamma$ and $\Delta$.
An example of a continuous ordered $\Delta$-algebra is the 
language semiring $P(A^*)$ of all subsets of $A^*$, ordered by set inclusion, where each letter $a\in A$ 
denotes the set $\{a\}$ and $0$ and $1$ are interpreted as the 
empty set and the set $\{\epsilon\}$ containing only the empty word $\epsilon$, and where 
$L_1 + L_2 = L_1 \cup L_2$ and $L_1\cdot L_2 = L_1L_2$ is the concatenation
of $L_1$ and $L_2$, for all $L_1,L_2 \subseteq A^*$. Alternatively, we may 
view $P(A^*)$ as a continuous ordered $\Gamma$-algebra with $a(L) = \{a\}\cdot L$
for all $a\in A$ and $L\subseteq A^*$.

The initial continuous $\Sigma$-algebra may be described as the algebra 
of \emph{$\Sigma$-term trees}. A term tree may be represented as (the isomorphism 
class of) a possibly infinite directed ordered tree whose vertices are 
labelled with the letters in $\Sigma$ such that a vertex 
labelled in $\Sigma_n$ has exactly $n$ successors. In particular, 
vertices labelled in $\Sigma_0$ are leaves. Finite term trees over $\Sigma$ may be identified
with variable-free $\Sigma$-terms. As usual, we identify each symbol
in $\Sigma_0$ with a term.

Let $T_\Sigma^\omega$ denote the set of all $\Sigma$-term trees. We equip 
$T_\Sigma^\omega$ with a partial order defined by $t \leq t'$ iff $t'$ 
can be constructed from $t$ by replacing some leaves labelled $\bot$ by 
some term trees in $T_\Sigma^\omega$. It is well-known 
(see e.g. \cite{Courcelle83,GoguenTWW77,Gue81}) that this relation 
turns $T_\Sigma^\omega$ into an $\omega$-complete partial order with least 
element $\bot$. We may further turn $T_\Sigma^\omega$ into a continuous 
ordered $\Sigma$-algebra. For each $\sigma \in \Sigma_n$ and term trees $t_1,\ldots,t_n$,
the tree $\sigma(t_1,\ldots,t_n)$ is defined as usual as the tree whose 
root is labelled $\sigma$ which has $n$ subtrees isomorphic 
in order to $t_1,\ldots,t_n$. The following fact is known, see e.g. \cite{BEmezei}.

\begin{prop}
\label{prop-freeness}
For every continuous categorical $\Sigma$-algebra $\cca$ there is, 
up to natural isomorphism, a unique continuous categorical $\Sigma$-algebra
morphism $T_\Sigma^\omega \to \cca$. 
\end{prop}

Sometimes it is convenient to represent a term tree in $T_\Sigma^\omega$
as a partial function $t: \N^* \to \Sigma$ whose domain is a prefix closed 
nonempty subset of $\N^*$, where $\N$ denotes the set of positive integers, 
such that whenever $t(u) \in \Sigma_n$,
then for each $i \in \N$, $t(ui)$ is defined if and only if $i \in [n]$. 
The operations are defined in the usual way. For trees $t,t'\in T_\Sigma^\omega$,
we have $t \leq t'$ if and only if for all $u \in \N^*$, either $t(u) = t'(u)$ 
or $t(u) = \bot$. See \cite{Gue81,Nivat75,Scott71} for details.

Yet another 
representation of a term tree in $T_\Sigma^\omega$ is by edge-labelled 
trees. To this end, let $\overline{\Sigma}$ denote the ordinary alphabet whose letters
are the symbols $\sigma_i$ where $\sigma\in \Sigma_n$, $n > 0$ and $i \in [n]$.
When $t$ is a $\Sigma$-term tree, each vertex $u$ may be `addressed" by a word 
 $\overline{u} \in \overline{\Sigma}^*$ which encodes the unique path from the 
root to that vertex. The edge-labelled tree corresponding to $t$ has as its vertex sets 
the addresses of its vertices, together with a vertex $\overline{u}\sigma$ 
whenever $\overline{u}$ is the address of a leaf labelled $\sigma$. 
The edges are given 
as follows. Suppose that $u$ and $v$ are vertices of $t$ with associated addresses 
$\overline{u}$ and $\overline{v}$, respectively. If $v$ is the $i$th successor of $u$ 
and $u$ is labelled $\sigma \in \Sigma_n$ (so that $n > 0$ and $i \in [n]$), then 
there is an edge from $\overline{u}$ to $\overline{v}$ labelled $\sigma_i$. 
Moreover, if $u$ is a leaf vertex of $t$, labelled $\sigma \in \Sigma_0$, then 
there is an edge labelled $\sigma$ from $\overline{u}$ to $\overline{u}\sigma$.

Consider, for example, the following term tree over $\Gamma$.

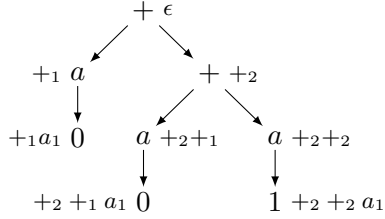
\begin{center}
    \begin{tikzpicture}[transform shape,scale=1.0,edge from parent/.style={draw,-latex}]
 \tikzstyle{n}=[circle, draw=black, fill=black, inner sep = 0.25mm, outer sep= 0.5mm]
 \tikzstyle{m}=[inner sep = 0.25mm, outer sep= 0.5mm]
 \tikzstyle{e}=[ inner sep = 0mm, outer sep= 0mm]
 \tikzstyle{l}=[font=\footnotesize,label distance=-1mm]
 \tikzstyle{level}=[level distance=0.85cm,sibling distance=1cm]
 \tikzstyle{level 1}=[level distance=0.85cm,sibling distance=1.75cm]
 \tikzstyle{level 2}=[level distance=0.85cm,sibling distance=1.75cm]

\node[m,label={[l]right:{$\epsilon$}}] {$+$}
  child { node[m,label={[l]left:{$+_1$}}] {$a$}
  	child {node[m,label={[l]left:{$+_1a_1$}}] {$0$}}}	
  child { node[m,label={[l]right:{$+_2$}}] {$+$}
    child { node[m,label={[l]right:{$+_2 +_1$}}] {$a$}%
    child {node[m,label={[l]left:{$+_2 +_1 a_1$}}] {$0$}}}
    child { node[m,label={[l]right:{$+_2 +_2$}}] {$a$}%
    child {node[m,label={[l]right:{$+_2 +_2 a_1$}}] {$1$}}}
};
\end{tikzpicture}
\end{center}

The resulting edge-labelled tree is
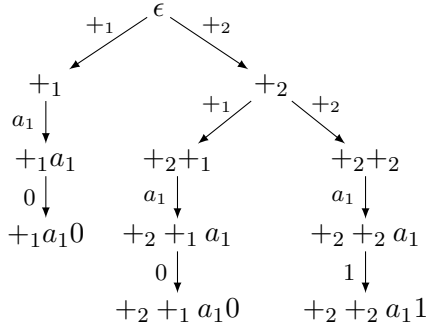
\begin{center}
    \begin{tikzpicture}[transform shape,scale=1.0,edge from parent/.style={draw,-latex}]
 \tikzstyle{n}=[circle, draw=black, fill=black, inner sep = 0.25mm, outer sep= 0.5mm]
 \tikzstyle{m}=[inner sep = 0.25mm, outer sep= 0.5mm]
 \tikzstyle{e}=[ inner sep = 0mm, outer sep= 0mm]
 \tikzstyle{l}=[font=\footnotesize,label distance=-1mm]
 \tikzstyle{level}=[level distance=1cm,sibling distance=5cm]
 \tikzstyle{level 1}=[level distance=1cm,sibling distance=3cm]
 \tikzstyle{level 2}=[level distance=1cm,sibling distance=2.5cm]

\node[m] {$\epsilon$}
  child { node[m] {$+_1$}
  		  child { node[m] {$+_1 a_1$}
  		  		  child { node[m] {$+_1 a_1 0$}
  		  		          edge from parent node[l,left] {$0$}}
  	              edge from parent node[l,left] {$a_1$}}		
  	      edge from parent {node[l,above,xshift=-1mm] {$+_1$}}}	
  child { node[m] {$+_2$}
	  	  child { node[m] {$+_2 +_1$}
  		  		  child { node[m] {$+_2 +_1 a_1$}
  		  		          child { node[m] {$+_2 +_1 a_1 0$}
  		  		          edge from parent node[l,left] {$0$}}
  		  		          edge from parent node[l,left] {$a_1$}}
  	              edge from parent node[l,above,xshift=-1mm] {$+_1$}}		   
           child { node[m] {$+_2 +_2$}
  		  		  child { node[m] {$+_2 +_2 a_1$}
  		  		          child { node[m] {$+_2 +_2 a_1 1$}
  		  		          edge from parent node[l,left] {$1$}}
  		  		          edge from parent node[l,left] {$a_1$}}
  	              edge from parent node[l,above,xshift=1mm] {$+_2$}		
              }edge from parent {node[l,above,xshift=1mm] {$+_2$}}
              }
;
\end{tikzpicture}
\end{center}

\begin{rem} 
Suppose that $C$ is a continuous categorical $\Sigma$-algebra
and $h$ is the essentially unique continuous categorical 
$\Sigma$-algebra morphism $T_\Sigma^\omega \to C$. Then, for each 
finite tree $t\in T_\Sigma^\omega$, $h(t)$ is uniquely determined 
up to isomorphism, since $h$ preserves the operations.
Moreover, when $t,t'$ are finite with\footnote{Here, we 
denote by $t \leq t'$ also the unique morphism $t\to t'$ 
in $T_\Sigma^\omega$ seen as a category.} $t \leq t'$, $h(t \leq t')$ is 
determined by the conditions that $h(\bot)$ is initial 
and $h$ preserves the operations. When $t$ is an infinite tree,
$t$ is the supremum of an $\omega$-chain $(t_n)_n$ of finite trees---say $t_n$ is the approximation of $t$ obtained by removing 
all vertices of $t$ at distance greater than or equal to $n$ 
from the root (and relabeling vertices at distance $n$ by $\bot$). 
Then $h(t)$ is the colimit of the $\omega$-diagram $(\phi_n: h(t_n) \to h(t_{n+1}))_n$,
where $\phi_n = h(t_n \leq t_{n+1})$. When $\Sigma = \Gamma$ 
or $\Sigma = \Delta$, and $C = \Sta$, the morphisms $\phi_n$ are injective,
and thus the colimit $h(t)$ is the `union' of the $h(t_n)$.  
\end{rem}

\subsection{Morphisms}
\label{sec:morphism}

By Proposition~\ref{prop-freeness}, there is an (essentially) unique
morphism of continuous categorical $\Gamma$-algebras $T_\Gamma^\omega
\to \Sta$, as well as an essentially unique continuous categorical
$\Delta$-algebra morphism $T_\Delta^\omega \to \Sta$. In the first
part of this section, we provide a combinatorial description of (the
object part) of these morphisms. In the second part of the section, we
will consider morphisms from $\Sta$ to $P(A^*)$, seen as a
$\Gamma$-algebra or a $\Delta$-algebra.

We start by describing the (object part of the) essentially 
unique continuous categorical $\Delta$-algebra morphism
$T_\Delta^\omega \to \Sta$. We will call the image $t'$ of a 
term tree $t$ under this morphism the 
\emph{synchronization tree denoted by} $t$,
or the \emph{synchronization tree associated with} $t$. 

Suppose that $t$ is a $\Delta$-term tree and $u$ and $v$ are 
leaves of $t$, labelled in $A \cup\{1\}$. 
Let $p$ (resp. $q$) denote the sequence of vertices 
along the unique path from the root to $u$ (resp. $v$),  
including $u$ (resp. $v$).
We say that $v \in S_t(u)$ if $p$ and $q$ are of the form $p= r w w_1 p'$
and $q = r w w_2 q'$ with $w$ labelled by $\cdot$ 
and successors $w_1,w_2$, ordered as indicated, such that
\begin{itemize} 
\item 
every vertex appearing in $p'$ which is different 
from $u$ is either labelled $+$, or if it is labelled $\cdot$ 
then its second successor belongs to $p'$, and 
\item 
every vertex appearing in $q'$ which is different 
from $v$ is either labelled $+$, or if it is labelled $\cdot$ 
then its first successor belongs to $q'$.
\end{itemize}  
We say that $u \in M(t)$ if each vertex appearing in $p$ 
which is different from $u$ is either labelled $+$, 
or if it is labelled $\cdot$, then its first successor
belongs to $p$.  
Finally, we say that $u\in E(t)$ if  
each vertex appearing in $p$ 
which is different from $u$ is either labelled $+$, 
or if it is labelled $\cdot$, then its second successor
appears in $p$. 
Note that if $u\in M(t)$ then there is no 
$w$ such that $u \in S_t(w)$, and similarly, 
if $u \in E(t)$ then $S_t(u) = \emptyset$.

Now form the edge-labelled directed graph $G(t)$ whose vertices are 
a new \emph{root vertex} $v_0$, the leaves of $t$ labelled in $A \cup \{1\}$, and 
the \emph{exit vertex} denoted $*$. The edges of $G(t)$ are the following: 
\begin{itemize}
\item
 For each $u \in M(t)$, there is an edge from the root $v_0$ to  
 $u$, labelled by the label of $u$ in $t$.
\item 
 For all leaf vertices $u$ and $v$ of $t$ labelled in $A \cup \{1\}$ 
 such that  $v \in S(u)$, there is an edge in $G(t)$ from $u$ to $v$
 whose label is the same as the label of $v$ in $t$.
\item Whenever $u$ 
 belongs to $E(t)$,  there is an edge in $G(t)$ labelled $\ex$ 
 from $u$ to $*$. 
\end{itemize} 
Note that $G(t)$ does not contain any cycle. We denote by $\tau(t)$
the synchronization tree obtained by unfolding $G(t)$ from $v_{0}$ and
then contracting edges labelled by $1$.

In Proposition~\ref{prop-tau} below, we will prove that $\tau(t) = t'$, 
the synchronization tree denoted by $t$.
However, before doing so, we find it instructive to 
give several examples. 

As a first example, consider the term
\[t = a \cdot ((1 + a) + ((1 + 1) \cdot b))\ .\]
As seen in Figure~\ref{fig:example-morphism-1}, the graph $G(t)$ contains 8 vertices, the root $v_0$, the 
exit vertex $*$, and vertices $v_1,\ldots,v_6$ 
corresponding to the $6$ leaves of $t$. 
There is an $a$-labelled edge from $v_0$ to $v_1$,
edges from $v_1$ to $v_2$, $v_4$ and $v_5$ labelled $1$, an edge from
$v_1$ to $v_3$ labelled $a$, and edges from $v_4$ and $v_5$ to $v_6$
labelled $b$. Finally, there is an edge labelled $\ex$ from 
each of $v_2$, $v_3$ and $v_6$ to $*$. 
It is clear that $\tau(t)$ is the synchronization tree denoted by $t$.

\begin{figure}
  \begin{center}
    \begin{tikzpicture}[transform shape,scale=1.0,edge from parent/.style={draw,-latex}]
 \tikzstyle{n}=[circle, draw=black, fill=black, inner sep = 0.25mm, outer sep= 0.5mm]
 \tikzstyle{m}=[inner sep = 0.25mm, outer sep= 0.5mm]
 \tikzstyle{e}=[ inner sep = 0mm, outer sep= 0mm]
 \tikzstyle{l}=[font=\tiny,label distance=-1mm]
 \tikzstyle{level}=[level distance=0.85cm,sibling distance=1cm]
 \tikzstyle{level 1}=[level distance=0.85cm,sibling distance=1.75cm]
 \tikzstyle{level 2}=[level distance=0.85cm,sibling distance=1.75cm]

\node[m,label={[l]above:{$(v0)$}}] {$\cdot$}
  child { node[m,label={[l]below:{$(v1)$}}] {$a$}}
  child { node[m] {$+$}
    child { node[m] {$+$}
      child { node[m,label={[l]below:{$(v2)$}}] {$1$}}
      child { node[m,label={[l]below:{$(v3)$}}] {$a$}} 
    }
    child { node[m] {$\cdot$}
      child { node[m] {$+$}
        child { node[m,label={[l]below:{$(v4)$}}] {$1$}}
        child { node[m,label={[l]below:{$(v5)$}}] {$1$}}
      }
      child { node[m,label={[l]below:{$(v6)$}}] {$b$}}
    }
  }
;

\begin{scope}[xshift=4.5cm]
\node[n,label={[l]above:{$(v0)$}}] (v0) at (0,0) {};
\node[n,label={[l]left:{$(v1)$}}] (v1) at (0,-0.85cm) {};
\node[n,label={[l]left:{$(v2)$}}] (v2) at (-1.2cm,-1.7cm) {};
\node[n,label={[l]left:{$(v3)$}}] (v3) at (-0.4cm,-1.7cm) {};
\node[n,label={[l]right:{$(v4)$}}] (v4) at (0.4cm,-1.7cm) {};
\node[n,label={[l]right:{$(v5)$}}] (v5) at (1.2cm,-1.7cm) {};
\node[n,label={[l]right:{$(v6)$}}] (v6) at (0.8cm,-2.55cm) {};
\node[n,label={[l]below:{$(*)$}}] (etoile) at (0,-3.4cm) {};

\draw (v0) edge[-latex] node[left] {$a$} (v1)
(v1) edge[-latex] node[above left] {$1$} (v2)
(v1) edge[-latex] node[left] {$a$} (v3)
(v1) edge[-latex] node[right] {$1$} (v4)
(v1) edge[-latex] node[above right] {$1$} (v5)
(v4) edge[-latex] node[left] {$b$} (v6)
(v5) edge[-latex] node[right] {$b$} (v6)
(v2) edge[-latex,bend right=20] node[left] {$\ex$} (etoile)
(v3) edge[-latex,bend right=10] node[left] {$\ex$} (etoile)
(v6) edge[-latex,bend left=10] node[below right] {$\ex$} (etoile);
\end{scope}

\begin{scope}[xshift=8cm]
   \tikzstyle{level}=[level distance=0.85cm,sibling distance=0.85cm]
 \tikzstyle{level 1}=[level distance=0.85cm,sibling distance=0.85cm]

\node[n] {}
  child { node[n] {}
  edge from parent node[draw=none,left] {$\ex$}}
  child { node[n] {}
    child { node[n] {}
    edge from parent node[draw=none,left] {$\ex$}}
  edge from parent node[draw=none,left] {$a$}}
  child { node[n] {}
    child { node[n] {}
    edge from parent node[draw=none,right] {$\ex$}}
  edge from parent node[draw=none,right] {$b$}}
  child { node[n] {}
    child { node[n] {}
    edge from parent node[draw=none,right] {$\ex$}}
  edge from parent node[draw=none,right] {$b$}};
\end{scope}

\end{tikzpicture}

  \end{center}
  \caption{The term tree $a \cdot ((1 + a) + ((1 + 1) \cdot b))$ (on the left), the associated graph $G(t)$ (in the middle) and the synchronization tree $\tau(t)$ (on the right).}
\label{fig:example-morphism-1}
\end{figure}
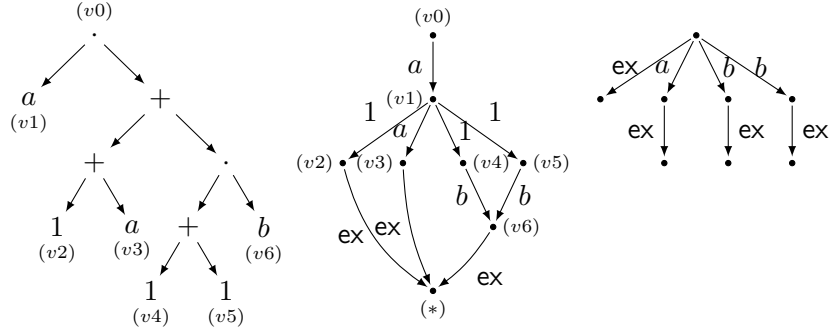

In the second example illustrated in Figure~\ref{fig:example-morphism-2}, consider the term tree $t$ defined by the scheme 
\begin{eqnarray*}
G &= &(a \cdot H) \cdot c + d\\
H &=& H \cdot 1\ .
\end{eqnarray*}
Now the root $v_{0}$ of $G(t)$ has two outgoing edges leading to different vertices $u_1$ and $u_2$, say. These edges are labelled $a$ and 
$d$, respectively. In addition, there is a sequence of vertices, 
call them $v_1,v_2,\ldots$, such that for each $i \geq 2$ there
is an edge from $v_{i+1}$ to $v_i$, which is labelled $1$, and there is 
a $c$-labelled edge from $v_2$ to $v_1$. Last, $*$ is a vertex,
and there exist edges from $u_2$ and $v_1$ to $*$ labelled $\ex$.  
$G(t)$ is infinite, but the tree $\tau(t)$ constructed from it 
is finite since $\tau(t) = (a\cdot 0)  + d$ is the synchronization 
tree denoted by $t$.

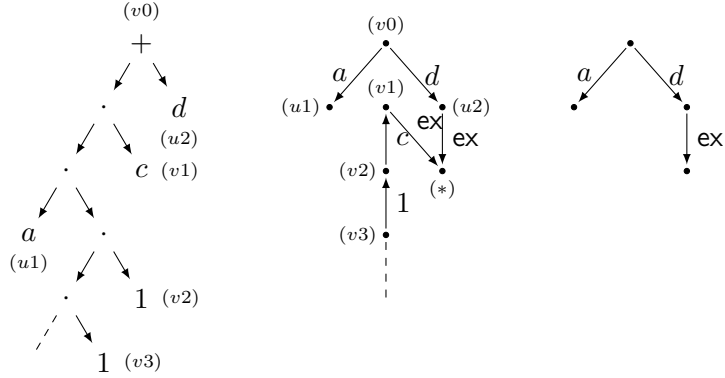
\begin{figure}
  \begin{center}
    \begin{tikzpicture}[transform shape,scale=1.0,edge from parent/.style={draw,-latex}]
 \tikzstyle{n}=[circle, draw=black, fill=black, inner sep = 0.25mm, outer sep= 0.5mm]
 \tikzstyle{m}=[inner sep = 0.25mm, outer sep= 0.5mm]
 \tikzstyle{e}=[ inner sep = 0mm, outer sep= 0mm]
 \tikzstyle{l}=[font=\tiny,label distance=-1mm]
 \tikzstyle{level}=[level distance=0.85cm,sibling distance=1cm]

\node[label={[l]above:{$(v0)$}}] {$+$}
  child { node {$\cdot$}
    child { node {$\cdot$}
      child { node[label={[l]below:{$(u1)$}}] {$a$}}
      child { node {$\cdot$}
        child { node {$\cdot$}
          child { node {} edge from parent[dashed,-]}
          child { node[label={[l]right:{$(v3)$}}] {$1$}}
        }
        child { node[label={[l]right:{$(v2)$}}] {$1$}}
      }
    }
    child { node[label={[l]right:{$(v1)$}}] {$c$}}
  }
  child { node[label={[l]below:{$(u2)$}}] {$d$}}
;

\begin{scope}[xshift=3.25cm]
\node[n,label={[l]above:{$(v0)$}}] (v0) at (0,0) {};
\node[n,label={[l]left:{$(u1)$}}] (u1) at (-0.75cm,-0.85cm) {};
\node[n,label={[l]right:{$(u2)$}}] (u2) at (0.75cm,-0.85cm) {};
\node[n,label={[l]above:{$(v1)$}}] (v1) at (0cm,-0.85cm) {};
\node[n,label={[l]left:{$(v2)$}}] (v2) at (0cm,-1.7cm) {};
\node[n,label={[l]left:{$(v3)$}}] (v3) at (0cm,-2.55cm) {};
\node[e] (v4) at (0cm,-3.4cm) {};
\node[n,label={[l]below:{$(*)$}}] (etoile) at (0.75cm,-1.7cm) {};

\draw (v0) edge[-latex] node[left] {$a$} (u1)
(v0) edge[-latex] node[right] {$d$} (u2)
(v2) edge[-latex] node[right] {$c$} (v1)
(v3) edge[-latex] node[right] {$1$} (v2)
(v4) edge[dashed] (v3)
(v1) edge[-latex] node[above right,xshift=-1mm] {$\ex$} (etoile)
(u2) edge[-latex] node[right] {$\ex$} (etoile);

\end{scope}

\begin{scope}[xshift=6.5cm]
   \tikzstyle{level}=[level distance=0.85cm,sibling distance=1.5cm]

\node[n] {}
  child { node[n] {}
  edge from parent node[draw=none,left] {$a$}}
  child { node[n] {}
    child { node[n] {}
    edge from parent node[draw=none,right] {$\ex$}}
  edge from parent node[draw=none,right] {$d$}};

\end{scope}

\end{tikzpicture}

  \end{center}
  \caption{A term tree (on the left) and the associated graph $G(t)$ (on the right).}
\label{fig:example-morphism-2}
\end{figure}

In the last example illustrated in Figure~\ref{fig:example-morphism-3}, consider the term $t = (1\cdot (0 \cdot b))\cdot 1$.
The vertices of $G(t)$ are the root $v_0$, vertices $v_1,v_2,v_3$ 
corresponding respectively to the leaves of $t$ with nonzero label, and the exit vertex $*$.  There are $3$ edges, an edge labelled $1$ from $v_0$ to $v_1$,
an edge labelled $1$ from $v_2$ to $v_3$,  and an edge labelled $\ex$ from $v_3$ to $*$.
Now $\tau(t)$ contains only the root and no edges, so that $\tau(t)$ 
is the synchronization tree $0$ denoted by $t$. 

\begin{figure}
  \begin{center}
    \begin{tikzpicture}[transform shape,scale=1.0,edge from parent/.style={draw,-latex}]
 \tikzstyle{n}=[circle, draw=black, fill=black, inner sep = 0.25mm, outer sep= 0.5mm]
 \tikzstyle{m}=[inner sep = 0.25mm, outer sep= 0.5mm]
 \tikzstyle{e}=[ inner sep = 0mm, outer sep= 0mm]
 \tikzstyle{l}=[font=\tiny,label distance=-1mm]
 \tikzstyle{level}=[level distance=0.85cm,sibling distance=1cm]

\node[label={[l]above:{$(v0)$}}] {$\cdot$}
  child { node[] {$\cdot$}
    child { node[label={[l]left:{$(v1)$}}] {$1$}}
    child { node[] {$\cdot$}
      child { node[] {$0$}}
      child { node[label={[l]below:{$(v2)$}}] {$b$}}
    }
  }
  child { node[label={[l]right:{$(v3)$}}] {$1$}}
;

\begin{scope}[xshift=6cm]
\node[n,label={[l]above:{$(v0)$}}] (v0) at (0,0) {};
\node[n,label={[l]left:{$(v1)$}}] (v1) at (-1.35cm,-0.85cm) {};
\node[n,label={[l]below:{$(v2)$}}] (v2) at (-0.45cm,-0.85cm) {};
\node[n,label={[l]below:{$(v3)$}}] (v3) at (0.45cm,-0.85cm) {};
\node[n,label={[l]right:{$(*)$}}] (etoile) at (1.35cm,-0.85cm) {};

\draw (v0) edge[-latex] node[above left] {$1$} (v1)
(v2) edge[-latex] node[above] {$b$} (v3)  
(v3) edge[-latex] node[above] {$\ex$} (etoile);
\end{scope}

\end{tikzpicture}

  \end{center}
  \caption{The term tree of $t=(1\cdot (0 \cdot b))\cdot 1$ (on the left) and the associated graph $G(t)$ (on the right).}
\label{fig:example-morphism-3}
\end{figure}
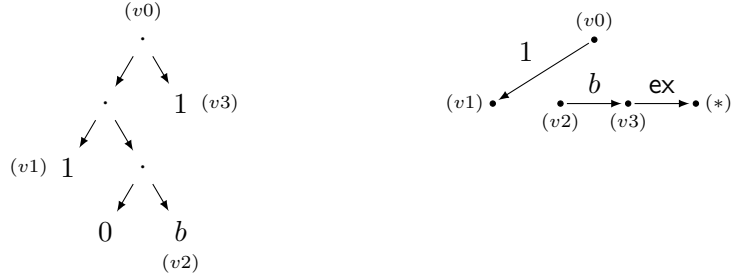

We still need to prove that the definition of $\tau(t)$ is correct.

\begin{prop}
\label{prop-tau}
\label{prop:constructing-delta-trees}
For every $\Delta$-term tree $t$, the image $t'$ of $t$ with respect to the 
essentially unique continuous $\Delta$-algebra morphism 
$T_\Delta^\omega \to \Sta$ is $\tau(t)$. 
\end{prop}

\begin{proof}
Suppose first that $t$ is finite. We will prove the claim by induction on the 
structure of $t$. If $t = 0$, then $G(t)$ has two vertices, the root and the exit 
vertex, and no edges. If $t$ is a symbol in $A \cup\{1\}$, then 
$G(t)$ has three vertices, the root $v_0$, a vertex $v_1$ and the vertex $*$.
There is an edge from $v_0$ to $v_1$ and an edge from $v_1$ to $*$. 
The first edge is labelled $a$ if $t = a \in A$, and $1$ if $t = 1$.
The second edge is labelled $\ex$. In either case, $\tau(t)= t'$. 

In the induction step, first suppose that $t = t_1 + t_2$ for some terms $t_1,t_2$
denoting the synchronization trees $t_1'$ and $t_2'$, respectively. 
Then $G(t)$ is isomorphic to the disjoint union of $G(t_1)$ and $G(t_2)$ with roots 
and exit vertices merged.  Thus, $\tau(t) = \tau(t_1) + \tau(t_2) 
= t_1' + t_2' = t'$.

Suppose next that $t = t_1 \cdot t_2$ with $t_1$ denoting $t_1'$ and 
$t_2$ denoting $t_2'$. Then $G(t)$ can be constructed from $G(t_1)$ and $G(t_2)$ 
as follows. For all edges from the root $v_2$ of $G(t_2)$ to a vertex 
$v$ of $G(t_2)$ (which necessarily corresponds to a leaf  vertex of $t_2$ 
in $M(t_2)$),  and for all $u\in E(t_1)$,
add a new edge from $u$ to $v$ labelled by the symbol
which is the label of the edge from $v_2$ to $v$ in $G(t_2)$  (i.e.,  
the label of $v$  in $t_2$). 
Then remove the edge from $u$ to the exit vertex of $G(t_1)$.
Finally remove all edges originating in the root of $G(t_2)$. 
The vertices of $G(t)$ 
are the non-exit vertices of $G(t_1)$ and the non-root vertices of $G(t_2)$. 
It should be clear that $\tau(t)$ is the sequential product of 
$\tau(t_1) \cdot \tau(t_2)$ and thus $\tau(t) = t_1' \cdot t_2' = 
t'$ by the induction hypothesis.

Suppose now that $t$ is infinite. For each $n\geq 0$, let $t_n$ denote the 
approximation of $t$ obtained by relabelling each vertex of $t$ of depth $n$
by $0$ and removing all vertices of depth greater than $n$. For each 
leaf vertex $u$ of $t$ there is some $n_0$ such that $u$ is a vertex of $t_n$
with the same label for all $n \geq n_0$. Moreover, for any two leaf vertices 
$u,v$ of $t$ with a nonzero label, $v \in S_t(u)$ iff there is some $n_0$
such that $v \in S_{t_n}(u)$ for all $n \geq n_0$. Similarly, for 
each leaf $u$ of $t$ with a nonzero label, we have $u \in M(t)$ 
($u \in E(t)$) iff there is some $n_0$ such that $u \in M(t_n)$ ($u \in E(t_n)$, resp.) for all $n \geq n_0$. This implies that $G(t)$ is the union (colimit) 
of the $G(t_n)$ and then it follows that $\tau(t)$ is also the union (colimit) 
of the $\tau(t_n)$. We conclude that
$t' = \Colim\ t'_n = \Colim\ \tau(t_n) = \tau(t)$, where for each $n$, $t'_n$ is the 
synchronization tree denoted by $t_n$.
\end{proof}

Next we describe the essentially unique morphism of continuous categorical 
 $\Gamma$-algebras $T_\Gamma^\omega \to \Sta$. Fix a $\Gamma$-term tree $t\in T_\Gamma^\omega$.  We define a synchronization tree, denoted $H(t)$, which is essentially a representation of $t$ as a synchronization tree in which leaves labelled by $1$ (in $t$) are added a dangling outgoing edge labelled by $\ex$.

The vertices of the synchronization tree $H(t)$ are the vertices of 
$t$ together with a fresh vertex $u'$ for each leaf $u$ labelled by $1$.
If $u$ is labelled by $+$ in $t$ then there is in $H(t)$ an edge labelled 
by $+_{1}$ from $u$ to its first successor and an edge labelled 
by $+_{2}$ from $u$ to its second successor. If $u$ is labelled by $a \in A$ in $t$, then there is in $H(t)$ an edge labelled by $a$ from $u$ to its unique successor. Finally, 
for each leaf of $u$ labelled by $1$, there is an edge in $H(t)$ from $u$ to $u'$ labelled by $\ex$. The synchronization tree  defined by $t$ is the $\{+_{1},+_{2}\}$-contraction of $H(t)$. This construction is illustrated 
for the $\Gamma$-term $a.(c.0+d.1)+1$ in Figure~\ref{fig:example-morphism-4}.

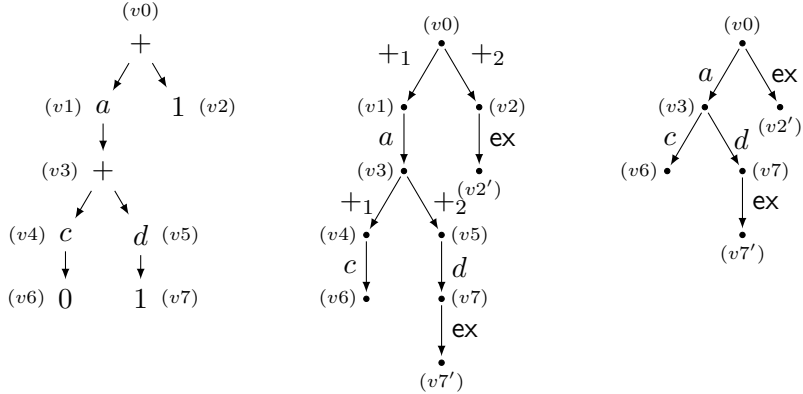
\begin{figure}
  \begin{center}
    \begin{tikzpicture}[transform shape,scale=1.0,edge from parent/.style={draw,-latex}]
 \tikzstyle{n}=[circle, draw=black, fill=black, inner sep = 0.25mm, outer sep= 0.5mm]
 \tikzstyle{m}=[inner sep = 0.25mm, outer sep= 0.5mm]
 \tikzstyle{e}=[ inner sep = 0mm, outer sep= 0mm]
 \tikzstyle{l}=[font=\tiny,label distance=-1mm]
 \tikzstyle{level}=[level distance=0.85cm,sibling distance=1cm]

\node[label={[l]above:{$(v0)$}}] {$+$}
  child { node[label={[l]left:{$(v1)$}}] {$a$}
    child { node[label={[l]left:{$(v3)$}}] {$+$}
      child { node[label={[l]left:{$(v4)$}}] {$c$}
        child { node[label={[l]left:{$(v6)$}}] {$0$}}
      }
      child { node[label={[l]right:{$(v5)$}}] {$d$}
        child { node[label={[l]right:{$(v7)$}}] {$1$}}
      }
    }
  }
  child { node[label={[l]right:{$(v2)$}}] {$1$}}
;

\node[n,label={[l]above:{$(v0)$}},xshift=4cm] {}
  child { node[n,label={[l]left:{$(v1)$}}] {}
    child { node[n,label={[l]left:{$(v3)$}}] {}
      child { node[n,label={[l]left:{$(v4)$}}] {}
        child { node[n,label={[l]left:{$(v6)$}}] {}
        edge from parent node[draw=none,left] {$c$}}
      edge from parent node[draw=none,left] {$+_1$}}
      child { node[n,label={[l]right:{$(v5)$}}] {}
        child { node[n,label={[l]right:{$(v7)$}}] {}
          child { node[n,label={[l]below:{$(v7')$}}] {}
          edge from parent node[draw=none,right] {$\ex$}}
        edge from parent node[draw=none,right] {$d$}}
      edge from parent node[draw=none,right] {$+_2$}}
    edge from parent node[draw=none,left] {$a$}}
  edge from parent node[draw=none,above left] {$+_1$}}
  child { node[n,label={[l]right:{$(v2)$}}] {}
    child { node[n,label={[l]below:{$(v2')$}}] {}
    edge from parent node[draw=none,right] {$\ex$}}
  edge from parent node[draw=none,above right] {$+_2$}}
;

\node[n,label={[l]above:{$(v0)$}},xshift=8cm] {}
  child { node[n,label={[l]left:{$(v3)$}}] {}
    child { node[n,label={[l]left:{$(v6)$}}] {}
    edge from parent node[draw=none,left] {$c$}}
    child { node[n,label={[l]right:{$(v7)$}}] {}
      child { node[n,label={[l]below:{$(v7')$}}] {}
      edge from parent node[draw=none,right] {$\ex$}}
    edge from parent node[draw=none,right] {$d$}}
  edge from parent node[draw=none,left] {$a$}}
  child { node[n,label={[l]below:{$(v2')$}}] {}
  edge from parent node[draw=none,right] {$\ex$}}
;

\end{tikzpicture}

  \end{center}
  \caption{The term tree of $t=a.(c.0+d.1)+1$ (on the left) and the synchronization tree $H(t)$ (in the middle) and its $\{+_{1},+_{2}\}$-contraction (on the right).}
\label{fig:example-morphism-4}
\end{figure}

Let $\tau'$ denote the function that maps $t$ to the $\{+_{1},+_{2}\}$-contraction of $H(t)$.

\begin{prop}
\label{prop-Gamma-hom}
For every $t \in T_\Gamma^\omega$, $\tau'(t)$ is the synchronization tree 
in $\Sta$ denoted by $t$. 
\end{prop}

We omit the proof which is essentially a simplification of that of 
Proposition~\ref{prop-tau}.

\begin{prop}
The function which maps a synchronization tree $t \in \Sta$ to its 
path language in $P(A^*)$ is the (object part) of a categorical
$\Gamma$-algebra, as well as $\Delta$-algebra, morphism $\Sta \to P(A^*)$. 
\end{prop} 

\begin{proof}
It is clear that the empty tree is mapped
to the empty language and the operations are preserved. When there is a morphism 
$t \to t'$ for synchronization trees $t,t' \in \Sta$, then the path language 
of $t$ is included in the path language of $t'$. Finally, suppose that $(\phi_n: t_n \to t_{n+1})_n$ is 
an $\omega$-diagram in $\Sta$ with colimit $(\psi_n : t_n \to t)_n$. 
Then for every branch of $t$ ending in an edge labelled $\ex$ 
there is some $n_0$ such that the branch is the image of a corresponding
branch of $t_{n_0}$ with respect to the morphism $\psi_{n_0} : t_{n_0} \to t$,
and then the same holds for each $n \geq n_0$. Using this fact, 
it follows easily that the path language 
of $t$ is the union of the path languages of the $t_n$, 
proving that colimits of $\omega$-diagrams are preserved. 
\end{proof}

\subsection{Basic Properties}\label{Sect:closure}

In this section, we give some basic properties of synchronization
trees defined by first-order recursion schemes.

First, we remark that $\Gamma$-regular and $\Gamma$-algebraic functors are closed under the sequential product. This closure property is immediate for $\Delta$-regular and $\Delta$-algebraic functors as the sequential product is an operation
of the continuous $\Delta$-algebra $\Sta$ (which is not the case for the continuous $\Gamma$-algebra $\Sta$).

\begin{prop}
If $f,f':\Sta^k \to \Sta$ are $\Gamma$-algebraic (resp. $\Gamma$-regular), then so is $f \cdot f'$. 
\end{prop}
\begin{proof}
Suppose that $f$ and $f'$ are the first components 
of the initial solutions over $\Sta$ 
of the $\Gamma$-recursion schemes $E$ and $E'$.  Without loss of
generality, we may assume that $E$ and $E'$ have disjoint sets of
functor variables.  Let $F_1(x_1,\ldots,x_{k})$ and
$F'_1(x_1,\ldots,x_{k})$ denote the left-hand sides of the first equations
of $E$ and $E'$. Then replace each occurrence of the symbol $1$ on the
right-hand-side term of each equation of $E$ with $F_1'(x_1,\ldots,x_k)$,
and consider the recursion scheme consisting of these equations
together with the equations of $E'$. The component of the initial
solution of this scheme which corresponds to $F_1$ is $f \cdot
f'$.   
\end{proof}

 \begin{exa}\label{Ex:seqcomp}
The tree $b^\omega$ determined by single infinite branch with edges labelled $b$ is $\Gamma$-regular since it is the initial
solution of the following $\Gamma$-regular recursion scheme
\begin{equation*}
X = b.X .  
\end{equation*}
We have already seen in Example~\ref{Ex:anyfinitedepth}
on page \pageref{Ex:anyfinitedepth} 
that the tree $\sum_{i \geq 1} a^{i}$ is $\Gamma$-algebraic. According to the above proposition, the tree $(\sum_{i\geq 1} a^i)\cdot b^\omega$ is also $\Gamma$-algebraic. Indeed, it is the initial
solution of the recursion scheme
\begin{eqnarray*}
F_1 & = & F_2(a.X) \\
F_2(v)  & = & v + F_2(a.v) \\
X & = & b.X .
\end{eqnarray*}
\end{exa}

All our classes of synchronization trees are also closed under the contraction operation for each non-empty
subset of $A$.

\begin{prop}
\label{prop:contr-alg}
  The classes of $\Gamma$-regular, $\Gamma$-algebraic, $\Delta$-regular and
  $\Delta$-algebraic synchronization trees in $\Sta$ are closed under contraction for any  non-empty subset of $A$.
\end{prop}

\begin{proof}
Let $B$ be a non-empty subset of $A$.  Let $t$ be a synchronization tree defined by an algebraic  $\Gamma$-scheme $E$  and let $t'$ be the synchronization tree obtained  by contracting $t$ with respect to $B$. 

 A $\Gamma$-scheme $E'$ defining $t'$ is easily obtained from $E$ by replacing
 each equation of the form 
  $F(v_{1} \ldots v_{n}) = t$ by the equation  $F(v_{1} \ldots v_{n}) = \overline{t}$ 
  where $\overline{t}$ is inductively defined by $\overline{0}=0$,
 $\overline{1}=1$, $\overline{t_{1} + t_{2}}=\overline{t_{1}} + \overline{t_{2}}$, 
 $\overline{G (t_{1}, \ldots ,t_{k})}=G (\overline{t_{1}}, \ldots, \overline{t_{k}})$, $\overline{a(t)}=a(\overline{t})$ if $a \in A \setminus B$
 and $\overline{b(t)}=\overline{t}$ for $b \in B$. Note that $E'$ is $\Gamma$-regular, if so is $E$.

 The $\Gamma$-term defined by $E'$ is the contraction of the $\Gamma$-term
 defined by $E$ from its root with respect to $B$. It then follows, from the definition of the mapping $\tau'$ in \ref{sec:morphism} and Proposition~\ref{prop-Gamma-hom}, that the synchronization tree defined by $E'$ is the contraction of the synchronization
  tree defined by $E$ with respect to $B$. 

Let $t$ be a synchronization tree defined by an algebraic  $\Delta$-scheme $E$  and let $t'$ be the synchronization tree obtained  by contracting $t$ with respect to $B$. Consider the $\Delta$-scheme $E'$ obtained by replacing
in $E$ all occurrences of a constant in $B$ by the constant $1$.  (Note that $E'$ is $\Delta$-regular, if so is $E$.) The $\Delta$-term defined by $E'$ is the $\Delta$-term defined by $E$ in which
each occurrence of a constant in $B$ is replaced by the constant $1$. It follows
from the definition of the mapping $\tau$ in \ref{sec:morphism} and Proposition~\ref{prop-tau}, that the synchronization tree defined by $E'$ is the contraction of the synchronization tree defined by $E$ with respect to $B$. 
\end{proof}

Finally, the $\Gamma$-regular synchronizations trees can be characterized by a syntactic subfamily of the $\Delta$-regular recursion schemes.  This
characterization is similar in spirit to the characterization of regular
languages by right-linear context-free grammars. A $\Delta$-regular
scheme is said to be \emph{right-linear} if the right-hand side of
each equation is of the form
 \[t_0 + \sum_{i \in [n]}t_i\cdot G_i\]
 up to commutativity and associativity of sum, where the $G_i$, $i \in
 [n]$, are (constant) functor variables and the $t_j$, $j \in
 \{0,\ldots,n\}$ are terms over the signature $\Delta$ not containing
 variables. (The empty sum stands for $0$.)  It is rather standard to
 prove the following fact:

\begin{prop}
\label{prop:right-linear-implies-regular}
A synchronization tree in $\Sta$ is $\Gamma$-regular iff it can be defined
by a right-linear $\Delta$-regular scheme.
\end{prop}

\begin{proof}
It is clear how to transform a $\Gamma$-regular recursion scheme into a
right-linear $\Delta$-scheme by turning each prefixing operation into 
a sequential product. 

We show how to transform a $\Delta$-term $t\cdot G$ with $t$ containing no 
functor variables into a corresponding $\Gamma$-term $t'(G)$, possibly containing $G$. 
We proceed by induction on the structure of $t$. When $t = 0$, let 
$t'(G) = 0$, and when $t = 1$ let $t'(G) = G$. When $t = a$ where $a \in A$,
define $t'(G) = a.G$. Suppose now that $t = t_1 + t_2$. Then let 
$t' = t_1'(G)+ t_2'(G)$. Finally, consider the case when $t = t_1\cdot t_2$.
In this case define $t'(G)$ as the term obtained by substituting 
$t_2'(G)$ for each occurrence of $G$ in $t_1'(G)$. 

A routine calculation shows that for each evaluation of $G$ in $\Sta$, the terms
$t\cdot G$ and $t'(G)$ yield the same synchronization tree. {For an example of such calculation, we refer the reader to the proof of Theorem~\ref{Thm:Delta2Gamma} on page \pageref{Thm:Delta2Gamma}.}
Using this fact, we may transform a right-linear $\Delta$-scheme into
a regular $\Gamma$-scheme by changing the right-hand side 
$t_0 + \sum_{i \in [n]}t_i\cdot G_i$ of each equation to $t_0'(1) + 
\sum_{i \in [n]}t_i'( G_i)$, 
where $t_0'(1)$ is the term obtained 
by substituting $1$ for $G$ in $t_0'(G)$.
\end{proof}

\section{Comparison between the $\Gamma$-algebra and the $\Delta$-algebra}\label{Sect:comparison}

In this section, we interpret recursion schemes over the 
continuous categorical algebra $\Sta$, viewed either as 
a $\Gamma$-algebra or a $\Delta$-algebra. We compare the resulting classes of synchronization trees with respect to language equivalence, bisimulation equivalence and isomorphism equivalence.

First for language equivalence, we show in Section~\ref{ssec:up-to-languages} that the following hierarchy holds.

\begin{equation}
\label{eq:hierarchy-languages}
\underbrace{\Gamma\textrm{-regular}}_{\textrm{regular languages}}
 \;\;\subsetneq\;\; 
\underbrace{\Delta\textrm{-regular} \;\;=\;\; \Gamma\textrm{-algebraic}}_{\textrm{context-free languages}}
 \;\;\subsetneq\;\; \underbrace{\Delta\textrm{-algebraic}}_{\textrm{indexed languages}} 
\end{equation}

Up to bisimulation or isomorphism, we show in Section~\ref{ssec:up-to-bisimulation-iso} that the following hierarchy holds.

\begin{equation}
\label{eq:hierarchy-bisim}
 \Gamma \textrm{-regular}    \subsetneq    \Delta\textrm{-regular}   \subsetneq   \Gamma\textrm{-algebraic}   \subsetneq   \Delta\textrm{-algebraic}                
\end{equation}

We conclude the section by a comparison with the classes of synchronization trees defined by BPA and BPP.

\subsection{Comparison up to language equivalence}
\label{ssec:up-to-languages}

As already mentioned in the introduction, the path languages of the different
classes can be characterized as follows.
\begin{prop}
  The following properties holds.
  \begin{enumerate}
  \item The path languages of the $\Gamma$-regular trees are the regular languages.
  \item The path languages of the $\Delta$-regular trees and of the $\Gamma$-algebraic trees are the context-free languages.
  \item The path languages of the $\Delta$-algebraic languages are 
the indexed languages.
  \end{enumerate}
\end{prop}

\begin{proof}
By the Mezei-Wright theorem, the path languages of the $\Gamma$-regular trees and the 
$\Delta$-regular trees in $\Sta$ are just the $\Gamma$-regular and $\Delta$-regular elements 
(objects) of $P(A^*)$, seen as a continuous $\Gamma$-algebra or $\Delta$-algebra. By classic results {(e.g. \cite[Theorem 1.21, page 116]{MandrioliG88})}, these in turn are the regular and context-free languages  Similarly, the indexed languages (or OI-macro languages) are exactly the $\Delta$-algebraic elements of $P(A^*)$, 
cf. \cite{Fischer}, i.e., the path languages of the $\Delta$-algebraic trees by 
the Mezei-Wright theorem. 
By Theorem~\ref{Thm:Delta2Gamma}, every $\Delta$-regular tree is $\Gamma$-algebraic. 
Thus, to complete the proof, it remains to show that the path language of a 
$\Gamma$-algebraic tree is context-free.

Suppose that $L$ is the path language of a $\Gamma$-algebraic tree in $\Sta$, 
and let $E$ denote a $\Gamma$-algebraic scheme defining it. Then consider the 
term tree $t\in T_\Gamma^\omega$ defined by $E$. By the Mezei-Wright theorem, 
$L$ is the image of $t$ with respect to the unique continuous $\Gamma$-algebra 
morphism $h: T_\Gamma^\omega \to P(A^*)$.

Suppose that $v$ is a leaf of $t$ labelled $1$, and consider the branch 
of $t$ from the root to $v$. The label of this branch is a word $w_v$ over the 
alphabet $A \cup \{+_1,+_2\}$. Let us consider the image $\pi(w_v)$ of 
$w_v$ under the {erasing morphism} $\pi: (A \cup \{+_1,+_2\})^* \to A^*$ 
that removes the letters $+_1$ and $+_2$. Then $L = h(t) \subseteq A^*$ is the set of all such words $\pi(w_v)$ obtained by considering all leaves 
$v$ of $t$ labelled $1$. 

It follows from Courcelle's characterization of the algebraic 
term trees by deterministic context-free languages \cite[Theorem 5.5.1,page 157]{Courcelle83} (see also \cite{Courcelle78a,Courcelle78b})
that the set of all words $w_v$, where $v$ is a 
leaf of $t$ labelled $1$, is a deterministic context-free language. Since the 
image of a (deterministic) context-free language with respect to a homomorphism
is context-free, we conclude that $L$ is a context-free language.
\end{proof}

\subsection{Comparison up to bisimulation and isomorphism}
\label{ssec:up-to-bisimulation-iso}

The aim of this section is to establish the strict inclusions stated
in Equation~\eqref{eq:hierarchy-bisim}.

As noted in Remark~\ref{rem:delta-implies-gamma} on page~\pageref{rem:delta-implies-gamma}, a $\Gamma$-term can be transformed into an equivalent $\Delta$-term
by replacing for each letter $a \in A$, all occurrences of 
a subterm $a(t)$ by  $a \cdot t$. 
In particular, every $\Gamma$-regular tree is $\Delta$-regular and every $\Gamma$-algebraic tree is $\Delta$-algebraic. This establishes the first and
third inclusions of \eqref{eq:hierarchy-bisim}. These inclusions are strict up to bisimulation and up to isomorphism as they are already strict with respect to language equivalence (cf. \eqref{eq:hierarchy-languages}).

It only remains to establish the second strict inclusion of Equation~\eqref{eq:hierarchy-bisim}. First we  establish the inclusion up to isomorphism in Theorem~\ref{Thm:Delta2Gamma}. In Corollary~\ref{cor:charac-delta-regular-in-gamma-algebraic}, we characterize a syntactical subfamily of the $\Gamma$-algebraic schemes which captures exactly the $\Delta$-regular schemes.

\begin{thm}\label{Thm:Delta2Gamma} 
Every $\Delta$-regular tree is $\Gamma$-algebraic. 
\end{thm}
\begin{proof}
Consider a regular $\Delta$-recursion scheme $E$,
\begin{eqnarray*}
F_1 &=& t_1\\
&\vdots &\\
F_n &=& t_n , 
\end{eqnarray*}
which defines the $\Delta$-regular tree $s \in \Sta$.

Let $\Phi = \{F_1,\ldots,F_n\}$ and 
$\Psi = \{G_0,G_1,\ldots,G_n\}$, where 
each $F_i$ is of rank $0$, $G_0$ is of rank $0$, 
and each $G_i$ with $i \geq 1$ is of rank $1$.

Given a variable-free $\Delta \cup \Phi$-term $t$,
we define its translation $t'$ to be a $\Gamma\cup (\Psi - \{G_0\})$-term 
in the variable $v_1$. 
\begin{itemize}
\item If $t = F_i$, for some $i \in [n]$,  then $t' = G_i(v_1)$.
\item If $t = 0$ then $t' = 0$. 
\item If $t = 1$ then $t' = v_1$. 
\item If $t = a$ then $t' = a.v_1$.
\item If $t = t_1 + t_2$ then $t' = t_1' + t_2'$. 
\item If $t= t_1\cdot t_2$ then $t' = t_1'(t_2')$, the term obtained 
by substituting $t_2'$ for each occurrence of $v_1$ in $t_1'$. 
\end{itemize}
Note that 
\[t^{\Sta} : \Sta^n \to \Sta\]
and 
\[t'^{\Sta} : [\Sta \to \Sta]^n \to [\Sta \to \Sta]\]
The two functors are related.

For a tree $r\in \Sta$, let $\overline{r}$ denote the functor 
`left composition with $r$', $r \cdot (-)$ in $[\Sta \to \Sta]$.

\emph{Claim 1.} For each variable-free $\Delta\cup \Phi$-term $t$ and its 
translation $t'$, and for each sequence of trees $r_1,\ldots,r_n$ 
in $\Sta$, it holds that 
\begin{eqnarray*}
\overline{t^{\Sta}(r_1,\ldots,r_n)} 
&=& 
t'^{\Sta}(\overline{r}_1,\ldots,\overline{r}_n).
\end{eqnarray*}

Indeed, when $t= F_i$, for some $i \in [n]$, then $t'$ is $G_i(v_1)$ and both sides are equal to 
the functor $\overline{r}_i$. If $t= 0$, then both sides are equal to the constant
functor $\Sta \to \Sta$ determined by the tree $0^{\Sta}$, and when $t = 1$, 
both sides are equal to the identity functor $\Sta \to \Sta$. Indeed, 
seen $t^{\Sta}(r_1,\ldots,r_n) = 1^{\Sta}$, left composition with $t^{\Sta}(r_1,\ldots,r_n)$
is the identity functor as is $t'^{\Sta}(\overline{r}_1,\ldots,\overline{r}_n) = v_1^{\Sta}(\overline{r}_1,\ldots,\overline{r}_n)$. 
Suppose now that $t = a$ for some $a$.
Then both sides are equal to the functor $\overline{a}$, left composition with $a$.
Next let $t = t_1 + t_2$, and suppose that the claim holds for $t_1$ and $t_2$.  Then 
\begin{eqnarray*}
\overline{t^{\Sta}(r_1,\ldots,r_n)} 
&=& \overline{t_1^{\Sta}(r_1,\ldots,r_n)} + \overline{t_2^{\Sta}(r_1,\ldots,r_n)} \\
&=& t_1'^{\Sta}(\overline{r}_1,\ldots,\overline{r}_n) + t_2'^{\Sta}(\overline{r}_1,\ldots,\overline{r}_n)\\
&=& t'^{\Sta}(\overline{r}_1,\ldots,\overline{r}_n).
\end{eqnarray*} 
Last, suppose that $t= t_1\cdot t_2$, and that the claim holds for both terms $t_1$ and $t_2$. 
Then 
\[\overline{t^{\Sta}(r_1,\ldots,r_n)} = \overline{t_1^{\Sta}(r_1,\ldots,r_n)}\circ \overline{t_2^{\Sta}(r_1,\ldots,r_n)}\]
is the composition of the functors $\overline{t_1^{\Sta}(r_1,\ldots,r_n)}$ and $\overline{t_2^{\Sta}(r_1,\ldots,r_n)}$
(where the second functor is applied first), as is the functor 
\begin{eqnarray*}
t'^{\Sta}(\overline{r}_1,\ldots,\overline{r}_n) 
&=& t_1'^{\Sta}(\overline{r}_1,\ldots,\overline{r}_n)
\circ t_2'^{\Sta}(\overline{r}_1,\ldots,\overline{r}_n)\\
&=& \overline{t_1^{\Sta}(r_1,\ldots,r_n)}\circ \overline{t_2^{\Sta}(r_1,\ldots,r_n)}.
\end{eqnarray*}

\emph{Claim 2.} For each variable-free $\Delta\cup \Phi$-term $t$ and its 
translation $t'$, and for each sequence of trees $r_1,\ldots,r_n$ 
in $\Sta$, it holds that 
\begin{eqnarray*}
t^{\Sta}(r_1,\ldots,r_n) 
&=& 
(t'^{\Sta}(\overline{r}_1,\ldots,\overline{r}_n))(1^{\Sta})
\end{eqnarray*}

Indeed, by Claim 1, we have 
\[
t^{\Sta}(r_1,\ldots,r_n) = (\overline{t^{\Sta}(r_1,\ldots,r_n)})(1^{\Sta}) = 
(t'^{\Sta}(\overline{r}_1,\ldots,\overline{r}_n))(1^{\Sta}).
\]
Let $E'$ denote the $\Gamma$-algebraic scheme
\begin{eqnarray*}
G_0 &=& G_1(1)\\
G_1(v_1) &=& t_1'\\
&\vdots &\\
G_n(v_1) &=& t_n'
\end{eqnarray*}
We  claim that $E'$ is equivalent to $E$, i.e., 
$E'$ also defines $s$. To this end, let $G$ denote the 
recursion scheme consisting of the last $n$ equations of $E'$. 
In order to compare the schemes $E$ and $G$, define 
\begin{eqnarray*}
  s_0 &=&  (s_{0,1},\ldots,s_{0,n}) = (0^{\Sta},\ldots,0^{\Sta})\\
  s_{i+1} &=& (s_{i+1,1},\ldots,s_{i+1,n}) = (t_1^{\Sta}(s_i),\ldots, t_n^{\Sta}(s_i))\\
  g_0 &=& (g_{0,1},\ldots,g_{0,n}) = (0^{[\Sta \to \Sta]},\ldots,0^{[\Sta \to \Sta]})\\
  g_{i+1} &=& (g_{i+1,1},\ldots,g_{i+1,n}) = (t_1'^{\Sta}(g_i),\ldots, t_n'^{\Sta}(g_i)) 
\end{eqnarray*} 
For each $i$, define $\overline{s}_i = (\overline{s_{i,1}},\ldots,\overline{s_{i,n}})$. 
We prove by induction on $i$ that $g_i = \overline{s_i}$. 

This is clear when $i = 0$, since for each $j \in [n]$, $g_{0,j} = 0^{[\Sta \to \Sta]} =
\overline{0^{\Sta}} = \overline{s_{0,j}}$. To prove the induction step, 
suppose that we have established our claim for some $i\geq 0$.
Then for all $j\in [n]$, 
\begin{eqnarray*}
g_{i+1,j} &=& t_j'^{\Sta}(g_i)\\
&=& t_j'^{\Sta}(\overline{s_i})\\
&=& \overline{t_j^{\Sta}(s_i)},\quad {\rm by}\ {\rm Claim}\ {\rm 1,}\\
&=& \overline{s_{i+1,j}}
\end{eqnarray*} 
For each $j\in [n]$, let $\phi_{0,j}$ denote the unique morphism
$0^{\Sta} \to s_{1,j}$. 
Then define 
\begin{eqnarray*}
  \phi_0 &=& (\phi_{0,1},\ldots,\phi_{0,n}) \\
  \phi_{i+1} &=& (\phi_{i+1,1},\ldots,\phi_{i+1,n}) = (t_1^{\Sta}(\phi_i),\ldots,t_n^{\Sta}(\phi_i))
\end{eqnarray*}   
Next,  let $\psi_{0,j}$ denote the unique natural transformation 
$0^{[\Sta \to \Sta]} \to g_{1,j}$, for each $j\in [n]$.
Define
\begin{eqnarray*}
\psi_0 &=& (\psi_{0,1},\ldots,\psi_{0,n})\\
\psi_{i+1} &=& (\psi_{i+1,1},\ldots,\psi_{i+1,n}) =
 (t_1'^{[\Sta \to \Sta]}(\psi_i),\ldots, t_n'^{[\Sta \to \Sta]}(\psi_i))
\end{eqnarray*}
Thus, each $\psi_{i,j}$ is a natural transformation from 
$g_{i,j} = \overline{s_{i,j}}$ to $g_{i+1,j} = \overline{s_{i+1,j}}$,
and each $\phi_{i,j}$ is a morphism $s_{i,j} \to s_{i+1,j}$. 
Define $\overline{\phi_{i,j}}$ to be the natural transformation 
$\overline{s_{i,j}} \to \overline{s_{i+1,j}}$ such that 
for any tree $f$, the corresponding component of $\overline{\phi_{i,j}}$ is 
$\phi_{i,j} \cdot f$. Let $\overline{\phi_i} = (\overline{\phi_{i,1}},\ldots,
\overline{\phi_{i,n}})$. 

\emph{Claim 3.} For each $i$, it holds that 
$\psi_i = \overline{\phi_i}$. 

The proof is similar to the above argument. 
Using this claim, it follows that $\phi_{i,j} = \psi_{i,j}(1^{\Sta})$
for each $i,j$, since 
\[\psi_{i,j}(1^{\Sta}) = \overline{\phi_{i,j}}(1^{\Sta}) = \phi_{i,j}.\]

It is now easy to complete the proof. By the Beki\'c identity,
\begin{eqnarray*} 
|E'^{\Sta}| 
&=& 
|G^{\Sta}|(1^{\Sta})\\
&=& 
\Colim((\psi_{i,1}(1^{\Sta}) : g_{i,1}(1) \to g_{i+1,1}(1^{\Sta}))_{i\geq 0})\\
&=& 
\Colim((\phi_{i,1} : s_{i,1} \to s_{i+1,1})_{i \geq 0})\\
&=& 
|E|^{\Sta},
\end{eqnarray*} 
up to isomorphism.
\end{proof}

\begin{exa}\label{Ex:convexpl}
Suppose that $E$ is given by the single equation 
\begin{eqnarray*}
F &=&  1 + a\cdot F \cdot b 
\end{eqnarray*}
Then $E'$ is 
\begin{eqnarray*}
G_0 &=& G(1)\\
G(v) &=& v + a.(G(b.v))
\end{eqnarray*}
Both of them define the tree depicted on Figure~\ref{Fig:anbn}. 
\end{exa}

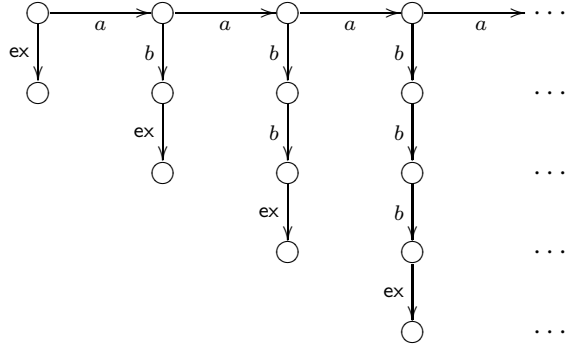
\begin{figure} 
\newcommand{\state}{*+[o][F-]{\makebox(3,3){}}}
\[\xymatrix@C=8ex@R=4ex{
\state{} \ar[d]_{\ex} \ar[r]_a & \state{} \ar[d]_b \ar[r]_a & \state{} \ar[d]_b \ar[r]_a   & \state{} \ar[d]_b \ar[r]_a & \cdots  \\
\state{}  & \state{}  \ar[d]_{\ex}  & \state{} \ar[d]_b  &  \state{}\ar[d]_b  & \cdots \\
          & \state{} & \state{} \ar[d]_{\ex}   & \state{} \ar[d]_b & \cdots \\
          &          & \state{}                & \state{} \ar[d]_{\ex}  & \cdots \\ 
          &          &                         & \state{} & \cdots \\
}
\]
\caption{The $\Delta$-regular and $\Gamma$-algebraic tree from Example~\ref{Ex:convexpl}}
\label{Fig:anbn}
\end{figure}

  The translation given in the proof of Theorem~\ref{Thm:Delta2Gamma}
allows us to characterize $\Delta$-regular trees by a syntactic restriction 
on $\Gamma$-algebraic schemes.

 \begin{cor}
\label{cor:charac-delta-regular-in-gamma-algebraic}
A tree is $\Delta$-regular if and only if it is defined by an algebraic recursion scheme $G$ over $\Gamma$ of the
  form
  \begin{eqnarray*}
    G_0 &=& G_1(1)\\
    G_1(v_1) &=& p_1\\
    &\vdots &\\
    G_n(v_1) &=& p_n
  \end{eqnarray*} 
  where $G_0$ has rank $0$, each $G_i$ with $i \geq 1$ has rank $1$,
  and where none of the terms $p_i$ has an occurrence of the constant
  $1$. 
\end{cor}

\begin{proof} In the light of the proof of Theorem~\ref{Thm:Delta2Gamma},
it is enough to show that any such scheme can be obtained from some
  regular recursion scheme $E$ over $\Delta$.

  Let $\Psi = \{G_1,\ldots,G_n\}$ and $\Phi = \{F_1,\ldots,F_n\}$,
  where each $F_i$ has rank $0$.  We give a transformation of a
  $\Gamma \cup \Psi$-term $p$ in the variable $v_1$, which contains no
  occurrence of the constant $1$, into a variable-free $\Delta \cup
  \Phi$-term $\widehat{p}$ such that $(\widehat{p})' = p$, where
  $(\widehat{p})'$ is defined in the proof of
  Theorem~\ref{Thm:Delta2Gamma}.
  \begin{enumerate}
  \item If $p = 0$ then $\widehat{p} = 0$.
  \item If $p = v_1$ then $\widehat{p} = 1$.
  \item If $p = p_1 + p_2$ then $\widehat{p} = \widehat{p_1} +
    \widehat{p_2}$.
  \item If $p = a.p_1$ then $p' = a \cdot \widehat{p_1}$.
  \item If $p = G_i(p_1)$ then $\widehat{p} = F_i \cdot
    \widehat{p_1}$.
  \end{enumerate}

  {\it Claim.} Let $p$ be a $\Gamma \cup \Psi$-term in the variable
  $v_1$ containing no occurrence of the constant $1$. Then
  $(\widehat{p})' = p$.

  We prove this claim by induction on the structure of $p$. When $p =
  0$, $\widehat{p} = 0$ and $(\widehat{p})' = 0$, and when $p = v_1$,
  $\widehat{p} = 1$ and $(\widehat{p})' = v_1$.  Suppose that $p$ is
  of the form $p_1 + p_2$ and that the claim holds for $p_1$ and
  $p_2$. In this case $\widehat{p} = \widehat{p_1} + \widehat{p_2}$
  and $(\widehat{p})' = (\widehat{p_1})' + (\widehat{p_2})' = p_1 +
  p_2 = p$, by the induction hypothesis.  Next let $p = a.p_1$ where
  $a \in A$ and suppose that the claim holds for $p_1$.  Then
  $\widehat{p} = a \cdot \widehat{p_1}$ and $(\widehat{p})' =
  (a.v_1)((\widehat{p_1})') = a.(\widehat{p_1})' = a.p_1 = p$. Last,
  suppose that $p = G_i(p_1)$ for some $i \in [n]$ and $p_1$
  satisfying the claim. Then we have $\widehat{p} = F_i\cdot
  \widehat{p_1}$ and $(\widehat{p})' = G_i((\widehat{p_1})') =
  G_i(p_1) = p$.  This completes the proof of the claim.

  Now consider the regular $\Delta$-scheme $E$
  \begin{eqnarray*}
    F_1 &=& \widehat{p_1}\\
    &\vdots &\\ 
    F_n &=& \widehat{p_n}
  \end{eqnarray*} 
  By the proof of Theorem~\ref{Thm:Delta2Gamma}, $G$ corresponds to
  $E$. Thus, $E$ and $G$ define the same tree.   
\end{proof}

\begin{rem}
The restriction on the use of the constant $1$ in Corollary~\ref{cor:charac-delta-regular-in-gamma-algebraic} is necessary.
Indeed the proof of Proposition~\ref{prop:gamma-algebraic-not-delta-regular} to follow gives an example of $\Gamma$-algebraic scheme of rank 1 generating a tree that is not $\Delta$-regular even up to bisimulation equivalence.
\end{rem}

To establish Equation~\eqref{eq:hierarchy-bisim}, it remains to show the strictness of the second inclusion up to bisimulation equivalence.

\begin{prop}
\label{prop:gamma-algebraic-not-delta-regular}
There exists a $\Gamma$-algebraic synchronization tree that is not 
bisimilar to any $\Delta$-regular tree. 
\end{prop}

\begin{proof}
Let $A = \{a,b\}$, and consider the synchronization tree 
$T\in \Sta$, defined by the $\Gamma$-algebraic recursion scheme:
\begin{eqnarray*}
S      & =&  F (1 + b.1)   \\
F(v_1) & =&   v_1 + a.(F (1 + b. v_1))\ .
\end{eqnarray*}

The tree $T$, depicted in Figure~\ref{fig:tree-strict}, has a single infinite branch whose edges are labelled 
$a$, and the out-degree of each vertex on this branch is $3$,
since each such vertex is the source of an edge labelled $a$, 
an edge labelled $b$, and an edge labelled $\ex$. Since 
$T$ is deterministic, its vertices may be identified 
with the words in the prefix closed language 
\[\{a^{n}b^{m}, a^{n} b^{m} \ex \mid n \geq 0 \tand m \leq n+1 \}.\] 
A key feature is that every vertex is the source of an edge labelled 
$\ex$.

\begin{figure}
  \begin{center}
    \begin{tikzpicture}[transform shape,scale=1.0,edge from parent/.style={draw,-latex}]
 \tikzstyle{n}=[circle, draw=black, fill=black, inner sep = 0.25mm, outer sep= 0.5mm]
 \tikzstyle{e}=[ inner sep = 0mm, outer sep= 0mm]
 \tikzstyle{l}=[font=\tiny,label distance=-1mm]
 \tikzstyle{level}=[level distance=0.95cm,sibling distance=1.75cm]

\node[n,label={[l]above:{$(\varepsilon)$}}] {}
  child { node[n,label={[l]above left:{$(a)$}}] {}
    child { node[n,label={[l]above left:{$(aa)$}}] {}
      child { node[e] {}
         edge from parent[dashed,-] }
      child { node[n,label={[l]left:{$(aab)$}}] {}
        child { node[e] {}
        edge from parent[draw=none] }
        child { node[n,label={[l]left:{$(aabb)$}}] {}
          child { node[e] {}
          edge from parent[draw=none] }
          child { node[n,label={[l]left:{$(aabbb)$}}] {}
            child[sibling distance=0.75cm] { node[n] {}
            edge from parent node[draw=none,right] {$\ex$}}
          edge from parent node[draw=none,left] {$b$}}
          child[sibling distance=0.75cm] { node[n] {}
          edge from parent node[draw=none,right] {$\ex$}}
        edge from parent node[draw=none,left] {$b$}}
        child[sibling distance=0.75cm] { node[n] {}
        edge from parent node[draw=none,right] {$\ex$}}
      edge from parent node[draw=none,left] {$b$}}
      child[sibling distance=0.75cm] { node[n] {}
      edge from parent node[draw=none,right] {$\ex$}}
    edge from parent node[draw=none,above left] {$a$}}
    child { node[n,label={[l]left:{$(ab)$}}] {}
      child { node[e] {}
              edge from parent[draw=none] }
      child { node[n,label={[l]left:{$(abb)$}}] {}
        child[sibling distance=0.75cm] { node[n] {}
        edge from parent node[draw=none,right] {$\ex$}}
      edge from parent node[draw=none,left] {$b$}}
      child[sibling distance=0.75cm] { node[n] {}
      edge from parent node[draw=none,right] {$\ex$}}
    edge from parent node[draw=none,left] {$b$}}
    child[sibling distance=0.75cm]  {node[n] {}
    edge from parent node[draw=none,right] {$\ex$}}
  edge from parent node[draw=none,above left] {$a$}}
  child { node[n,label={[l]left:{$(b)$}}] {}
    child[sibling distance=0.75cm] { node[n] {}
    edge from parent node[draw=none,right] {$\ex$}}
  edge from parent node[draw=none,left] {$b$}}
  child[sibling distance=0.75cm] { node[n] {}
  edge from parent node[draw=none,right] {$\ex$}}
;

\end{tikzpicture}
  \end{center}
  \caption{The tree $T$.}
\label{fig:tree-strict}
\end{figure}
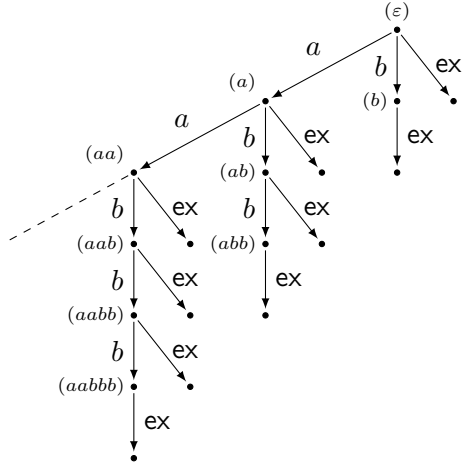

We are going to show that every $\Delta$-regular scheme $E$ defining a
synchronization tree bisimilar to $T$ is equivalent to a right-linear
one, modulo bisimilarity.  By
Proposition~\ref{prop:right-linear-implies-regular}, this would imply
that $T$ is $\Gamma$-regular, which yields a contradiction.  Indeed,
the tree $T$ has a countably infinite set of subtrees that are
pairwise non-bisimilar.

Without loss of generality, we may assume that the equations of $E$ 
are of one of three forms:
\begin{enumerate}
\item $G = G_{1} + G_2$,
\item $G = G_{1} \cdot G_{2}$, and 
\item $G = c$ for $c \in \{a,b\}\cup \{ 0, 1\}$\ .
\end{enumerate}

In the following, we are interested in a particular family 
$\mathcal{F}$ of `subtrees" of $T$, containing, for all $k \geq 0$, 
the subtree $T_k$ rooted at $a^k$, and the trees obtained from $T_k$  
by removing the exit edge originating in the root together with its target, or 
the edge labelled $b$ originating in the root together with all vertices and edges 
accessible from the end vertex of that edge, or both. We denote these 
synchronization trees by $T_k(1),T_k(b)$ and $T_k(1,b)$, respectively.

\begin{lem}
\label{lem-lem2}
Suppose that a tree $s\in \mathcal{F}$ is bisimilar to a tree $s_1 \cdot s_2$,
where neither $s_1$ nor $s_2$ is bisimilar to $1$. 
Then for some $k$, $s = T_k(1,b)$, $s_1$ is bisimilar to $a = a.1$, and $s_2$ is 
bisimilar to $T_{k+1}$.  
\end{lem}  

\begin{proof}
Suppose that $s$ is bisimilar to $s_1 \cdot s_2$ and 
neither $s_1$ nor $s_2$ is bisimilar to the tree $1$. Note that $s_2$ is not $0$. 
Clearly, each vertex of $s_1$, except possibly the root, 
must be the source of an exit edge. Suppose that $s$ 
contains and edge labelled $a$ from $x_1$ to $x_2$ 
such both $x_1$ and $x_2$ are sources of an exit edge. 
Then in the tree $s_1\cdot s_2$, they have successor vertices 
$y_1$ and $y_2$ such that the subtrees rooted at $y_1$ and $y_2$ 
are isomorphic (and thus bisimilar) and contain at least one edge.
But $T$ does not have such vertices connected by an 
edge labelled $a$ and therefore neither does $s$.
For this reason, $s_1$ cannot have two 
consecutive edges labelled $a$ either. This in turn yields that 
$s_2$ has at least one edge labelled $a$ and therefore 
$s_1$ cannot have an edge labelled $b$. We conclude that
$s_1$ is bisimilar to the tree $a = a.1$ and then 
$s_2$ is bisimilar to $T_{k+1}$ for some $k$.  
\end{proof}

Now, by Lemma~\ref{lem-lem2}, we may transform $E$ into a right-linear scheme 
defining $T$ up to bisimilarity. First mark all those variables $G$ 
such that the corresponding component in the initial solution of $E$ 
over $\Sta$ has an infinite branch. The first variable is clearly marked. 
Suppose that $G$ is marked. If the equation for $G$ is 
$G = G_1 + G_2$, then $G_1$ or $G_2$ is marked. If one of them 
is not marked, then it can be replaced by a variable-free term.
If the equation for $G$ is $G = G_1 \cdot G_2$ and 
the component in the initial solution of $E$ 
over $\Sta$ corresponding to one of the $G_i$
is bisimilar to $1$, then we may simply remove it.
Otherwise we apply  Lemma~\ref{lem-lem2} and replace 
$G_1$ by $a$ and mark $G_2$ if it is not yet marked. 
 Eventually, we keep only the marked functor 
variables and obtain a right-linear scheme defining $T$ 
up to bisimilarity.
\end{proof}

\subsection{Comparison with BPA and BPP} 

The $\Delta$-regular trees that can be defined using
regular $\Delta$-recursion schemes that do not contain occurrences of
the constants $0$ and $1$ correspond to unfoldings of the labelled
transition systems denoted by terms in Basic Process Algebra (BPA)
with recursion, see, for
instance,~\cite{BaetenBR2009,Baetenetal,BergstraK84}. Indeed, the
signature of BPA contains one constant symbol $a$ for each action as
well as the binary $+$ and $\cdot$ operation symbols, denoting
nondeterministic choice and sequential composition,
respectively. In the remainder of this paper, we write
BPA for `BPA with recursion'. 

Alternatively, following~\cite{Moller96}, one may view
BPA as the class of labelled transition systems associated with
context-free grammars in Greibach normal form in which only leftmost
derivations are permitted. 

The class of Basic Parallel Processes (BPP) is a parallel counterpart
of BPA introduced by Christensen~\cite{Christensen83}. BPP consists of
the labelled transition systems associated with context-free grammars
in Greibach normal form in which arbitrary derivations are allowed. We
refer the interested readers to~\cite{Moller96} for the details of the
formal definitions, which are not needed to appreciate the results to
follow, and further pointers to the literature.

\begin{prop}\label{Prop:moreexpthanBPA}\quad
\begin{enumerate}
\item Every synchronization tree that is the unfolding of a BPA
process is $\Gamma$-algebraic.
\item There is a $\Gamma$-algebraic synchronization tree that is
  neither definable in BPA modulo bisimilarity nor in BPP modulo
  language equivalence.
\end{enumerate}
\end{prop} 
\begin{proof}
The former claim follows easily from
Theorem~\ref{Thm:Delta2Gamma}. In order to prove the latter statement,
consider the LTS depicted on Figure~\ref{fig:algebraic-not-BPA}. This
LTS is not expressible in BPA modulo modulo bisimilarity and is not
expressible in BPP modulo language equivalence (if the states $q$ are
$r$ are the only final states in the LTS)---see~\cite[page~206,
  Example (f)]{Moller96}. On the other hand, the synchronization tree
associated with that LTS is $\Gamma$-algebraic because it is the
unique solution of the recursion scheme below.
\begin{eqnarray*}
F_1 & = & b + c + a . F_2(b^2,c^2) \\
F_2(v_1,v_2) & = & v_1 + v_2 + a . F_2(b . v_1, c . v_2) 
\end{eqnarray*}
\end{proof}

So non-regular $\Gamma$-algebraic recursion schemes are more
expressive than BPA modulo bisimilarity and can express
synchronization trees that cannot be defined in BPP up to language
equivalence, and therefore up to bisimilarity. 

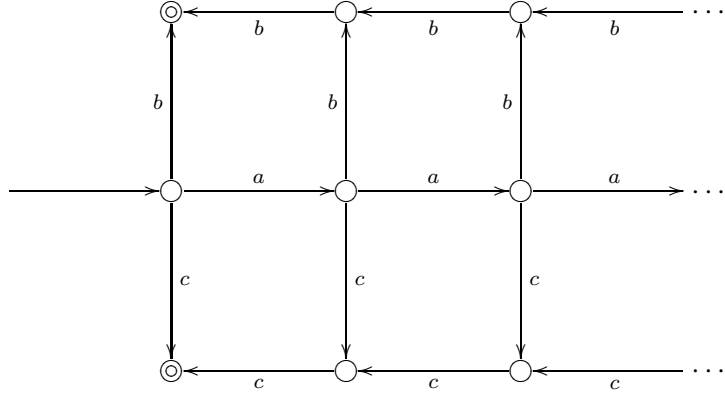
\begin{figure} 
\newcommand{\state}{*+[o][F-]{\makebox(3,3){}}}
\newcommand{\istate}{*+[o][F-]{\makebox(3,3){}}}
\newcommand{\fstate}{*+[o][F=]{\makebox(3,3){}}}
\[\xymatrix@C=12ex@R=12ex{
& \fstate{} & \state{} \ar[l]^b  & \state{} \ar[l]^b & \ar[l]^b \cdots  \\
\ar[r] & \istate{} \ar[u]^b \ar[r]^a \ar[d]^c & \state{}\ar[u]^b \ar[r]^a \ar[d]^c &  \state{}\ar[u]^b \ar[r]^a \ar[d]^c & \cdots \\
& \fstate{} & \state{} \ar[l]^c  & \state{} \ar[l]^c & \ar[l]^c \cdots  
}
\]
\caption{An LTS whose unfolding is an algebraic synchronization tree}
\label{fig:algebraic-not-BPA}
\end{figure}

\section{Comparison with the Caucal hierarchy}\label{Sect:caucal}

In this section, we compare the expressiveness of recursion schemes to
that of the low classes in the Caucal
hierarchy~\cite{Caucal}. Section~\ref{Sect:definition-caucal} gives a
general overview of the Caucal
hierarchy. Section~\ref{Sect:first-levels} presents in more details
the properties of the graphs and trees sitting in the first
levels of the hierarchy. Section~\ref{Sect:treesinCH} shows that the classes of
synchronization trees we introduced belong to the Caucal
hierarchy. Finally Section~\ref{sec:charac-gamma-algebraic} characterises the
$\Gamma$-algebraic synchronization trees as contractions 
of the synchronization trees in the class $\Tree_{2}$.

\subsection{The Caucal hierarchy}
\label{Sect:definition-caucal}

The Caucal hierarchy (also known as the pushdown hierarchy) is 
a hierarchy of classes of edge-labelled graphs. Following~\cite{CarayolWohrle}, the Caucal hierarchy is
\[
\Tree_{0} \subseteq \Graph_{0} \subseteq \Tree_{1} \subseteq \Graph_{1} \subseteq \cdots
\]
where  $\Tree_0$ and $\Graph_0$ denote the classes of finite, edge-labelled trees and graphs, respectively. Moreover, for each $n\geq 0$, $\Tree_{n+1}$
stands for the class of trees isomorphic to unfoldings of graphs in $\Graph_n$, and the graphs in $\Graph_{n+1}$ are those that can be obtained from the trees in $\Tree_{n+1}$ by applying a monadic interpretation (or
transduction)~\cite{Courcelle94}. As shown by Caucal in \cite{Caucal03}, every graph in the Caucal hierarchy has a decidable {theory for monadic second order logic.}

\begin{rem}
\label{rem:rooted-graphs}
For a graph $G \in \Graph_{n}$ and a vertex $v$ of $G$, the graph $H$ obtained 
by restricting $G$ to the set of vertices reachable from $v$ is also a graph in $\Graph_{n}$ \cite{CarayolWohrle}. In particular, we can always assume that
a tree in $\Tree_{n+1}$ is obtained by unfolding a graph in $\Graph_{n}$ from one of its root.
\end{rem}

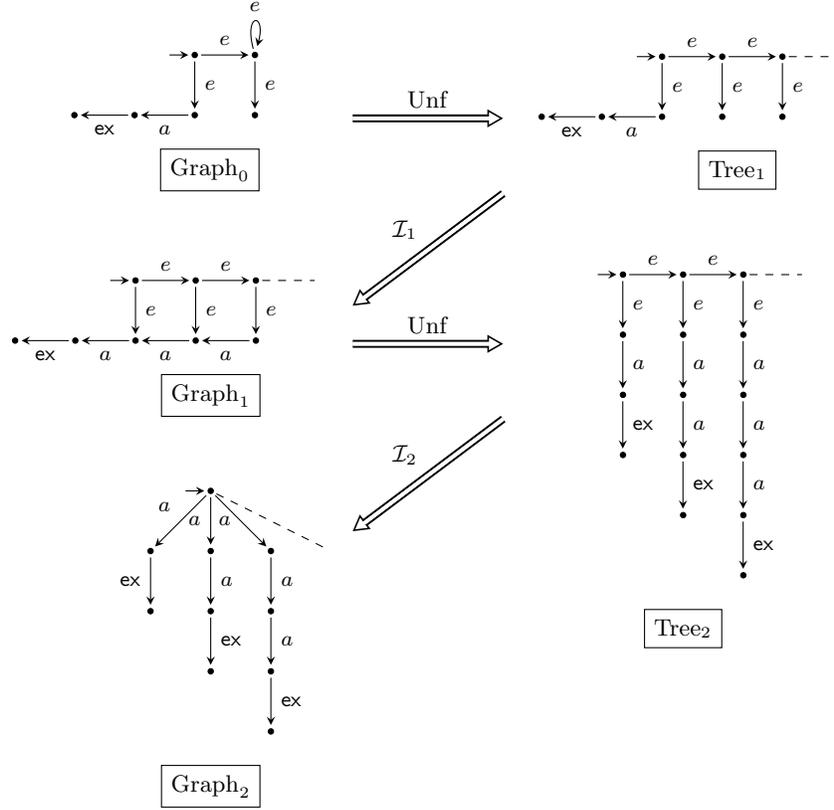
\begin{figure}
  \begin{center}
 \tikzstyle{n}=[circle, draw=black, fill=black, inner sep = 0.25mm, outer sep= 0.5mm]
 \tikzstyle{e}=[ inner sep = 0mm, outer sep= 0mm]
 \tikzstyle{l}=[font=\tiny,label distance=-1mm]
 \tikzstyle{o}=[draw,rectangle]

\tikzstyle{vecArrow} = [thick, decoration={markings,mark=at position
   1 with {\arrow[semithick]{open triangle 60}}},
   double distance=1.4pt, shorten >= 5.5pt,
   preaction = {decorate},
   postaction = {draw,line width=1.4pt, white,shorten >= 4.5pt}]

\begin{tikzpicture}[scale=1,>=stealth,font=\scriptsize]

 \begin{scope}
     \matrix[matrix of nodes,column sep={0.8cm,between origins},row
     sep={0.8cm,between origins}]{
       &    &    |[n,initial,initial text=,initial distance=0.25cm](h0)| & |[n](g0)|  \\
       |[n](h3)| &   |[n](h2)|   &    |[n](h1)| & |[n](g1)| \\
        & |[o,xshift=1cm]| {\footnotesize$\mathrm{Graph}_{0}$} & & \\
     }; 
\draw (h0) edge[->] node[right] {$e$} (h1) (h1) edge[->]
     node[below] {$a$} (h2) (h2) edge[->] node[below] {$\ex$} (h3)
     (h0) edge[->] node[above] {$e$} (g0) (g0) edge[->] node[right]
     {$e$} (g1) (g0) edge[loop above] node[above] {$e$} (g0);
 \end{scope}

 \begin{scope}[xshift=7cm]
   \matrix[matrix of nodes,column sep={0.8cm,between origins},row sep={0.8cm,between origins}]{ & & |[n,initial,initial text=,initial
     distance=0.25cm](h1)| & |[n](h2)|
     & |[n](h3)| & |[e](h4)|   \\
     |[n](h5)| &   |[n](h6)|   &    |[n](h7)| & |[n](h8)| & |[n](h9)| & \\
      &  &  |[o,xshift=1cm]| {\footnotesize$\mathrm{Tree}_{1}$} &  \\
   }; \draw (h1) edge[->] node[above] {$e$} (h2) (h2) edge[->]
   node[above] {$e$} (h3) (h1) edge[->] node[right] {$e$} (h7) (h2)
   edge[->] node[right] {$e$} (h8) (h3) edge[->] node[right] {$e$}
   (h9) (h7) edge[->] node[below] {$a$} (h6) (h6) edge[->] node[below]
   {$\ex$} (h5) (h3) edge[dashed] (h4);
 \end{scope}

 \begin{scope}[yshift=-3cm]
   \matrix[matrix of nodes,column sep={0.8cm,between origins},row
   sep={0.8cm,between origins}]{ & & |[n,initial,initial text=,initial
     distance=0.25cm](h1)| & |[n](h2)|
     & |[n](h3)| & |[e](h4)|   \\
     |[n](h5)| &   |[n](h6)|   &    |[n](h7)| & |[n](h8)| & |[n](h9)| & \\
      &  &  & |[o,xshift=0.2cm]| {\footnotesize$\mathrm{Graph}_{1}$}  & & \\
   }; \draw (h1) edge[->] node[above] {$e$} (h2) (h2) edge[->]
   node[above] {$e$} (h3) (h1) edge[->] node[right] {$e$} (h7) (h2)
   edge[->] node[right] {$e$} (h8) (h3) edge[->] node[right] {$e$}
   (h9) (h7) edge[->] node[below] {$a$} (h6) (h6) edge[->] node[below]
   {$\ex$} (h5) (h3) edge[dashed] (h4) (h8) edge[->] node[below] {$a$}
   (h7) (h9) edge[->] node[below] {$a$} (h8);
 \end{scope}

 \begin{scope}[xshift=7cm,yshift=-4.5cm]
   \matrix[matrix of nodes,column sep={0.8cm,between origins},row
   sep={0.8cm,between origins}]{
     |[n,initial,initial text=,initial distance=0.25cm](h1)| & |[n](h2)| & |[n](h3)| & |[e](h4)|   \\
     |[n](g1)| & |[n](g2)| & |[n](g3)| & \\
     |[n](i1)| & |[n](i2)| & |[n](i3)| & \\
     |[n](j1)| & |[n](j2)| & |[n](j3)| & \\
     & |[n](k2)| & |[n](k3)| & \\
     &           & |[n](l3)| & \\
      &    |[o]| {\footnotesize$\mathrm{Tree}_{2}$} & &  \\
   }; \draw (h1) edge[->] node[above] {$e$} (h2) (h2) edge[->]
   node[above] {$e$} (h3) (h1) edge[->] node[right] {$e$} (g1) (h2)
   edge[->] node[right] {$e$} (g2) (h3) edge[->] node[right] {$e$}
   (g3) (g1) edge[->] node[right] {$a$} (i1) (g2) edge[->] node[right]
   {$a$} (i2) (g3) edge[->] node[right] {$a$} (i3) (i1) edge[->]
   node[right] {$\ex$} (j1) (i2) edge[->] node[right] {$a$} (j2) (i3)
   edge[->] node[right] {$a$} (j3) (j2) edge[->] node[right] {$\ex$}
   (k2) (j3) edge[->] node[right] {$a$} (k3) (k3) edge[->] node[right]
   {$\ex$} (l3) (h3) edge[dashed] (h4);
 \end{scope}

 \begin{scope}[yshift=-7cm,xshift=1cm]
   \matrix[matrix of nodes,column sep={0.8cm,between origins},row
   sep={0.8cm,between origins}]{
     &  |[n,initial,initial text=,initial distance=0.25cm](h1)| &  & \\
     |[n](g1)|  & |[n](g2)| & |[n](g3)| & |[e](d)|\\
     |[n](i1)| & |[n](i2)| & |[n](i3)| &  \\
     & |[n](j2)| & |[n](j3)| & \\
     &           & |[n](k3)| & \\
      &    |[o]| {\footnotesize$\mathrm{Graph}_{2}$} & &  \\
   };

   \draw (h1) edge[->] node[above left] {$a$} (g1) (h1) edge[->]
   node[left] {$a$} (g2) (h1) edge[->] node[ left] {$a$} (g3)
 
   (g1) edge[->] node[left] {$\ex$} (i1) (g2) edge[->] node[right]
   {$a$} (i2) (g3) edge[->] node[right] {$a$} (i3)

   (i2) edge[->] node[right] {$\ex$} (j2) (i3) edge[->] node[right]
   {$a$} (j3)

   (j3) edge[->] node[right] {$\ex$} (k3)

   (h1) edge[dashed] (d);
 \end{scope}

\draw (2.5cm,-0cm) edge[vecArrow] node[above] {\footnotesize $\mathrm{Unf}$} (4.5cm,-0cm);
\draw (4.5cm,-1cm) edge[vecArrow] node[above left] {\footnotesize $\mathcal{I}_{1}$} (2.5cm,-2.5cm);
\draw (2.5cm,-3cm) edge[vecArrow] node[above] {\footnotesize $\mathrm{Unf}$} (4.5cm,-3cm);
\draw (4.5cm,-4cm) edge[vecArrow] node[above left] {\footnotesize $\mathcal{I}_{2}$} (2.5cm,-5.5cm);
  \end{tikzpicture}
\end{center}

  \caption{A construction of the synchronization tree $\sum_{i \geq 1} a^{i}$ in the Caucal hierarchy.}
  \label{fig:caucal}
\end{figure}

\begin{exa}
Figure~\ref{fig:caucal} shows a sequence of transformations constructing the synchronization tree $\sum_{i \geq 1} a^{i}$ starting from a finite graph. 
{
The vertices from which the graphs are unfolded are signalled by an incoming arrow.
The MSO-interpretation $\mathcal{I}_{1}$ adds an $a$-labelled edge from a vertex $u$
to a vertex $v$ if $u$ has no outgoing edges and there exist two vertices $s$ and $t$ with 
$s \era{e} t$, $s \era{e} v$, $t \era{e} u$. The MSO-interpretation $\mathcal{I}_{2}$ is the $\{e\}$-contraction operation.}
\end{exa}

\subsection{First levels of the hierarchy}
\label{Sect:first-levels}

The class $\Tree_{1}$ contains the regular trees of finite outdegree (i.e.,
the trees of finite degree having finitely many non-isomorphic subtrees). 
It is well known that $\Graph_1$ is the set of all prefix-recognizable graphs~\cite{Caucal03}. 

A \emph{prefix recognizable relation} over a finite alphabet $C$ is a
finite union of relations of the form\footnote{ $U \cdot (V\times W)$ denotes the relation $\{ (uv,uw) \mid u \in U, v\in V \;\textrm{and}\; w \in W\}$.} $U \cdot (V\times W)$, for some
nonempty regular languages $U,V,W\subseteq C^*$. A \emph{prefix
  recognizable graph} over an alphabet $B$ is an edge-labelled graph
that is isomorphic to a graph of the form $(V,
(\stackrel{b}{\longrightarrow})_{b \in B})$, where for some alphabet
$C$, $V$ is a regular subset of $C^*$ and the edge relations
$\stackrel{b}{\longrightarrow}$ are prefix recognizable relations over
$C$. (Of course, the alphabet $C$ may be fixed to be a 2-element
alphabet.)

\begin{prop}
\label{prop:level-one}
For a labelled graph $G$, the following statements are equivalent:
\begin{itemize}
\item $G$ belongs to $\Graph_{1}$,
\item $G$ is isomorphic to a prefix-recognizable graph \cite{Caucal},
\item $G$ can be MSO-interpreted in the full binary tree $\Delta_{2}$ \cite{Blumensath}.
\end{itemize}
\end{prop}

\begin{prop}
\label{prop:det-closure}
Let $G$ be a graph in $\Graph_{1}$ labelled in $A$, and let $r$ be a root of $G$. The graph $G$ is isomorphic to the $B$-contraction of a deterministic graph $H \in \Graph_{1}$ labelled in $A \uplus B$ from one of its root $r'$.

Moreover $H$ can be chosen such that for any $a \in A$ and any two vertices $u$  and $v$ belonging to the $B$-contraction from $r'$, there exists at most one path 
from $u$ to $v$ 
labelled by a word in $B^{*}a$. 
\end{prop}

\begin{proof}
Let $G=(V,(R_{a})_{a \in A})$ be a prefix-recognizable graph labelled
in $A$. We may assume, without loss of generality, that the vertices
in $V$ are words over the alphabet $C=\{0,1\}$. Moreover we may assume that the relations $R_{a}$, $a \in A$, can be expressed as the disjoint union of  relations $U_{a,1} \cdot
(V_{a,1} \times W_{a,1}), \ldots, U_{a,n_{a}} \cdot (V_{a,n_{a}}
\times W_{a,n_{a}})$. In addition, we may require that
 $\First(V_{a,i}) \cap \First(W_{a,i})
\subseteq \{\,\varepsilon\,\}$, for all $i \in [n_{a}]$ 
where $\First(L)=\{ a \in C \mid au \in L \;\textrm{for some $u\in C^{*}$}\}
\cup \{ \varepsilon \mid \varepsilon \in L\}$ \cite[Proposition 2.1]{Carayol05}. These assumptions guarantee that if $(x,y)$ belongs to $R_{a}$,
then there exists a unique $i \in [n_{a}]$ and a unique decomposition 
$x=uv$ and $y=uw$ such that $u \in U_{a,i}, v \in V_{a,i}$ and $w \in W_{a,i}$.

Before proceeding with the construction, we need to introduce some
notations.  Let $a \in A$ and let $i \in [n_{a}]$. We take
$\mathcal{A}_{a,i}=(Q_{a,i},q_{i,a},F_{a,i},\delta_{a,i})$ to be a
complete DFA accepting the reverse of $V_{a,i}$ and
$\mathcal{B}_{a,i}=(Q_{a,i}',q_{i,a}',F_{a,i}',\delta_{a,i}')$ to be a
complete DFA accepting $W_{a,i}$. We assume that the sets of states of
these automata are pairwise disjoint and we take $Q=\bigcup_{a \in A,
i \in [n_{a}]} Q_{a,i}$ and $Q'=\bigcup_{a \in A, i \in [n_{a}]}
Q_{a,i}'$.

We are now going to define a prefix-recognizable graph
$H=(V',(R_{a}')_{a \in A \cup B})$ satisfying the properties stated
above.

The set of labels $B$ is $B=\{e_{0},e_{1},e_{2}\} \cup \{ e_{a,i} \mid
a \in A \tand i \in [n_{a}] \}$. The vertices of $H$ are words over
the alphabet $\{0,1\} \cup Q \cup Q' \cup \{ (a,i) \mid a \in A
\;\text{and}\; i \in [n_{a}] \} \cup \{ \star \}$.

Intuitively, the graph $H$ is constructed in such a way that a path
labelled by a word in $B^{*} a$ from a vertex $x \in V$ to a vertex $y
\in V$ simulates the relation $R_{a,i}$ for  $a \in A$ and $i \in
[n_{a}]$. 
For a fixed $a \in A$ and $i \in [n_{a}]$, the simulation is
done using two sets 
 of vertices $V_{a,i}= \{0,1\}^{*}
Q_{a,i} (a,i)$ and $V_{a,i}'= \{0,1\}^{*} Q'_{a,i} (a,i)$. 
Starting from a vertex $x \in V$, the vertices in $V_{a,i}$ are used
to remove a suffix $v$ of $x$ belonging to $V_{a,i}$ and the vertices
in $V'_{a,i}$ are used to add a suffix $w$ in $W_{a,i}$. When moving
from the vertices in $V_{a,i}$ to the vertices in $V_{a,i}'$, we check
that the remaining prefix $u$ belongs to $U_{a,i}$. This guarantees
that $x$ and $y$ can be respectively written as $uv$ and $uw$
with $u \in U_{a,i}, v \in V_{a,i}$ and $w \in W_{a,i}$.

Formally, from a vertex $x \in V$, there is 
an edge labelled $e_{a,i}$ to the vertex $x q_{a,i} (a,i)$ (cf. (1) below).  For all
words $u,v \in \{0,1\}^{*}$ and for all $q \in Q_{a,i}$,  $uv q (a,i) \erb{e_{0}^{|v|}}{H} u \delta_{a,i}(q,u) (a,i)$ (cf. (2) below). If $u \in \{0,1\}^{*}$ belongs to $U_{a,i}$ and $q \in F_{a,i}$,
$u q (a,i) \erb{e_{1}}{H} u q_{a,i}' (a,i)$ (cf. (3) below). For all
words $u,v \in \{0,1\}^{*}$ and for all $q \in Q'_{a,i}$,  $u q (a,i) \erb{e_{0}^{*}}{H} uv \delta_{a,i}'(q,v) (a,i)$ (cf. (4) below). Finally, for $u \in \{0,1\}^{*}$ and $q \in F_{a,i}'$,
we have $u q (a,i) \erb{e_{2}}{H} u \star (a,i) \erb{a}{H} u$ 
(cf. (5) and (6) below).

The set of vertices $V'$ is taken to be $V \cup \bigcup_{a \in A}
\Dom(R_{a}') \cup \Ima(R_{a}')$. The edges of $H$ are defined, for all
$a \in A$, $i \in [n_{a}]$, $q \in Q_{a,i}, q' \in Q_{a,i}'$, $u \in
\{0,1\}^{*}$ and $b \in \{0,1\}^{*}$, by:
\[
\begin{array}{lcclclr}
u  & \erb{e_{a,i}}{H} & \;\;& u \,q_{a,i} (a,i) & & \text{for $u \in V$}  & (1)\\[1em]
u\,b q (a,i) & \erb{e_{0}}{H} &&  u \,\delta_{a,i}(q,b)   && \text{for $b \in \{0,1\}$} & (2) \\[1em]
u\,q (a,i)  &\erb{e_{1}}{H}& & u \,q_{a,i}' (a,i)  & & \text{if}\; q \in F_{i,a} \tand u \in U & (3) \\[1em]
u\, q' (a,i) & \erb{e_{0}}{H} & & ub \,\delta_{a,i}'(q,b) (a,i)  & & \text{for $b \in \{0,1\}$} & (4) \\[1em] 
u \,q' (a,i)  &\erb{e_{2}}{H}& & u \, \star (a,i)  & & \text{if}\; q' \in F_{a,i}'  & (5)\\[1em]
u \,\star (a,i)  &\erb{a}{H} & & u & &\text{for $u \in V$}   & (6)\\
\end{array}
\]

The graph $H$ is a deterministic prefix-recognizable graph. 
We have already seen that:

\begin{claim}
  For all $x,y \in \{0,1\}^{*}$,  if $x \erb{a}{G} y$ then 
there exists a path in $H$ labelled in $B^{*} a$ from $x$ to $y$.
\end{claim}

It remains to show the other direction.

\begin{claim}
  For all $x,y \in \{0,1\}^{*}$, if there exists a path from $x$
to $y$ in $H$ labelled in $B^{*}a$ then this path is unique and $x \erb{a}{G} y$.
\end{claim}

\begin{proof}
Let $x,y \in \{0,1\}^{*}$ be such that there exists a path from $x$ to $y$ in $H$ labelled in $B^{*}a$. By the construction of $H$,  there exist $i \in [n_{a}]$,
and $u,v$ and $w \in \{0,1\}^{*}$ with $x=uv$ and $y=uv$ such that the path is of the form
\[
\begin{array}{l}
x \erb{e_{a,i}}{H} x q_{a,i} \erb{e_{0}^{|u|}}{H} u \delta(q_{a,i},u) (a,i) \era{e_{1}}  \\
\quad u q_{a,i}' (a,i) 
\erb{e_0^{|w|}}{H} 
uw \delta_{a,i}'(q_{a,i}',w) \erb{e_{2}} uw \star (a,i) \erb{a}{H} uw \\
\end{array}
\]
with $\delta(q_{a,i},u) \in F_{a,i}$, 
$\delta(q'_{a,i},w) \in F'_{a,i}$ and $u \in U_{a,i}$. It
follows that $(x,y)$ belongs to $R_{a,i}$. As remarked before, the index $i$ is unique
and so is the decomposition into $u,v$ and $w$. Hence the uniqueness
of the path follows.    \end{proof}

From the above two claims, it follows that the $B$-contraction of $H$ is equal to $G$. 
  \end{proof}

As first remarked in \cite{Stirling00}, the graphs in $\Graph_{1}$ can be obtained by a form of contraction\footnote{In \cite{Stirling00}, the notion of $B$-contraction keeps the vertices that are the source of an edge in $A$ and adds
an edge labelled by $a \in A$ between two such vertices whenever there is a path labelled in $a B^{*}$.} of the configurations graph of pushdown automata \cite{CarayolWohrle}.

Let us now consider the class $\Tree_{2}$. The main property of the class $\Tree_{2}$ is that it contains the $\Sigma$-term trees defined by any
$\Sigma$-algebraic recursion scheme for any signature $\Sigma$.

As $\Tree_{2}$ only contains labelled trees, we need to fix a
representation of $\Sigma$-term trees as labelled trees. 
We follow the
one from \cite{Caucal} which is slightly different from the one
presented in Section~\ref{sec:morphism} (see also
Remark~\ref{rem:ht-caucal}).  A $\Sigma$-term tree $t$ is represented
by a tree $t'$ labelled by $\Sigma \cup \{\underline{1}, \ldots
,\underline{m} \}$ where $m$ is the maximum rank of a symbol of
$\Sigma$.

The vertices of $t'$ are prefixes of  sequences of the form $u_{0} \underline{d_{0}} \cdots u_{n-1} \underline{d_{n-1}} u_{n}$ where $u_{0}$ is the root of $t$, $u_{n}$ is a leaf of $t$ and for all
$i \in [n]$, $u_{i}$ is the $d_{i-1}$-th successor of $u_{i-1}$. Let $u$ be a vertex of $t$ labelled by $a \in \Sigma$, if  $w$ and 
$w u$  are vertices of $t'$ then there is an edge
from $w$ to $w u$ labelled by $a$. Let $i \in [m]$, if $w$ and $w \underline{i}$ are vertices of $t'$ then there is an edge from $w$ to $w \underline{i}$ labelled by $\underline{i}$. This representation is illustrated in Figure~\ref{fig:representation-caucal} for the $\Delta$-term $(a+0) \cdot c$.

\begin{figure}
  \begin{center}
    \begin{tikzpicture}[transform shape,scale=1.0,edge from parent/.style={draw,-latex}]
 \tikzstyle{n}=[circle, draw=black, fill=black, inner sep = 0.25mm, outer sep= 0.5mm]
 \tikzstyle{e}=[ inner sep = 0mm, outer sep= 0mm]
 \tikzstyle{l}=[font=\tiny,label distance=-1mm]
 \tikzstyle{level}=[level distance=0.85cm,sibling distance=1cm]
 \tikzstyle{level 1}=[level distance=0.85cm,sibling distance=1.25cm]

\node[label={[l]above:{$(v0)$}}] {$\cdot$}
  child { node[label={[l]left:{$(v1)$}}] {$+$}
    child { node[label={[l]left:{$(v3)$}}] {$a$}}
    child { node[label={[l]right:{$(v4)$}}] {$0$}}
  }
  child { node[label={[l]right:{$(v2)$}}] {$c$}}
;

\node[n,label={[l]above:{$(\varepsilon)$}},xshift=4cm] {}
  child { node[n,label={[l]right:{$(v_0)$}}] {}
    child { node[n,label={[l]left:{$(v_0\underline{1})$}}] {}
      child { node[n,label={[l]left:{$(v_0\underline{1}v_1)$}}] {}
        child { node[n,label={[l]left:{$(v_0\underline{1}v_1\underline{1})$}}] {}
          child { node[n,label={[l]left:{$(v_0\underline{1}v_1\underline{1}v_3)$}}] {}
          edge from parent node[draw=none,left] {$a$}}
        edge from parent node[draw=none,left] {$\underline{1}$}}
        child { node[n,label={[l]right:{$(v_0\underline{1}v_1\underline{2})$}}] {}
          child { node[n,label={[l]right:{$(v_0\underline{1}v_1\underline{1}v_4)$}}] {}
          edge from parent node[draw=none,right] {$0$}}
        edge from parent node[draw=none,right] {$\underline{2}$}}
      edge from parent node[draw=none,left] {$+$}}
    edge from parent node[draw=none,left] {$\underline{1}$}}
    child { node[n,label={[l]right:{$(v_0\underline{2})$}}] {}
      child { node[n,label={[l]right:{$(v_0\underline{2}v_2)$}}] {}
      edge from parent node[draw=none,right] {$c$}}
    edge from parent node[draw=none,right] {$\underline{2}$}}
  edge from parent node[draw=none,left] {$\cdot$}}
;

\end{tikzpicture}
  \end{center}
  \caption{The representation of the $\Delta$-term $(a+0) \cdot c$ as a labelled tree following \cite{Caucal}}
  \label{fig:representation-caucal}
\end{figure}
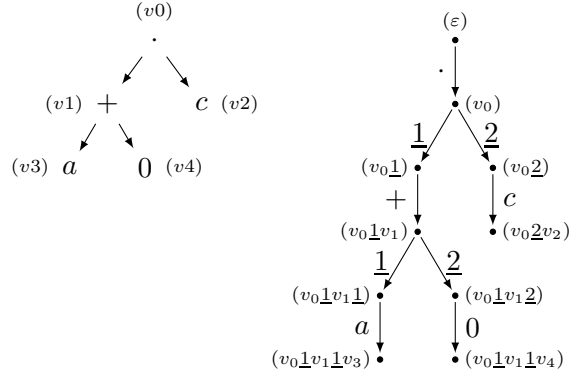

\begin{thm}[\cite{Carayolthesis,Caucal}]
\label{theorem:caucal}
Let $\Sigma$ be a signature. A $\Sigma$-term tree is defined by some
$\Sigma$-algebraic recursion scheme if and only if the tree representing it belongs to $\Tree_{2}$.
\end{thm}

\begin{rem} 
\label{rem:ht-caucal}
It immediately follows that the synchronization tree  $H(t)$ associated to a $\Gamma$-term tree $t$ defined in Section~\ref{sec:morphism} belongs to 
$\Tree_{2}$ (resp. $\Tree_{1}$) whenever $t$ is defined by a $\Gamma$-algebraic (resp. $\Gamma$-regular) recursion scheme. Indeed the transformation going from the presentation of \cite{Caucal} to ours is an MSO-interpretation that commutes with the unfolding operation.
\end{rem}

We conclude with some properties of the trees of $\Tree_{2}$ and their contractions.

\begin{prop}
\label{prop:trees-in-trees2}
The following properties hold:
\begin{enumerate}
\item The contraction of a tree in $\Tree_{2}$ is bisimilar to some tree in
$\Tree_{2}$.
\item Every tree in $\Tree_{2}$ can be obtained as the contraction of a deterministic tree in $\Tree_{2}$. 

\end{enumerate}
\end{prop}

\begin{proof}
As illustrated in Figure~\ref{fig:counter-example}, the contraction
operation does not in general commute with the unfolding
operation. However, it is easy to show that the contraction of the
unfolding from a root $r$ and the unfolding of the contraction from
the same root are bisimilar. 

\begin{figure}
  \begin{center}
\begin{tikzpicture}[transform shape,scale=1.0,edge from parent/.style={draw,-latex}]
 \tikzstyle{n}=[circle, draw=black, fill=black, inner sep = 0.25mm, outer sep= 0.5mm]
 \tikzstyle{e}=[ inner sep = 0mm, outer sep= 0mm]
 \tikzstyle{l}=[font=\tiny,label distance=-1mm]

      \begin{scope}
        \matrix[matrix of nodes,column sep={0.8cm,between origins},row sep={0.8cm,between origins},anchor=north]{
          &  |[n,label={[l]above:{$(r)$}}](r)|  &     \\
         |[n](s)|  &     &  |[n](t)| \\
          & |[n](u)| &  \\}; 
        
        \draw (r) edge[-latex] node[above left] {$e$} (s)
              (r) edge[-latex] node[above right] {$e$}(t)
              (s) edge[-latex] node[below left] {$a$}(u)
              (t) edge[-latex] node[below right] {$a$}(u)
              (u) edge[-latex] node[right] {$e$} (r);
      \end{scope}

      \begin{scope}[xshift=3cm]
         \tikzstyle{level}=[level distance=0.85cm,sibling distance=1cm]
        \node[n,label={[l]above:{$(r)$}}] {}
        child {
          node[n] {}
          child {
            node[n] {}
             child {
               node[e] {}
               edge from parent[dashed,-]}
            edge from parent node[left] {$a$}}
          edge from parent node[left] {$a$}};
      \end{scope}
      \begin{scope}[xshift=7cm]
         \tikzstyle{level}=[level distance=0.85cm,sibling distance=0.85cm]
         \tikzstyle{level 1}=[level distance=0.85cm,sibling distance=1.75cm]
        \node[n,label={[l]above:{$(r)$}}] {}
  child { node[n] {}
    child { node[n] {}
      child { node[e] {} edge from parent[dashed,-]}
      child { node[e] {} edge from parent[draw=none]}
    edge from parent node[draw=none,above left] {$a$}}
    child { node[n] {}
      child { node[e] {} edge from parent[draw=none]}
      child { node[e] {} edge from parent[dashed,-]}
    edge from parent node[draw=none,above right] {$a$}}
  edge from parent node[draw=none,above left] {$a$}}
  child { node[n] {}
    child { node[n] {}
      child { node[e] {} edge from parent[dashed,-]}
      child { node[e] {} edge from parent[draw=none]}
    edge from parent node[draw=none,above left] {$a$}}
    child { node[n] {}
      child { node[e] {} edge from parent[draw=none]}
      child { node[e] {} edge from parent[dashed,-]}
    edge from parent node[draw=none,above right] {$a$}}
  edge from parent node[draw=none,above right] {$a$}}
;
      \end{scope}
    \end{tikzpicture}
  \end{center}

\caption{A graph $G$ (on the left), the unfolding from $r$ of its $\{e\}$-contraction from $r$ (in the middle) and the $\{e\}$-contraction of its unfolding from $r$ (on the right).}\label{fig:counter-example}
\end{figure}
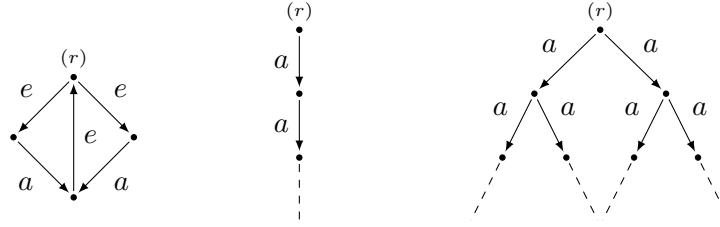

{
For the first property, let $t$ be a tree in $\Tree_{2}$ which is obtained by 
unfolding a graph $G$ in $\Graph_{1}$ from one of its roots $r$. Let $t'$ be the $B$-contraction of $t$ from its root. The $B$-contraction $G'$ of $G$ from $r$ also belongs to $\Graph_{1}$ (as the $B$-contraction is a particular case of MSO-interpretation). From the previous remark, we have that $t'$ is bisimilar to $\Unf(G',r)$ which belongs to $\Tree_2$ as $G$ belongs to $\Graph_1$.}

For the second property, let $t$ be a tree in $\Tree_{2}$ which is obtained by unfolding a graph $G$ in $\Graph_{1}$ from one of its root $r$. Furthermore assume that both $t$ and $G$ are labelled by $A$. By Proposition~\ref{prop:det-closure}, $G$ can be obtained by $B$-contraction of a deterministic graph $H \in \Graph_{1}$ from one of its roots $r$.

For the $B$-contraction to commute with the unfolding, it is enough that the graph $H$ satisfies, for any label $a \in A$ and any two vertices $u$ and $v$ of $H$ belonging to the $B$-contraction, that there exists at most one path labelled in $B^{*}a$ from $u$ to $v$. This is the case for the graph $H$ obtained from Proposition~\ref{prop:det-closure}. Hence $t$ is isomorphic to
the $B$-contraction of  $\Unf(H,r')$, which is a deterministic tree in $\Tree_{2}$.\end{proof}

\subsection{Synchronization trees in the Caucal hierarchy}
\label{Sect:treesinCH}

The $\Gamma$-regular synchronization trees, being the regular trees with potentially infinite degree, belong to $\Graph_{1}$. As $\Tree_{1}$
only contains the regular trees of finite degree, it does not contain 
the class of $\Gamma$-regular synchronization trees. In this section, we show that the classes of  $\Delta$-regular and $\Gamma$-algebraic synchronization trees are included in
$\Graph_{2}$ but not in $\Tree_{2}$. Similarly, we establish that the class of $\Delta$-algebraic trees is included in $\Graph_{3}$ but not in $\Tree_{3}$.

\begin{prop}
\label{prop:carac}
The $\Gamma$-algebraic (and hence the $\Delta$-regular) synchronization
trees are contractions of trees in $\Tree_{2}$ and hence belong to $\Graph_{2}$.
\end{prop}

\begin{proof}
Let $E$ be a $\Gamma$-algebraic recursion scheme defining a $\Gamma$-term tree
$t$ and a synchronization tree $t'$. From Proposition~\ref{prop-Gamma-hom}, $t'$ is obtained by contraction of the synchronization tree $H(t)$ with respect to $\{+_1,+_2\}$.  From \cite[Theorem 3.5]{Caucal} (cf. Remark~\ref{rem:ht-caucal}), we know that $H(t)$ belongs to $\Tree_{2}$, and, as
the contraction operation is a particular case of MSO-interpretation, $t'$ belongs to $\Graph_{2}$.
  \end{proof}

\begin{prop}
There is a $\Delta$-regular synchronization tree which is not in 
$\Tree_2$. 
\end{prop}

\begin{proof}
Consider the following $\Delta$-regular recursion scheme:
\begin{eqnarray*}
G &=& 1 + a\cdot G \cdot (1 + 1). 
\end{eqnarray*}
The synchronization tree $t$ defined by it has a single infinite branch $u_0,u_1,\ldots$ 
(with edges labelled a). The out-degree of each vertex $u_i$ is 
$2^i +1$ since it is the source of $2^i$ edges labelled $\ex$.  

Assume, towards a contradiction, that $t$ belongs to $\Tree_2$. 
Hence $t$ is isomorphic to the unfolding of some graph $H$ in $\Graph_1$
from a vertex $v_0$. By Remark~\ref{rem:rooted-graphs}, we can assume 
that all vertices are reachable from $v_0$ and hence that they all have finite 
out-degree. This implies that $H$ contains a path $v_0,v_1\ldots$ along which the out-degree of the vertices grows exponentially (i.e. the out-degree of $v_i$ is $2^i +1$).
The following lemma shows that this is not possible in $\Graph_{1}$.

\begin{lem}
  Let $G$ be a graph in $\Graph_{1}$ whose vertices have finite
  out-degree.  Let $(u_i)_{i \geq 0}$ be an infinite path in $G$,
  possibly with repetitions.  Then there exists a constant $C\geq 0$
  such that for all $i \geq 0$, the out-degree of $u_i$ is at most
  $C(i+1)$.
\end{lem} 

\begin{proof}
  By Proposition~\ref{prop:level-one}, we can assume without loss of
  generality that $G$ is a prefix-recognizable graph. Let
  $U_1(V_1\times W_1),\ldots,U_n(V_n \times W_n)$ be the relations
  involved in the definition of $G$. As the vertices of $G$ have
  finite out-degree, without loss of generality we may assume that all
  the $W_j$ are finite sets. Let $d$ denote the length of the longest
  word appearing in the $W_j$, and let $W$ denote the number of words
  appearing in the $W_j$. Then clearly $|u_{i+1}| \leq |u_i| + d$ for
  all $i\geq 0$, where for a word $u$ we denote its length by $|u|$.
  Thus, introducing the notation $d' = \max\{d,|u_0|\}$, we have
  $|u_i| \leq d'(i+1)$ for all $i \geq 0$. As the out-degree of a
  vertex $u$ of $G$ is at most $(|u|+1)W$, we conclude that the
  out-degree of $u_i$ is at most $C(i+1)$ for all $i \geq 0$ with $C =
  (d'+1)W$.    \end{proof}
\end{proof}

\begin{prop}\label{Prop:Delta-alg-in-Graph3}
All the $\Delta$-algebraic synchronization trees are in $\Graph_{3}$
and therefore
 have a decidable MSO-theory.
\end{prop}

\begin{proof}
Let $E$ be an algebraic scheme over $\Delta$. Let $t$ be the $\Delta$-term 
defined by $E$, and $t'$ be the synchronization tree defined by $E$.
By the Mezei-Wright theorem and Proposition~\ref{prop-tau}, we have that 
$t' = \tau(t)$. More precisely, the tree $t'$ is obtained by unfolding the graph
$G(t)$ (as defined in Section~\ref{sec:morphism}) from the vertex $v_{0}$ and 
then applying a contraction with respect to $\{1\}$. The graph $G(t)$ can be interpreted in the tree representing $t$, which belongs to $\Tree_{2}$ (by Theorem~\ref{theorem:caucal}), and hence $G(t)$ belongs to $\Graph_{2}$. The
unfolding of $G(t)$ from $v_{0}$ belongs to $\Tree_{3}$ and its contraction $t'$ to  $\Graph_{3}$ (as contractions are particular cases of MSO-interpretations).
\end{proof}

\begin{prop}
There exists a $\Delta$-algebraic synchronization tree which does not belong to
$\Tree_{3}$.
\end{prop}

\begin{proof}
In this proof, we consider $\{a,b\}$-labelled trees  of a particular shape illustrated in Figure~\ref{fig:td0}. These trees have a unique infinite branch whose edges are labelled $a$. In addition the vertex at depth $n$ along this branch has $d(n)$ outgoing edges labelled $b$, for some mapping $d : \mathbb{N} \rightarrow \mathbb{N}$. Up to isomorphism, the tree is entirely characterized by the mapping $d$ and is denoted $t_d$.

In  \cite[Theorem~4.5.3]{Braudthesis}, Braud gives a necessary condition for a tree  of this form\footnote{Actually the condition holds for a larger class of graphs called $\#$-graph-combs which encompasses all trees $t_d$.} to belong to $\Graph_2$. Namely, if a tree $t_d$, for $d : \mathbb{N} \rightarrow \mathbb{N}$ belongs to $\Graph_2$ then there exists a constant $c>0$ such that $d(n) < 2^{c(n+1)}$ for all $n \geq 0$. 

Our proof goes as follows. We first show that for the mapping $d_0 : n \mapsto 2^{2^n}$, the tree $t_{d_0}$, depicted in Figure~\ref{fig:td0}, is $\Delta$-algebraic. Towards a contradiction, we assume that $t_{d_0}$ belongs to $\Tree_3$. We then show there  would exist a mapping $d_1$ satisfying $d_1(n) \geq 2^{2^{n-1}}$, $n \geq 1$ such that $t_{d_1}$ belongs to $\Graph_2$. This statement  contradicts Braud's condition as there cannot exist a constant $c$ such that $2^{2^{n-1}} < 2^{c(n+1)}$, for all $n \geq 1$.

Consider the $\Delta$-algebraic scheme
\[
\begin{array}{lcl}
S & = & G (1+1) \\
G(v) & = & a \cdot G (v \cdot v) + v \cdot {(b \cdot 0)}  \\
\end{array}
\]
The tree defined by this scheme is isomorphic to the tree $t_{d_0}$ with $d_0(n)=2^{2^n}$ for all $n \geq 0$.

\begin{figure}[ht]
\begin{center}
	\begin{tikzpicture}[transform shape,scale=1.0]
 \tikzstyle{n}=[circle, draw=black, fill=black, inner sep = 0.25mm, outer sep= 0.5mm]
 \tikzstyle{e}=[ inner sep = 0mm, outer sep= 0mm]
      	\node[n] (a0) at (0,0) {};  	
      	\node[n] (a1) at (3,0) {};  	
      	\node[n] (a2) at (6,0) {};
      	\node[n] (b00) at (-0.5,-1) {};  	
      	\node[n] (b01) at (0.5,-1) {};  	
      	\node[n] (b10) at (2,-1) {};
        \node[n] (b11) at (2.75,-1) {};    
        \node[n] (b12) at (3.25,-1) {};    		
      	\node[n] (b13) at (4,-1) {};
      	\node[n] (b20) at (4.5,-1) {};
      	\node[n] (b21) at (7.5,-1) {};
      	\draw[-latex] (a0) to[above] node {$a$} (a1);
      	\draw[-latex] (a1) to[above] node {$a$} (a2);
      	\draw[-latex] (a0) to[left] node {$b$} (b00);
      	\draw[-latex] (a0) to[right] node {$b$} (b01);
        \draw[-latex] (a1) to[left] node {$b$} (b10);
        \draw[-latex] (a1) to[left] node {$b$} (b11);
        \draw[-latex] (a1) to[right] node {$b$} (b12);
       	\draw[-latex] (a1) to[right] node {$b$} (b13);
       	\draw[-latex] (a2) to[left] node[yshift=1mm] {$b$} (b20);
       	\draw[-latex] (a2) to[right] node[yshift=1mm] {$b$} (b21);
      	\draw[dashed] (a2) -- (9,0); 
      	\draw[dashed] (a2) -- (5.5,-1);   
        \draw[dashed] (a2) -- (6.5,-1);   
		\draw[thick,decoration={brace,mirror},decorate]  (b20) to node[below] {$d_0(2)=16$} (b21) ;

 \end{tikzpicture}
 \end{center}
 \caption{The tree $t_{d_0}$ where $d_0(n)=2^{2^n}$ for all $n \geq 0$.\label{fig:td0}}
 \end{figure}
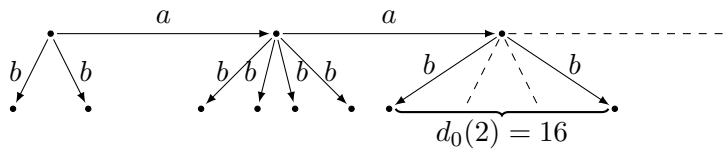

Assume now, towards a contradiction, that $t_{d_0}$ belongs to $\Tree_{3}$. Consider the graph $G$ in $\Graph_{2}$ from which $t$ can be obtained by unfolding from one of its roots $r$. 

For all $n \geq 0$, we denote by $a_n$ the unique vertex in $G$ such that $r \era{a^n} a_n$ and by $V_n$ the set of vertices $\{ v \mid r \era{a^nb} v \}$.
As $\Unf(G,r) \approx t_{d_0}$, we must have $a_n \neq a_m$ for $n \neq m$ and $|V_n|=2^{2^n}$ for all $n \geq 0$. The set of vertices of $G$ is therefore equal to $\{ a_n \mid  n \geq 0\} \cup \cup_{n \geq 0} V_n$ and its set of edges is:
\[
 \begin{array}{cl}
  & \{ (a_n,a,a_{n+1}) \mid n \geq 0 \}\\
\cup  & \{ (a_n,b,v) \mid n \geq 0 \;\textrm{and}\; v \in V_n \} \\	
 \end{array}
\]

Note that the sets $V_n$ are not necessarily pairwise disjoint and that in particular, $G$ is not necessarily a tree, as illustrated in Figure~\ref{fig:transform}. 

\begin{figure}[ht]
\begin{center}
	\begin{tikzpicture}[transform shape,scale=1.0]
 \tikzstyle{n}=[circle, draw=black, fill=black, inner sep = 0.25mm, outer sep= 0.5mm]
 \tikzstyle{e}=[ inner sep = 0mm, outer sep= 0mm]
 \begin{scope}
      	\node[n] (a0) at (0,0) {};  	
      	\node[n] (a1) at (2,0) {};  	
      	\node[n] (a2) at (4,0) {};
      	\node[n] (b00) at (-0.5,-1) {};  	
      	\node[n] (b01) at (0.5,-1) {};  	
      	\node[n] (b10) at (1.5,-1) {};
        \node[n] (b11) at (2,-1) {};    
        \node[n] (b12) at (2.5,-1) {};    		
      	\draw[-latex] (a0) to[above] node {$a$} (a1);
      	\draw[-latex] (a1) to[above] node {$a$} (a2);
      	\draw[-latex] (a0) to[left] node {$b$} (b00);
      	\draw[-latex] (a0) to[right] node {$b$} (b01);
        \draw[-latex] (a1) to[above] node {$b$} (b01);
        \draw[-latex] (a1) to[left] node {$b$} (b10);
        \draw[-latex] (a1) to[right] node[xshift=-1mm] {$b$} (b11);
       	\draw[-latex] (a1) to[right] node {$b$} (b12);
       	\draw[-latex] (a2) to[above] node[xshift=4mm,yshift=1mm] {$b$} (b01);
       	\draw[-latex] (a2) to[below] node {$b$} (b12);
      	\draw[dashed] (a2) -- (5,0); 
      	\draw[dashed] (a2) -- (3.5,-1);   
      	\draw[dashed] (a2) -- (4,-1);   
        \draw[dashed] (a2) -- (4.5,-1);   
\end{scope}

\begin{scope}[xshift=6.5cm]
      	\node[n] (a0) at (0,0) {};  	
      	\node[n] (a1) at (2,0) {};  	
      	\node[n] (a2) at (4,0) {};
      	\node[n] (b00) at (-0.5,-1) {};  	
      	\node[n] (b01) at (0.5,-1) {};  	
      	\node[n] (b10) at (1.5,-1) {};
        \node[n] (b11) at (2,-1) {};    
        \node[n] (b12) at (2.5,-1) {};    		
      	\draw[-latex] (a0) to[above] node {$a$} (a1);
      	\draw[-latex] (a1) to[above] node {$a$} (a2);
      	\draw[-latex] (a0) to[left] node {$b$} (b00);
      	\draw[-latex] (a0) to[right] node {$b$} (b01);
        \draw[-latex] (a1) to[left] node {$b$} (b10);
        \draw[-latex] (a1) to[right] node[xshift=-1mm] {$b$} (b11);
       	\draw[-latex] (a1) to[right] node {$b$} (b12);
      	\draw[dashed] (a2) -- (5,0); 
      	\draw[dashed] (a2) -- (3.5,-1);   
      	\draw[dashed] (a2) -- (4,-1);   
        \draw[dashed] (a2) -- (4.5,-1);   
\end{scope}
 \end{tikzpicture}
 \end{center}
 \caption{An example of a possible graph $G$ whose unfolding is isomorphic to $t_{d_0}$ (on the left) and of the tree $\mathcal{I}(G)$ (on the right).\label{fig:transform}}
 \end{figure}
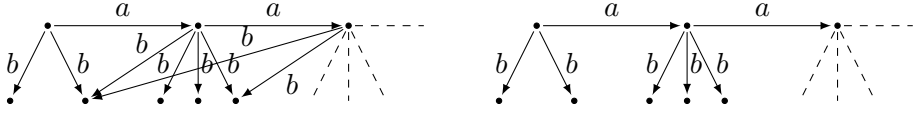

Consider the MSO-interpretation $\mathcal{I}$ that erases the $b$-labelled edges $(u,v)$ whenever there exists a vertex $s$ such that $s \era{a^+} u$ and $s \era{b} v$ and keeps all other edges unchanged. By applying $\mathcal{I}$ to $G$ (as illustrated in Figure~\ref{fig:transform}, we obtain
a tree $\mathcal{I}(G) \in \Graph_2$ with the same set of vertices as $G$. This tree has a unique infinite branch $(a_n)_{n \geq 0}$ and, for all $n \geq 0$, the vertex $a_n$ has 
outgoing edges labelled $b$ to the vertices in
\[
	V_n'= V_n \setminus (\cup_{m<n} V_m)
\]
The tree $\mathcal{I}(G)$ is hence isomorphic to $t_{d_1}$ where $d_1(n) = |V_n'|$ for all $n \geq 0$. To obtain a contradiction with \cite[Theorem~4.5.3]{Braudthesis}, we remark that for all $n \geq 1$, we have:
\[
\begin{array}{lcl}
  d_1(n)=|V_n'|  & \geq &  2^{2^n} - \sum_{m=0}^{n-1} 2^{2^m} \\[1mm]
          & \geq &  2^{2^n} - n \cdot 2^{2^{n-1}} \\
          & \geq  &  2^{2^{n-1}} ( 2^{2^{n-1}} - n) \\
          & \geq  &  2^{2^{n-1}}. \\
\end{array}
\]
  \end{proof}

\subsection{Contractions of synchronization trees in $\Tree_{2}$}
\label{sec:charac-gamma-algebraic}

In this section, we prove Theorem~\ref{thm:carac-gamma-algebraic} which states
that $\Gamma$-algebraic synchronization trees are the contractions of trees in $\Tree_{2}$. We start by showing that the synchronization trees in $\Tree_{2}$
are $\Gamma$-algebraic.

\begin{prop}
\label{prop:tree2-inc-alg}
Synchronization trees in $\Tree_{2}$ are $\Gamma$-algebraic.
\end{prop}

\begin{proof}
Let $t$ be a synchronization tree in $\Tree_{2}$. By the first property of Proposition~\ref{prop:trees-in-trees2}, $t$ is the $B$-contraction of a deterministic tree $t' \in \Tree_{2}$ labelled by $A \cup \{\ex\}$. As $\Gamma$-algebraic synchronization trees are closed under contraction (see Proposition~\ref{prop:contr-alg}), it is enough to show that $t'$ is $\Gamma$-algebraic.

We are going to show that a tree $t''$ representing a $\tilde{\Gamma}$-term tree defining $t'$ belongs to $\Tree_{2}$. {Recall that $\tilde{\Gamma}$ is the signature 
 $\{ +^{n} \mid n \geq 1 \} \cup A \cup
  \{0,1\}$ allowing for sums of arbitrary arity that was introduced in Remark~\ref{rem:arbitrary-rank-sum} on page~\pageref{rem:arbitrary-rank-sum}.}
 Thanks to Theorem~\ref{theorem:caucal}, this implies that $t''$ is defined by a $\tilde{\Gamma}$-algebraic recursion scheme. In turn this implies that $t'$ is $\tilde{\Gamma}$-algebraic and hence $\Gamma$-algebraic (see Remark~\ref{rem:arbitrary-rank-sum}).

Let $<$ be an arbitrary total order on $A$. The tree $t''$ is obtained
by applying a transduction $\trans$ to $t'$. This transduction $\trans$ does the following:
\begin{itemize}
\item Whenever a vertex $v$ is the target of an edge labelled in $A$ 
and is not the source of an edge, then the transduction adds a new edge labelled $0$ 
from $v$ to a new vertex introduced by the transduction. 
\item For every vertex $u$ with $k\geq 1$ outgoing edges labelled $a_{1}< \ldots < a_{k} \in A$, respectively, going to vertices $u_{1},\ldots,u_{k}$, the transduction adds new vertices $v$, $v_{1},\ldots,v_{k}$, $v_{1}',\ldots,v_{k}'$ and
edges $u \era{+^{k}} v$, $v \era{\underline{i}} v_{i}$ and $v_{i} \era{a_{i}} v_{i}'$. If $u$ has an outgoing edge labelled by $\ex$ to a vertex $v$ then the transduction adds a vertex $v'$ and two edges $u \era{\underline{k+1}} v'$ and 
$v' \era{1} v$.
\item All the edges of the original structure are removed.
\end{itemize}

The tree $t''$ obtained by applying  $\trans$ to $t'$ represents a term
tree that defines $t'$. As $t'$ belongs to $\Tree_{2}$, it is (up to isomorphism) obtained by unfolding a graph $G \in \Graph_{1}$ from one of its roots $r$. {Furthermore it can be checked that the transduction $\trans$ commutes with the unfolding operation.} Hence the tree $t''$ is 
isomorphic to $\Unf(\trans(G),r)$ and therefore belongs to $\Tree_{2}$ (as $\trans(G)$ belongs to $\Graph_{1}$). 
\end{proof}

\begin{thm}
\label{thm:carac-gamma-algebraic}
The contractions of the synchronization trees in $\Tree_{2}$ are the 
$\Gamma$-algebraic synchronization trees. 
\end{thm}

\begin{proof}
By Propositions~\ref{prop:tree2-inc-alg} and~\ref{prop:contr-alg}, we have that each contraction of 
a synchronization tree in $\Tree_{2}$ is $\Gamma$-algebraic. For the converse, 
a $\Gamma$-algebraic synchronization tree is defined as the contraction
with respect to $\{ +_{1},+_{2} \}$ of the $\Gamma$-term tree defined by an
algebraic scheme over $\Gamma$. From \cite[Theorem 3.5]{Caucal}, such a $\Gamma$-term tree
belongs to $\Tree_2$. Moreover, it may be seen as a synchronization tree 
over the alphabet which contains, in addition to the letters in $A$, 
the symbols $+_1$ and $+_2$.
  \end{proof}

Thanks to the second property of Proposition~\ref{prop:trees-in-trees2}, we have the following corollary.

\begin{cor}
Every $\Gamma$-algebraic tree is bisimilar to a tree in $\Tree_{2}$.
\end{cor}

\section{Branch languages of bounded synchronization trees}\label{Sect:branchlang}

Call a synchronization tree \textbf{bounded} if there is a constant
$k$ such that the outdegree of each vertex is at most $k$.  Our aim in
this section will be to offer a language-theoretic characterization of
the expressive power of $\Gamma$-algebraic recursion schemes defining
synchronization trees. We shall do so by following Courcelle---see,
e.g.,~\cite{Courcelle83}---and studying the branch languages of
synchronization trees whose vertices have bounded outdegree. More
precisely, we assign a family of branch languages to each bounded
synchronization tree over an alphabet $A$ and show that a bounded tree
is $\Gamma$-algebraic if, and only if, the corresponding language
family contains a deterministic context-free language
(DCFL). Throughout this section, we will call $\Gamma$-algebraic trees
just algebraic trees, and similarly for regular trees.

\begin{defi}
Suppose that $t=(V,v_0,E,l)$ is a bounded synchronization tree 
over the alphabet $A$. Denote by  $k$ the maximum 
of the outdegrees of the vertices of $t$.
Let $B$ denote the alphabet $A \times [k]$.
A \textbf{determinization of $t$} is a tree $t'=(V,v_0,E,l')$ 
over the alphabet $B$ which differs from $t$ only in the 
labelling as follows. Suppose that $v \in V$ with outgoing edges 
$(v,v_1),\ldots,(v,v_\ell)$ labelled 
$a_1,\ldots,a_\ell \in\A \cup \{\ex\}$ in $t$.  Then there is some
permutation $\pi$ of the set $[\ell]$ such that the label of each
$(v,v_i)$ in $t'$ is $(a_i,\pi(i))$.

Consider a determinization $t'$ of $t$.
Let $v \in V$ and let $v_0,v_1,\ldots,v_m = v$ 
denote the vertices on the unique path from the root to $v$.
The \textbf{branch word} corresponding to $v$ in $t'$ 
is the alternating word
\[k_0(a_1,i_1)k_1\ldots k_{m-1}(a_m,i_m)k_m\]
where $k_0,\ldots,k_m$ denote the outdegrees of the
vertices $v_0,\ldots,v_m$, and for each $j \in [m]$,
$(a_j,i_j)$ is the label of the edge $(v_{j-1},v_j)$ 
in $t'$.
The \textbf{branch language} $L(t')$ corresponding
to a determinization $t'$ of $t$ consists of all
branch words of $t'$.

Finally, the family of branch languages corresponding to
$t$ is:
\[\L(t) = \{L(t') \mid t'\ {\it is}\ {\it a}\ {\it determinization}\ {\it of}\ t \}.\]
\end{defi}

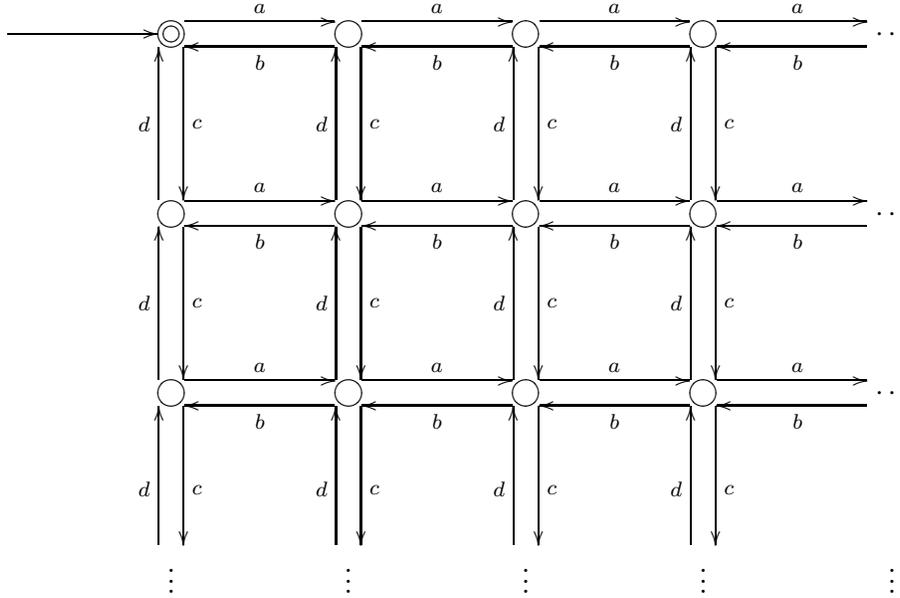
\begin{figure} 
\newcommand{\state}{*+[o][F-]{\makebox(4,4){}}}
\newcommand{\istate}{*+[o][F-]{\makebox(4,4){}}}
\newcommand{\fstate}{*+[o][F=]{\makebox(4,4){}}}
\[\xymatrix@C=12ex@R=12ex{
\ar[r] & \fstate{} \ar@<1ex>[d]^c \ar@<1ex>[r]^a & \state{} \ar@<1ex>[d]^c \ar@<1ex>[r]^a \ar@<1ex>[l]^b  & \state{} \ar@<1ex>[d]^c \ar@<1ex>[r]^a \ar@<1ex>[l]^b & \state{} \ar@<1ex>[d]^c \ar@<1ex>[r]^a \ar@<1ex>[l]^b & \cdots  \ar@<1ex>[l]^b \\
 & \state{} \ar@<1ex>[u]^d \ar@<1ex>[r]^a \ar@<1ex>[d]^c & \state{} \ar@<1ex>[u]^d \ar@<1ex>[r]^a \ar@<1ex>[l]^b \ar@<1ex>[d]^c &  \state{} \ar@<1ex>[u]^d \ar@<1ex>[r]^a \ar@<1ex>[l]^b \ar@<1ex>[d]^c & \state{} \ar@<1ex>[u]^d  \ar@<1ex>[d]^c \ar@<1ex>[r]^a \ar@<1ex>[l]^b & \cdots \ar@<1ex>[l]^b \\
 & \state{} \ar@<1ex>[u]^d \ar@<1ex>[r]^a \ar@<1ex>[d]^c & \state{} \ar@<1ex>[u]^d \ar@<1ex>[r]^a \ar@<1ex>[l]^b \ar@<1ex>[d]^c &  \state{} \ar@<1ex>[u]^d \ar@<1ex>[r]^a \ar@<1ex>[l]^b \ar@<1ex>[d]^c & \state{} \ar@<1ex>[u]^d \ar@<1ex>[r]^a \ar@<1ex>[l]^b \ar@<1ex>[d]^c & \cdots \ar@<1ex>[l]^b \\
&  \vdots \ar@<1ex>[u]^d & \vdots \ar@<1ex>[u]^d & \vdots \ar@<1ex>[u]^d & \vdots \ar@<1ex>[u]^d & \vdots 
}
\]
\caption{An LTS whose unfolding is not an algebraic synchronization tree. {The initial node as outdegree 3 (recall that there is an implicit outgoing edge labelled by $\ex$), the other nodes on the top-most row and on the left-most column also have outdegree 3 and all other "internal" nodes have outdegree 4.}}
\label{fig:bag}
\end{figure}

By way of example, consider the LTS depicted in
Figure~\ref{fig:bag}. This LTS describes the behaviour of a bag over a
{four}-letter alphabet, when we consider $b$ to stand for the output of
an item that was input via $a$, and $d$ to signal the output of an item
that was input via $c$.  The synchronization tree $t_{\text{bag}}$
that is obtained by unfolding this LTS from its start state is
bounded. In fact, the outdegree of each non-leaf node is three or four. The
{branch words} corresponding to the nodes of any determinization of
the tree $t_{\text{bag}}$ have the form
\[
k_0(a_1,i_1)k_1\ldots k_{m-1}(a_m,i_m)k_m ,
\]
where $k_0=3$, $k_m\in \{3,4,0\}$, $k_i \in \{3,4\}$ for each $i \in [m-1]$, each $i_j \in [k_{j-1}]$ for $j \in [m]$
and $a_1\ldots a_m$ is a word with
the property that, in any of its prefixes, the number of occurrences
of the letter $a$ is greater than or equal to the number of
occurrences of the letter $b$, and the number of occurrences of the
letter $c$ is greater than or equal to the number of occurrences of
the letter $d$. Moreover, for each $j \in [m]$, 
$a_j = \ex$ if and only if $j = m$ and $k_m = 0$. (Note that, when
$a_m = \ex$, the number of $a$'s in $a_1\ldots a_{m-1}$ equals
the number of $b$'s, and similarly for $c$ and $d$.)

\begin{thm}
\label{thm-DCFL}
A bounded synchronization tree $t$ is algebraic (respectively, regular) 
iff $\L(t)$ contains a DCFL (respectively, regular language).
\end{thm}

In the proof, we make use of the following construction.
Suppose that $t =(V,v_0,E,l)$ is a bounded synchronization tree 
over $A$. Let $k$ be defined as above, and let $\Sigma$ be the 
signature containing the symbol $+_i$ of rank $i$ for each $i \in [k]$, 
the constant symbols $0$ and $1$, and the letters of $A$ as symbols of
rank $1$. We will sometimes write $+_0$ for $0$. 
Then each determinization $t'$ of $t$ naturally corresponds 
to an `alternating'  \emph{term tree} $T_{t'}$ in the initial continuous categorical 
$\Sigma$-algebra $T_\Sigma^\omega$. As a partial function $\N^* \to \Sigma$, 
the term tree $T_{t'}$ 
is defined as follows. Consider a vertex $v \in V$ with corresponding
branch word $k_0(a_1,i_1)\ldots(a_n,i_n)k_n$. Then $T_{t'}$ is defined on 
both words $i_11\ldots i_{n}1$ and $i_11\ldots 1i_n$, 
and 
\begin{eqnarray*}
T_{t'}(i_11\ldots i_n1) &=& +_{k_n}\\
T_{t'}(i_11\ldots 1i_n) &=&
\left\{ 
\begin{array}{ll}
a_{n} & {\rm if}\ a_n \in A\\
1 &{\rm if}\ a_n = \ex . 
\end{array}
\right.
\end{eqnarray*}
In addition, $T_{t'}$ is defined on the empty word $\epsilon$ 
and $T_{t'}(\epsilon) = +_{k_0}$, where ${k_0}$ is the outdegree of the root.

\begin{lem}
\label{lem-A}
Suppose that $t\in \Sta$ is bounded and has a determinization $t'$
such that the $\Sigma$-term tree $T_{t'}$ is algebraic (regular, resp.). Then $t$ is
algebraic (regular, resp.).
\end{lem}
\begin{proof}
We can turn $\Sta$ into a continuous categorical
$\Sigma$-algebra by defining 
\[+_i(t_1,\ldots,t_i) = t_1 + \cdots + t_i\] for each $i \in [k]$ and
$t_1,\ldots,t_i \in \Sta$, and similarly for morphisms between 
trees in $\Sta$.   Now up to natural isomorphism, there is a
unique categorical $\Sigma$-algebra morphism
$h_\Sigma:T_\Sigma^\omega \to \Sta$. For any $t \in \Sta$ and for any
determinization $t'$ of $t$, $h$ maps $T_{t'}$ to a tree isomorphic to
$t$.  Since, by the Mezei-Wright theorem \cite{BEmezei}, 
$h_\Sigma$ preserves
algebraic and regular objects, if $T_{t'}$ is algebraic or regular,
then so is $t$.   
\end{proof}

\begin{lem}
\label{lem-B}
Suppose that $t \in \Sta$ is bounded and algebraic (regular, resp.).
Then $t$ has a determinization $t'$ such that the $\Sigma$-term tree $T_{t'}$ is 
algebraic (regular, resp.).
\end{lem}
\begin{proof}
Suppose that $t$ is algebraic and bounded by $k$. Then,
by the Mezei-Wright theorem \cite{BEmezei},  there is an
algebraic term tree $T_0\in T_\Gamma^\omega$ such that $h_\Gamma(T_0)$ is isomorphic
to $t$, where $h_\Gamma$ denotes the essentially unique continuous categorical 
$\Gamma$-algebra morphism $T_\Gamma^\omega \to \Sta$. We want to show that there
is an alternating algebraic $\Sigma$-term tree $T_1 \in T_\Sigma^\omega$ 
such that $h_\Sigma(T_1)$ is isomorphic to $t$, where $h_\Sigma$ is the essentially
unique continuous categorical $\Sigma$-algebra morphism $T_\Sigma^\omega \to \Sta$. 
Since such a term tree $T_1$ is $T_{t'}$ for some determinization $t'$ of $t$,
this completes the proof. 

We construct $T_1$ from $T_0$ in two steps. First, we replace all 
maximal subterms all of whose vertices are labelled $+$ or $0$ by a 
single vertex labelled $0$ to obtain a $\Gamma$-term tree $T_0'$.
Second, we consider vertices of $T_0'$ labelled 
$+$ whose parent, if any, is labelled in $A$. Since $t$ is 
bounded by $k$, the subterm rooted at such a vertex can be 
written as $s_0(S_1,\ldots,S_n)$, where $s_0$ is a finite 
term all of whose vertices are labelled $+$ or $0$ or a variable 
in the set $\{v_1,\ldots,v_n\}$ for some $n \leq k$, whose frontier word is 
$v_1\cdots v_n$, and each $S_i$ is a $\Gamma$-term tree whose root is labelled 
in $A \cup \{ 1 \}$. We replace each such vertex by a vertex labelled
$+_n$ having $n$ outgoing edges labelled $1,\ldots,n$ connecting this vertex to 
the roots of $S_1,\ldots,S_n$, respectively.

The first transformation is a monadic colouring
\cite{Braudthesis,CarayolWohrle} (a special case of monadic
interpretation, which is also known as monadic
marking~\cite{CarayolWohrle}), since there is a monadic second-order
formula $\phi(x)$ characterizing those vertices $x$ of $T_0$ such that
all vertices of the subterm rooted at $x$ are labelled $+$ or $0$, but
any other subterm containing $x$ has a vertex labelled in $A \cup
\{1\}$:
\begin{eqnarray*}
\forall y. (x \leq y \Rightarrow ( l_+(y) \vee l_0(y))
\wedge 
\forall y. (y < x \Rightarrow (\exists z. y \leq z \wedge \bigvee_{a \in A\cup \{1\}} l_a(z)) 
\end{eqnarray*}
(Here, $x \leq y$ denotes that there is a path from $x$ to $y$
and $x < y$ stands for $x \leq y \wedge x \neq y$.
Moreover, $l_a(z)$ denotes that $z$ is labelled $a$.) Thus, the first transformation gives a $\Gamma$-algebraic term tree,
since algebraic term trees (and in fact deterministic algebraic trees
in the Caucal's pushdown hierarchy) are closed under monadic
colourings~\cite[Proposition~1]{CarayolWohrle}.

In order to prove that the second transformation also gives an
algebraic term tree, we argue on the level of graphs.  Suppose that
$T_0'$ is the algebraic term tree obtained after the first step. Since
$T_0'$ is algebraic, it is the unfolding of a (deterministic) prefix
recognizable graph $G$ from its root $r$,
see~\cite{CarayolWohrle,Caucal}. Without loss of generality we may
assume that every vertex of $G$ is accessible from $r$. Our aim is to
define a monadic transduction which, when applied to $G$, produces a
graph $G'$ whose unfolding from vertex $r$ is $T_1$. Since, by
Proposition~1 in conjunction with Lemma~2 in~\cite{CarayolWohrle}, prefix
recognizable graphs are closed under monadic transductions, it
follows that $T_1$ is algebraic.

We start by considering $G$ together with a disjoint copy of $G$, 
whose vertices are ordered pairs $(v,1)$, where $v$ is 
vertex of $G$. The label of a vertex $(v,1)$ is the label
of $v$ in $G$. The edges are the edges of $G$ and an edge 
$v \to (v,1)$ labelled $\#$ for each vertex $v$ of $G$.
Edges in $G$ retain their label.

Then we drop all vertices of the form $(v,1)$, where the label of $v$ 
is different from $+$, using the formula
\[\exists y. E_{\#}(y,x) \wedge \neg l_+(x) \]
where the meaning of $E_{\#}(y,x)$ is that there is an edge 
labelled $\#$ from $y$ to $x$, 
which is satisfied by exactly those vertices $(v,1)$ labelled $+$.
Moreover, we define new edges. First of all, 
we keep all edges $v \to v'$ of $G$ such that $v$ is labelled in $A$,  
or $v$ is labelled $+$ but $v'$ is not. Each such edge retains its label. 
Second, whenever $v$ and $v'$ are both labelled $+$ in $G$ 
we introduce an edge $v \to (v',1)$ and an edge $(v,1) \to (v',1)$ 
whose label is the same as that of the edge $v \to v'$ in $G$. 
For this purpose, we use the formulas in the free variables $x,y$, 
\[l_+(x) \wedge l_+(y) \wedge \exists y'.(E_{\#}(y',y) \wedge 
(E_i(x,y') \vee \exists x'.(E_{\#}(x',x) \wedge E_i(x',y'))))\]
where $i = 1,2$. 
Last, for each edge $v \to v'$ of $G$ such that $v$ is labelled $+$ 
but $v'$ is not, we introduce an edge $(v,1) \to v'$ whose label 
is that of the edge $v \to v'$. This is done by utilizing the 
formula
\[l_+(x) \wedge \neg l_+(y) \wedge 
\exists x'. ( E_{\#}(x',x) \wedge E_i(x',y) )  \] 
where again $i = 1,2$. 

Note that the unfolding of the graph constructed above from the 
vertex $r$ is $T_0'$. Next, consider any vertex $v$  labelled $+$ together
with all the paths originating in $v$ leading to a vertex labelled
in $A \cup\{1\}$. Since $t'$ is bounded by $k$, each such 
path is simple and the number of such paths is at most $k$. 
Let $v_1,\ldots,v_n$ denote the (not necessarily different) 
end vertices of these paths, ordered `lexicographically'. 
We relabel $v$ by $+_n$ and introduce a new edge $v \to v_i$ labelled $i$ for each $i$.
This is accomplished by using the following formulas. 
Let $\Path(x,X,y)$ denote a formula that says that the set of vertices $X$ 
forms a path from $x$ to $y$, the label of each vertex in $X$ 
other than $y$ is different from $+$, and the label of $y$ is in $A \cup\{1\}$. 
Then for each $n$, the formula 
\[\exists X_1\ldots X_n,\exists x_1,\ldots,x_n .\Wedge_{i < j} \neg (X_i = X_j) 
\wedge \Wedge_i \Path(x,X_i,x_i)\] 
expresses that there are at least $n$ different paths from $x$ to some vertex $y$ 
labelled in $A \cup \{1\}$,  
all of whose vertices different from $y$ are labelled $+$. With the help 
of these formulas we can also express that there are exactly $n$ such
paths from $x$.  Finally, when $\Path(x,X,y)$ and $\Path(x,X',y')$ 
with $X \neq X'$, we can express the fact that $X$ is lexicographically 
less than $X'$ by the formula 
\[
\begin{array}{c}
\exists Y,Z,Z'. \exists z_0,z,z' 
(X = Y \cup Z \wedge X' = Y \cup Z' \wedge \Path(x,Y,z_0) \wedge \\
(z = y \vee \Path(z,Z,y)) \wedge (z' = y' \vee \Path(z',Z',y') )
\wedge E_0(z_0,z) \wedge E_1(z_0,z')
\end{array}\]

In addition to these new edges, we keep all edges 
originating from a vertex labelled in $A$ (that 
are necessarily labelled $1$). All vertices 
of the form $(v,1)$ become inaccessible from $r$.
The unfolding of the new graph from vertex $r$ 
is almost an alternating term. In order to make it alternating,
we have to add a new root labelled $+_1$ if $r$ 
is labelled in $A$ together with an edge to $r$, 
and replace each edge $v \to v'$
where both $v$ and $v'$ are labelled in $A$ by new edges 
$v \to u$ and $u \to v$, where $u$ is a new vertex labelled $+_1$.
These edges are labelled $1$. The new graph is still obtained by 
monadic transduction, and its unfolding is the alternating term $T_1$.

The same argument works in the regular case using the fact that 
regular terms are unfoldings of finite (deterministic) graphs,
and that finite graphs are closed under monadic transduction.
\end{proof}

{\sl Proof of Theorem~\ref{thm-DCFL}.}  Suppose that $t \in \Sta$ is
bounded. If $t$ is algebraic, then by Lemma~\ref{lem-B} there is some
determinization $t'$ of $t$ such that $T_{t'}$ is algebraic.  By
Courcelle's theorem, the branch language of $T_{t'}$ is a DCFL.  But
the branch language of $T_{t'}$ is essentially $L(t')$.

Suppose now that $t$ has a determinization $t'$ 
such that $L(t')$ is a DCFL. Then  the branch language of 
$T_{t'}$ is a DCFL, and thus by Courcelle's theorem, 
$T_{t'}$ is algebraic. The proof is completed by using 
Lemma~\ref{lem-A}. 

A similar reasoning applies in the regular case.

The language-theoretic characterization of the class of bounded
algebraic synchronization trees offered in Theorem~\ref{thm-DCFL} can
be used to prove that certain trees are {\em not} algebraic.

\begin{prop}\label{Prop:bagnotalg}
The synchronization tree $t_{\text{bag}}$ associated with the bag over
a binary alphabet depicted on Figure~\ref{fig:bag} is not algebraic,
even up to language equivalence.
\end{prop}
\begin{proof}
Recall that the {branch words} corresponding to the nodes
of any determinization of the tree $t_{\text{bag}}$ have the form
\[
k_0(a_1,i_1)k_1\ldots k_{m-1}(a_m,i_m)k_m ,
\]
where $k_0=3$, $k_m\in \{3,4,0\}$, $k_i \in \{3,4\}$ for each $i \in [m-1]$, each $i_j \in [k_{j-1}]$ for $j \in [m]$
and $a_1\ldots a_m$ is a word with
the property that, in any of its prefixes, the number of occurrences
of the letter $a$ is greater than or equal to the number of
occurrences of the letter $b$, and the number of occurrences of the
letter $c$ is greater than or equal to the number of occurrences of
the letter $d$. Moreover, $a_j =\ex$ if and only if $j = m$ 
and $k_m = 0$. The words accepted by that LTS are those that in
addition satisfy that the total number of occurrences of the letter
$a$ in $a_1\ldots a_m$ is equal to the number of occurrences of the
letter $b$, and the number of occurrences of the letter $c$ in
$a_1\ldots a_m$ is equal to the number of occurrences of the letter
$d$ and which end in $0$.

If the language associated with any determinization of
$t_{\text{bag}}$ were context-free, then so would the language
obtained by applying to each word in it the morphism that erases the
letters $k_j$ and renames each $(a_j,i_j)$ to $a_j$. However, that
language is not context free. Therefore, Theorem~\ref{thm-DCFL} yields
that $t_{\text{bag}}$ is not algebraic.   
\end{proof}

The above proposition is a strengthening of a classic result from the
literature on process algebra proved by Bergstra and Klop
in~\cite{BergstraK84}. Indeed, in Theorem~4.1 in~\cite{BergstraK84},
Bergstra and Klop showed that the bag over a domain of values that
contains at least two elements is not expressible in BPA. Moreover, by
Proposition~\ref{Prop:moreexpthanBPA}, the collection of
synchronization trees that are definable in BPA is strictly included
in the set of $\Gamma$-algebraic ones
(Theorem~\ref{Thm:Delta2Gamma}). Therefore,
Proposition~\ref{Prop:bagnotalg} is stronger than the above-mentioned
inexpressibility result by Bergstra and Klop, and offers an
alternative proof for it. Up to bisimilarity, we shall offer an
even stronger statement in Section~\ref{Sect:MSO-bag} (see Proposition~\ref{prop:bag-bisim}).

As a final result, we can use Theorem~\ref{thm-DCFL} to characterize the 
synchronization trees of bounded degree in $\Tree_{2}$.

\begin{cor}
The  synchronization trees of bounded-degree in $\text{Tree}_{2}$ are the 
$\Gamma$-algebraic synchronization trees of bounded-degree. 
\end{cor}

\begin{proof}
The direct inclusion is immediate. For the converse inclusion, we know by
Theorem~\ref{thm-DCFL} that a $\Gamma$-algebraic synchronization tree of
bounded-degree $t$ has a determinization $t'$ whose branch language is a 
deterministic context-free language. In particular, this implies that $t'$
belongs to $\Tree_{2}$. Indeed from a deterministic pushdown automaton accepting the branch language of $t'$, we can construct a deterministic graph $G \in \Graph_{1}$ whose unfolding 
from some root $r$ is isomorphic to $t'$. Let $A=\{x_{1},\ldots,x_{k}\}$ be
the alphabet labelling $G$. Consider the transduction $\trans$ defined as follows:
\begin{itemize}
\item for every vertex $u$, the transduction introduces new vertices $u_{1},\ldots,u_{k}$;
\item for every edge from $u$ to $v$ labelled by $x_{\ell}=(a,i)$, the transduction
adds edges from each of the $u_{j}$, $j \in [k]$, to $v_{\ell}$ labelled by $a$.
\end{itemize}
Unfolding the graph $\trans(G)$ from any of the vertices added by the transduction $\trans$ for the root $r$ gives
a tree isomorphic to $t$.
\end{proof}

Note that in the previous proof, $\trans$ cannot simply rename the edges labelled $(a,i)$ to $a$.
Consider for instance the case where $G$ consists of two vertices $r$ and $s$
and two edges $r \era{(a,1)} s$ and $r \era{(a,2)} s$. Applying such a transduction would lead to a graph with only one edge.

\section{Synchronization trees and logic}\label{Sect:treesandlogic}

In this section, we investigate the consequences of the decidability of
monadic second-order logic for our synchronization trees. 

\subsection{A synchronization tree with an undecidable monadic theory}
\label{Sect:MSO-bag}

In this subsection we point out that the synchronization tree
associated with the bag process depicted on Figure~\ref{fig:bag} has
an undecidable monadic theory (even without the root being the source
of an exit edge). Hence that tree is not in the Caucal hierarchy and
therefore, by Proposition~\ref{Prop:Delta-alg-in-Graph3}, is not
$\Delta$-algebraic not even up to bisimilarity (Proposition~\ref{prop:bag-bisim}). {A  similar result was obtained in \cite[Section 6.6.2]{Colcombetthesis} for a slightly richer structure. For completeness sake, we give below a detailed proof of this result.}

\begin{prop}
The synchronization tree $t_{bag}$ associated with  the bag process has an undecidable MSO-theory.
\end{prop}

\begin{proof}
Consider a $2$-counter machine whose program $P$ is given as a sequence of 
instructions $I_1,\ldots,I_n$ where each $I_j$ has one of the 
following forms: 
  \[z:= z+1,k\quad z:= z-1,k,\quad z=0?,k_1,k_2\]
where $z$ is one of the counters $x,y$ and $k,k_1,k_2$ are integers between 
$1$ and $n$. At any moment of time, the value of the counters $x$ and $y$ 
is described by an ordered pair $(m,n)$ of non-negative integers. 
The meaning of the above instructions
is standard, where $k$ denotes the index of the instruction to be performed 
after execution of the given instruction, if the instruction 
is an increment or decrement, and $k_1$ and $k_2$ denote 
the indices of the instructions 
to be performed depending on the outcome of the test, if the 
instruction is of the last form.  The machine with program $P$ halts if 
a decrement instruction  $z:= z-1,k$ is executed, but the current value 
of the counter $z$ is $0$. 
A well-known undecidable question for $2$-counter machines asks whether 
a $2$-counter machine started with $(0,0)$ ever halts. 

We encode the halting program for a counter machine with program $P$
in monadic second order logic as follows. Let $X_1,\ldots,X_n$ be a set of
variables associated with the instructions of $P$. Then, for each
instruction $I_i$, we consider the formula $\varphi_i(u)$ in the free
first-order variable $u$ in the language with four binary predicates
associated with the edge relations $\stackrel{e}{\to}$, $e \in
\{a,b,c,d\}$.
\begin{itemize}
\item
  If $I_i$ is of the form $z:= z+1,k$, then $\varphi_i(u)$ expresses that 
  $X_k(v)$ holds for all $v$ such that $u \stackrel{e}{\to} v$, 
  where $e$ is $a$ if $z= x$ and $c$ if $z = y$.
\item 
   If $I_i$ is of the form $z:= z-1,k$, then $\varphi_i(u)$ expresses that 
   there exists a $v$ with $u \stackrel{e}{\to} v$, and $X_k(v)$ holds 
    for all such $v$, where $e$ is $b$ if $z= x$ and $d$ if $z = y$.
\item 
   If $I_i$ is of the form $z = 0?,k_1,k_2$, then $\varphi_i(u)$ expresses that 
   $X_i(v_1)$ and $X_{k_1}(v_2)$ hold for all $v_1,v_2$ such that 
   $u \stackrel{a}{\to} v_1$ and $v_1 \stackrel{b}{\to} v_2$, provided that 
   there is no $v$ with $u\stackrel{e}{\to} v$, and that 
   $X_i(v_1)$ and $X_{k_2}(v_2)$ hold for all such $v_1,v_2$ otherwise,
   where again $e$ is $b$ if $z= x$ and $d$ if $z = y$.
\end{itemize}
Now we assign to the machine with program $P$ the formula 
\[\varphi_P =\exists X_1\ldots\exists X_n[\psi_1 \wedge \psi_2 \wedge
  \forall u (X_1(u) \to \varphi_1(u)\wedge \ldots\wedge X_n(u) \to
  \varphi_n(u))]\] where $\psi_1$ asserts that the $X_1,\ldots ,X_n$
are pairwise disjoint and jointly form an infinite path starting with
the root, and $\psi_2$ says that the root belongs to $X_1$.  Then the
machine does not halt iff the synchronization tree defined by the
process on Figure~\ref{fig:bag} is a model of $\varphi_P$.
  \end{proof}

\begin{rem}
Thomas showed in~\cite[Theorem~10]{Thomas01} that the monadic
 second-order theory of the infinite two-dimensional grid is
 undecidable. However, we cannot use that result to prove that the
 synchronization tree $t_{\text{bag}}$ has an undecidable monadic
 second-order theory. Indeed, the unfolding of the infinite
 two-dimensional grid is the full binary tree, which has a decidable
 monadic second-order theory by Rabin's celebrated Tree
 Theorem~\cite{Rabin1969}.
\end{rem}

\begin{prop}
\label{prop:bag-bisim}
The synchronization tree $t_{\text{bag}}$ is not $\Delta$-algebraic
up to bisimilarity.
\end{prop}

\begin{proof}
This statement follows from the more general remark that any synchronization tree $t$
that is bisimilar to a deterministic synchronization tree $t_{0}$ having an undecidable MSO-theory also has an undecidable MSO-theory. Given a formula $\varphi$ with no free variables, consider the formula
\[
 \varphi^{*}=\exists X\; \varphi_{\textrm{det}}(X) \wedge \varphi'
\]
where $\varphi'$ is the formula $\varphi$ in which all quantifications are relativized to $X$ and $\varphi_{\textrm{det}}(X)$ states that if a node 
has a successor by an $a$-labelled edge then it has one and only one successor 
by an $a$-labelled edge which belongs to $X$. If the formula $\varphi_{\textrm{det}}(X)$ is satisfied on $t$ for some set of vertices $U$ then
$t$ restricted to the vertices in $U$ is a deterministic tree isomorphic to $t_{0}$ (cf. Lemma~\ref{lem:bisim-det-iso} on page~\pageref{lem:bisim-det-iso}). Clearly $\varphi^{*}$ holds on $t$ if and only if $\varphi$ holds on $t_{0}$. This implies that the MSO-theory of $t$ is undecidable.
  \end{proof}

{The argument used in this proof can also be used to show that $t_{bag}$ does not belong to the Caucal hierarchy up to bisimilarity.}

\subsection{Minimization}

\newcommand{\normm}[1]{|\!|\,u\,|\!|}

It is well known that, for each bisimulation equivalence class
$\mathcal{C}$ of synchronization trees in $\Sta$, there is a tree
$t_\mathcal{C} \in \mathcal{C}$ such that for all $t\in \mathcal{C}$
there is a \emph{surjective} morphism $\varphi: t \to t_\mathcal{C}$,
and, moreover, the relation 
\[
R = \varphi \cup \varphi^{-1} =
\{(u,v),(v,u) : \varphi(u) = v\}
\] 
is a bisimulation between $t$ and
$t_\mathcal{C}$.  Furthermore, $t_\mathcal{C}$ is unique up to
isomorphism. When $t \in \mathcal{C}$, we call $t_\mathcal{C}$ the
\emph{minimization} of $t$.

The minimization of a tree $t\in \Sta$ can be constructed in the
following way. We define an increasing sequence $V_0,V_1,\ldots$ of
sets of vertices of $t$, where $V_0$ is a singleton set containing
only the root of $t$. The construction will guarantee that, for each
$i$, the set $V_i$ contains only vertices of depth at most $i$, and whenever a
vertex $v$ belongs to $V_i$ and there is a path from a vertex $u$ to
$v$, then $u$ is also in $V_i$.  Given $V_i$ and $u \in V_i$ of depth
$i$, consider the set $S(u)$ of successors of $u$.  We may divide
$S(u)$ into equivalence classes according to the bisimulation
equivalence classes of the corresponding subtrees {and the label of the edge coming from $u$}. To this end, we
define $v \sim v'$, for $v,v'\in S(u)$, if the subtrees rooted at $v$
and $v'$ are bisimilar {and if $u \era{a} v$ and $u \era{a} v'$ for some $a \in A$}. Then we select a representative of each
$\sim$-equivalence class.  The set $V_{i+1}$ consists of all vertices
in $V_i$ together with those vertices in $S(u)$ of depth $i+1$
selected above, where $u$ ranges over the set of all vertices in $V_i$
of depth $i$. Finally, let $V = \bigcup_{i \geq0}V_i$.  The `subtree"
of $t$ spanned by the vertices in $V$ is the minimization of $t$.

It is known, see e.g. \cite{BEbook}, that the minimization of a 
$\Gamma$-regular synchronization tree is $\Gamma$-regular. 
In contrast with this result, we have: 

\begin{prop}\label{Prop:QuotientnotCaucal}
There exists a $\Gamma$-algebraic synchronization tree whose
minimization does not have a decidable MSO-theory,  and hence 
does not belong to the Caucal hierarchy and is neither a $\Gamma$-algebraic 
nor a $\Delta$-algebraic synchronization tree.
\end{prop}

\begin{proof}
Let $A=\{a,b,c,d,e,f\}$. Consider the following $\Gamma$-algebraic scheme:
\begin{eqnarray*}
  S      & = & F(0,0) \\
  F(x,y) & = & a.F(f.x,y) + b.F(x,f.y) + c.F(0,y) + d.F(x,0) + e.x + e.y\ . 
\end{eqnarray*}
Let $t$ be the synchronization tree defined the above scheme.  For a
word $u \in \{a,b,c,d \}^{*}$, we designate by $\normm{u}_{a}$
(resp. $\normm{u}_{b}$) the number of $a$'s (resp. $b$'s) in the
longest suffix of $u$ that does \emph{not} contain any occurrence of
the letter $c$ (resp. $d$).  Intuitively, the tree $t$ consists of a
full quaternary deterministic tree $t$ with edges labelled in
$\{a,b,c,d\}$ such that every vertex of $t$ is also the origin of
two additional parallel disjoint branches. Since $t$ is
deterministic, we may identify each vertex of $t$ with a word $u \in
\{a,b,c,d\}^*$.  The two additional branches at vertex $u$ of $t$
are such that their edge labels form the words $ef^{\normm{u}_{a}}$
and $e f^{\normm{u}_{b}}$, respectively.

The minimization $t'$ of $t$ is obtained by removing one of the 
two branches labelled 
$ef^{\normm{u}_{a}}$ for all vertices $u$ of $t$ such 
that $\normm{u}_{a}=\normm{u}_{b}$.

The fact that $t'$ has an undecidable MSO-theory is based on a
reduction from the halting problem for 2-counter machines with
increment, reset and equality test \cite{Wojna}. The idea is, as usual,
to define, for every such machine $M$, a closed MSO-formula
$\varphi_{M}$ such that $t' \models \varphi_{M}$ if and only if $M$
does not halt.

When constructing the formula $\varphi_M$, we associate to a vertex
$u$ in $\{a,b,c,d\}^{*}$ representing a possible history of the
computation of $M$ from its initial state, the integer $\normm{u}_{a}$
as the current value of the first counter, and the integer
$\normm{u}_{b}$ as the current value of the second counter of the
machine.  Assuming that the current values of the counters are given
by vertex $u$, we show how to simulate the various operations of the
machine. The increment of the first (resp. second) counter is obtained
by moving to vertex $ua$ (resp. $ub$).  The reset of the first
(resp. second) counter is obtained by moving to $uc$
(resp. $ud$). Finally, the test for equality between the two counters
is performed by testing that vertex $u$ has only one outgoing edge
labelled $e$. The details of the construction are similar to those in
Section~\ref{Sect:MSO-bag}.    
\end{proof}

The above result yields that the collection of
synchronization trees in the Caucal hierarchy is {\em not} closed
under minimization. Indeed, there is a
$\Gamma$-algebraic tree whose minimization is
not in the Caucal hierarchy.
This leaves open the corresponding question for $\Delta$-regular
synchronization trees.

\section{Open questions}\label{Sect:concl}

There are several questions that we leave open in this paper and that
lead to interesting directions for future research.

\begin{itemize}
\item Is there a strict expressiveness hierarchy for $\Gamma$- and
  $\Delta$-algebraic recursion schemes that is induced by the maximum
  rank of the functor variables used in defining recursion schemes?
\item Is the minimal synchronization tree that is bisimilar to a
$\Delta$-regular synchronization tree also $\Delta$-regular? If not,
is it in the Caucal hierarchy? 
\item The $\Gamma$-algebraic scheme we use in the proof of
Proposition~\ref{Prop:QuotientnotCaucal} uses a binary functor
variable. Is the minimal synchronization tree that is bisimilar to a
tree defined by a $\Gamma$-algebraic scheme involving only unary
functor variables $\Gamma$-algebraic?

\end{itemize}


\thebibliography{nn}

\bibitem{Abramsky91}
S.~Abramsky.
\newblock A domain equation for bisimulation. 
\newblock {\em Inform. and Comput.},  92 (1991),  no. 2, 161--218.

\bibitem{Aho68}
A.V. Aho. 
\newblock Indexed grammars --- an extension of context-free grammars. 
\newblock \textit{Journal of the ACM} 15 (1968), 647--671.

\bibitem{BaetenBR2009}
J.C.M. Baeten, T.~Basten, and M.A. Reniers.
\newblock {\em Process Algebra: Equational Theories of Communicating
  Processes}, volume~50 of {\em Cambridge Tracts in Theoretical Computer
  Science}.
\newblock Cambridge University Press, November 2009.

\bibitem{Baetenetal}
J.C.M. Baeten and W.P. Weijland.
\newblock \emph{Process Algebra}. 
\newblock Cambridge Tracts in Theoretical Computer Science, 
18. Cambridge University Press, Cambridge, 1990. 

\bibitem{deBakker}
J.W. de Bakker. 
\newblock \emph{Recursive Procedures}. 
\newblock Mathematical Centre Tracts, No. 24.
Mathematisch Centrum, Amsterdam, iv+108 pp., 1971. 

\bibitem{BergstraK84}
J.A. Bergstra and J.W. Klop. 
\newblock The algebra of recursively defined processes and the algebra
               of regular processes. 
\newblock Proceedings of ICALP 1984, LNCS 172, pp.~82--94, 
Springer, 1984.

\bibitem{BET}
S.L. Bloom, Z. \'Esik and D. Taubner.
\newblock Iteration theories of synchronization trees.  
\newblock {\em Inform. and Comput.},  102 (1993),  no. 1, 1--55.

\bibitem{BEbook}
S.L. Bloom and Z. \'Esik.
\newblock {\em Iteration Theories.} 
Springer, 1993. 


\bibitem{BEreg}
S.L. Bloom and Z. \'Esik.
\newblock The equational theory of regular words.
\newblock \textit{Information and Computation}, 197 (2005),  55--89.

\bibitem{BEmezei}
S.L. Bloom and Z. \'Esik.
\newblock A Mezei-Wright theorem for categorical algebras.
\newblock \textit{Theoretical Computer Science}, 411 (2010), 341--359.

\bibitem{BloomTWW83} S.L. Bloom, J. W. Thatcher, E. G. Wagner and
  J. B. Wright.  \newblock Recursion and iteration in continuous
  theories: The ``$M$-construction''.  \newblock
  \textit{J. Comput. System Sci.}, 27 (1983), 148--164.
  
\bibitem{Blumensath}
A. Blumensath.
\newblock Prefix recognizable graphs and monadic second order logic.
\newblock Technical Report AIB-06-2001, RWTH Aachen, 2001.

\bibitem{Braudthesis}
L. Braud.
\newblock The structure of linear orders in the pushdown hierarchy.
\newblock PhD thesis, Institut Pascal Monge, 2010.

\bibitem{vBreugel2001} F.~van Breugel.  \newblock An introduction to
metric semantics: Operational and denotational models for programming
and specification languages.  \newblock \textit{Theoretical Computer
Science}, 258(2001), 1--98.

\bibitem{CarayolWohrle}
A. Carayol and S. W\"ohrle. 
\newblock The Caucal hierarchy of infinite graphs in terms
of logic and higher-order pushdown automata.
\newblock Proceedings of FSTTCS 03, LNCS 2914, pp.~112--123, 
Springer, 2003.

\bibitem{Carayol05} Arnaud Carayol \newblock Regular sets of higher-order pushdown stacks. \newblock  Proceedings of MFCS 2005, LNCS 3618, pp.~168--179, Springer, 2005.

\bibitem{Carayolthesis}
A. Carayol.
\newblock Automates infinis, logique et langages.
\newblock PhD thesis, Universit{\'e} Rennes I, 2006.


\bibitem{Caucal}
D. Caucal. 
\newblock On infinite terms having a decidable monadic theory.
\newblock Proceedings of MFCS 02, LNCS 2420, 165--176, Springer, 2002.

\bibitem{Caucal03}
D. Caucal. 
\newblock On infinite transition graphs having a decidable monadic 
theory.
\newblock \textit{Theoretical Computer Science} 290(2003), 79--115.

\bibitem{Christensen83} S. Christensen.  \newblock Decidability and
  decomposition in process algebras.  \newblock PhD thesis
  ECS-LFCS-93-278, Department of Computer Science, University of
  Edinburgh, 1983.

\bibitem{Colcombetthesis}
T. Colcombet.
\newblock Propri{\'e}t{\'e}s et repr{\'e}sentation de structures infinies.
\newblock PhD thesis, Universit{\'e} Rennes I, March 2004.


\bibitem{Courcelle78a}
B. Courcelle.
\newblock 
A representation of trees by languages {I}. 
\newblock \textit{Theoretical Computer Science} 6 (1978), 255--279.

\bibitem{Courcelle78b}
B. Courcelle.
\newblock 
A representation of trees by languages {II}. 
\newblock \textit{Theoretical Computer Science} 7 (1978), 25--55.

\bibitem{Courcelle83}
B. Courcelle.
\newblock 
Fundamental properties of infinite trees, 
\newblock \textit{Theoretical Computer Science} 25 (1983), 69--95.

\bibitem{Courcelle94}
B. Courcelle. 
Monadic second-order definable graph transductions: A survey.
{\em Theoretical Computer Science}, 126 (1994), 53--75.

\bibitem{CourcelleN76}
B. Courcelle and M. Nivat. 
\newblock Algebraic families of interpretations.
\newblock In {\em 17th Annual Symposium on Foundations of Computer Science},
IEEE Computer Society, 1976, 137--146. 

\bibitem{Ebbinghaus95}
H.D. Ebbinghaus and J.~Flum.
\newblock {\em Finite Model Theory}.
\newblock Springer-Verlag, 1995.

\bibitem{Esiksynch} Z. \'Esik.  
\newblock Continuous additive algebras
  and injective simulations of synchronization trees.  
\newblock Fixed
  Points in Computer Science, 2000 (Paris).  
\newblock
  \textit{Journal of Logic and Computation}, 12 (2002), 271--300.

\bibitem{Fischer}
M.J. Fischer,
Grammars with macro-like productions, 
In {9th Annual Symp. Switching and Automata Theory, Schenedtady, NY, USA, 1968},
IEEE Press, 1968, 131--142. 

\bibitem{GoguenTWW77}
J.A. Goguen, J.W. Thatcher, E.G. Wagner and J.B. Wright. 
\newblock Initial algebra semantics and continuous algebras. 
\newblock \textit{Journal of the ACM} 24 (1977), 68--95.

\bibitem{Gue81}
I. Guessarian, 
\newblock \emph{Algebraic Semantics}. 
\newblock LNCS  99, Springer, 1981. 

\bibitem{MiliusM2006}
S. Milius and L. Moss. 
\newblock The category-theoretic solution of recursive program schemes.
\newblock \textit{Theoretical Computer Science} 366 (2006), 3--59.

\bibitem{MandrioliG88}
D. Mandrioli and C. Ghezzi.
\newblock \textit{Theoretical foundations of Computer Science.}
\newblock John Wiley and Sons, 1988.

\bibitem{Mi71}
R. Milner.
\newblock An algebraic definition of simulation between programs.
\newblock In {Proceedings 2nd Joint Conference on Artificial Intelligence},
  pages 481--489. BCS, 1971.
\newblock Also available as Report No.\ CS-205, Computer Science Department,
  Stanford University.

\bibitem{Milner}
R. Milner. 
\newblock \emph{A Calculus of Communicating Systems.} 
\newblock LNCS  92, Springer, 1980. 

\bibitem{Mil89CC}
R. Milner.
\newblock {\em Communication and Concurrency}.
\newblock Prentice Hall, 1989.

\bibitem{Moller96} F. Moller.  
\newblock Infinite results.  
\newblock Proceedings of CONCUR '96, Concurrency Theory, 7th International
  Conference, LNCS 1119, pp.~195--216, Springer, 1986

\bibitem{Nivat75}
M. Nivat, 
\newblock  On the interpretation of recursive polyadic program schemes. 
\newblock \emph{Symposia Mathematica} XV (1975), 255 -281.

\bibitem{Par81}
D.M.R. Park.
\newblock Concurrency and automata on infinite sequences.
\newblock In {Theoretical Computer Science, 5th GI-Conference}, LNCS 104, pp.~167--183, Springer, 1981.

\bibitem{Rabin1969}
M.O. Rabin.
\newblock Decidability of second-order theories and automata on infinite trees.
\newblock \emph{Transactions of the American Mathematical Society} 141 (1969), 
1--35.

\bibitem{Scott71}
D.S. Scott. 
\newblock The lattice of flow diagrams. 
\newblock In {Symposium on Semantics of
Algorithmic Languages 1971}, Lecture Notes in Mathematics, vol. 188,
Springer, 1971, pp.~311--366.

\bibitem{Stirling00}
C.~Stirling.
\newblock Decidability of bisimulation equivalence for pushdown processes.
\newblock Technical Report EDI-INF-RR-0005, School of Informatics, University
  of Edinburgh, 2000.


\bibitem{Thomas01} W. Thomas.  \newblock A short introduction to
infinite automata.  \newblock In {Developments in Language
Theory}, LNCS 2295, pp.~130--144, Springer, 2001.

\bibitem{Thomas03}
W. Thomas.
\newblock Constructing infinite graphs with a decidable MSO-theory. 
\newblock In {Proceedings of the 28th International Symposium on Mathematical Foundations of Computer Science},  LNCS 2747, pp.~113--124, Springer, 2003.

\bibitem{Winskel}
G. Winskel. 
\newblock Synchronization trees.  
\newblock \textit{Theoretical Computer Science}   34 (1984),  no. 1--2, 33--82. 

\bibitem{Wojna}
A. Wojna.
\newblock Counter machines.
\newblock \emph{Information Processing Letters} 71 (1991), 193--197.
\end{document}